\newif\ifarxiv\arxivtrue 
\newif\ifsubmit\submittrue \submitfalse 

\newif\ifchange\changetrue \changefalse 
\ifarxiv
\documentclass[acmsmall,screen,nonacm]{acmart} \else
\ifsubmit
\documentclass[acmsmall,screen,review,anonymous,nonacm]{acmart}
\else
\documentclass[acmsmall,screen,review,nonacm]{acmart}
\fi
\fi

\usepackage{microtype}
\usepackage{graphicx}
\usepackage{longtable}
\usepackage{hyperref}
\usepackage{color}
\usepackage{mathpartir}
\usepackage{stackrel}
\usepackage{stmaryrd}
\usepackage{wrapfig,subcaption}
\usepackage{relsize}
\usepackage{thmtools}
\usepackage{thm-restate}
\usepackage{whyrellang}
\usepackage{tikz}
\usepackage[framemethod=TikZ]{mdframed}
\usetikzlibrary{cd,automata,positioning,arrows,arrows.meta,decorations.pathmorphing}

\usepackage[normalem]{ulem} 
\setcopyright{rightsretained}
\acmPrice{}
\acmDOI{10.1145/xxxxxxxxx}
\acmYear{2026}
\copyrightyear{2026}
\acmJournal{PACMPL}
\acmVolume{X}
\acmNumber{POPL}
\acmMonth{1}

\ifsubmit
\renewcommand\footnotetextcopyrightpermission[1]{} 

\fi

\ifarxiv
\renewcommand\footnotetextcopyrightpermission[1]{} 

\fi

\begin{document}

\newtheorem{remark}[theorem]{Remark} 
\ifchange
\title{Forall-Exists Relational Verification by Filtering to Forall-Forall
{\small(diff change marks: \textcolor{red}{\sout{red
               deleted}} \textcolor{bleudefrance}{blue added})} }\else
\title{Forall-Exists Relational Verification by Filtering to Forall-Forall}
\fi

\ifsubmit
\else
\author{Ramana Nagasamudram}
\affiliation{\institution{Amazon Web Services}\country{USA}}
\authornote{Nagasamudram's work was done prior to his joining Amazon Web Services.}
\author{Anindya Banerjee}
\authornote{Banerjee’s research was based on work supported by the NSF, while working at the Foundation. Any opinions, findings, and conclusions or recommendations expressed in this article are those of the authors and do not necessarily reflect the views of the NSF.}
\affiliation{\institution{Dartmouth College}\country{USA}}
\author{Jude Kanjamala}
\affiliation{\institution{Stevens Institute of Technology}\country{USA}}
\author{David A. Naumann}
\affiliation{\institution{Stevens Institute of Technology}\country{USA}}
\authornote{Kanjamala and Naumann were partially supported by NSF grant CNS 2426414.}
\fi

\newcommand{\WhyRelE}{WhyRel-E} 
\definecolor{bleudefrance}{rgb}{0.19, 0.55, 0.91} 

\ifchange
\newcommand\changed[2]{\textcolor{red}{\sout{#1}}\textcolor{bleudefrance}{#2}\marginpar{$\checkmark$}}
\newcommand\Xchanged[2]{\textcolor{red}{#1}\textcolor{bleudefrance}{#2}\marginpar{$\checkmark$}}
\else 
\newcommand\changed[2]{#2}
\newcommand\Xchanged[2]{#2}
\fi

\newcommand{\dt}[1]{\textbf{\emph{#1}}} \newcommand{\rn}[1]{\textsc{\relsize{-1}#1}} 
\renewcommand{\iff}{\Leftrightarrow} 
\definecolor{light-gray}{gray}{0.88}
\definecolor{dark-gray}{gray}{0.25}
\newcommand{\graybox}[1]{\colorbox{light-gray}{#1}} 
\newcommand{\keyw}[1]{\ensuremath{\mathsf{#1}}}  
\newcommand{\size}{\keyw{size}}
\newcommand{\wlp}{\keyw{wlp}}
\newcommand{\wlpL}{\keyw{wlpL}}
\newcommand{\wlpR}{\keyw{wlpR}}
\newcommand{\gfp}{\keyw{gfp}}

\newcommand{\semframe}[2]{#2\mathbin{\Vdash} #1} \newcommand{\comFrame}{\keyw{cFrame}} \newcommand{\biFrame}{\keyw{bFrame}} \newcommand{\comFrameGhost}{\keyw{cFrameGhost}} 
\newcommand{\eraseGhost}{\keyw{eraseGhost}} 

\newcommand{\missingArg}{\keyw{\_}}
\newcommand{\mystrut}{\vrule height 3pt depth 2pt width 0pt} 
\newcommand{\skipc}{\keyw{skip}}
\newcommand{\ifc}[3]{\keyw{if}\ {#1}\ \keyw{then}\ {#2}\ \keyw{else}\ {#3}}
\newcommand{\ifFourXlong}[5]{\keyw{if}\ {#1}\ \keyw{thenThen}\ {#2}\ \keyw{thenElse}\ {#3}\ \keyw{elseThen}\ {#4}\ \keyw{elseElse}\ {#5}\ \keyw{fi}}
\newcommand{\ifFourLong}[5]{\keyw{if}\ {#1}\ \keyw{thth}\ {#2}\ \keyw{thel}\ {#3}\ \keyw{elth}\ {#4}\ \keyw{elel}\ {#5}\ \keyw{fi}}
\newcommand{\ifFour}[5]{
  \keyw{if}\ {#1}\ {#2}\ {#3}\ {#4}\ {#5}}

\newcommand{\whilec}[2]{\keyw{while}\ {#1}\ \keyw{do}\ {#2}}
\newcommand{\whilev}[3]{\keyw{while}\ {#1}\ \keyw{vnt}\ {#2}\ \keyw{do}\ {#3}}

\newcommand{\havc}[1]{\keyw{hav}\ {#1}} 
\newcommand{\assertc}[1]{\keyw{assert}\ {#1}} 
\newcommand{\assumec}[1]{\keyw{assume}\ {#1}} 
\newcommand{\biwhilev}[4]{\keyw{while}\ {#1}\ \keyw{algn}\ {#2}\ \keyw{vnt}\ {#3}\ \keyw{do}\ {#4}}
\newcommand{\biwhile}[3]{\keyw{while}\ {#1}\ \keyw{algn}\ {#2}\ \keyw{do}\ {#3}}

\newcommand{\havRtKeyword}{\keyw{havf}}
\newcommand{\havRt}[2]{\havRtKeyword\ {#1}\ {#2}} 

\newcommand{\biEmb}[2]{\langle #1 \sep #2 \rangle}
\newcommand{\BiEmb}[2]{\langle #1 \Sep #2 \rangle}
\newcommand{\bisync}[1]{\lfloor #1 \rloor}

\newcommand{\bigmid}{\mathrel{\mbox{\large$|$}}} 
\newcommand{\kateq}{\mathrel{\simeq}} \newcommand{\bieq}{\mathrel{\cong}} \newcommand{\agree}[1]{\stackrel{#1}{=}} 
\newcommand{\False}{\mathrm{ff}}
\newcommand{\True}{\mathrm{tt}}
\newcommand{\agreeRel}{\mathbin{\ddot{=}}} 

\newcommand{\sep}{\mathbin{\mid}} \newcommand{\Sep}{\mathbin{\:\mid\:}} \newcommand{\smSep}{\mbox{\tiny{$\pmb{|}$}}} 
\newcommand{\imp}{\Rightarrow}
\newcommand{\impby}{\Leftarrow}

\newcommand{\leftex}[1]{ \raisebox{.25ex}{\relsize{-1}$\langle\hspace*{-2.0pt}[$} #1 \raisebox{.25ex}{\relsize{-1}$\langle\hspace*{-2.4pt}]$} }
\newcommand{\rightex}[1]{ \raisebox{.25ex}{\relsize{-1}$[\hspace*{-2.5pt}\rangle$} #1 \raisebox{.25ex}{\relsize{-1}$]\hspace*{-2.3pt}\rangle$} }

\newcommand{\leftF}[1]{\leftex{#1}}
\newcommand{\rightF}[1]{\rightex{#1}}
\newcommand{\bothF}[1]{  \raisebox{.25ex}{\relsize{-1}$\langle\hspace*{-2.1pt}[$}
  #1 
  \raisebox{.25ex}{\relsize{-1}$]\hspace*{-2.1pt}\rangle$ } 
}

\renewcommand{\P}{\mathcal{P}}  \renewcommand{\S}{\mathcal{S}} \newcommand{\Lrel}{\mathcal{L}} 
\newcommand{\Q}{\mathcal{Q}} 
\newcommand{\J}{\mathcal{J}} 
\newcommand{\R}{\mathcal{R}} 
\newcommand{\I}{\mathcal{I}} 
\newcommand{\T}{\mathcal{T}} 
\newcommand{\X}{\mathcal{X}} 

\newcommand{\union}{\cup}
\newcommand{\intersect}{\cap}
\newcommand{\proves}{\vdash}
\newcommand{\aftqua}{.\:}
\newcommand{\all}[2]{\forall #1 \aftqua #2}
\newcommand{\some}[2]{\exists #1 \aftqua #2}
\newcommand{\subst}[3]{{#1}^{#2}_{#3}}
\newcommand{\catenate}{\mbox{+\!+}}

\newcommand{\quant}[4]{(#1\,#2 \::\: #3 \::\: #4)} \newcommand{\hint}[1]{{\qquad\mbox{\guillemotleft\, #1 \guillemotright}}} 
\newcommand{\mathconst}[1]{\mbox{\upshape\textsf{#1}}} \newcommand{\indep}{\mathconst{indep}} 
\newcommand{\ghost}{\mathconst{ghost}}
\newcommand{\erase}{\mathconst{erase}}
\newcommand{\uChk}{\mathconst{uchk}}
\newcommand{\chk}{\mathconst{chk}}
\newcommand{\Store}{\mathconst{Store}}
\newcommand{\enc}{\mathconst{enc}}
\newcommand{\dec}{\mathconst{dec}}

\newcommand{\modVars}{\mathconst{modVars}} \newcommand{\modVarsR}{\mathconst{modVarsR}} 
\def\rightharpoonupfill{$\mathsurround=0pt \mathord- \mkern-6mu
  \cleaders\hbox{$\mkern-2mu \mathord- \mkern-2mu$}\hfill
  \mkern-6mu \mathord\rightharpoonup$}

\def\leftharpoonupfill{$\mathsurround=0pt \mathord\leftharpoonup \mkern-6mu
  \cleaders\hbox{$\mkern-2mu \mathord- \mkern-2mu$}\hfill
  \mkern-6mu \mathord-$}
\def\overleftharpoonup#1{\vbox{\ialign{##\crcr
      \leftharpoonupfill\crcr\noalign{\kern-1pt\nointerlineskip}
      $\hfil\displaystyle{#1}\hfil$\crcr}}}

\def\overrightharpoonup#1{\vbox{\ialign{##\crcr
      \rightharpoonupfill\crcr\noalign{\kern-1pt\nointerlineskip}
      $\hfil\displaystyle{#1}\hfil$\crcr}}}
\newcommand{\Left}[1]{\overleftharpoonup{#1}}
\newcommand{\Right}[1]{\overrightharpoonup{#1}}

\def\Rrightharpoonupfill{$\mathsurround=0pt \mathord- \mkern-6mu
  \cleaders\hbox{$\mkern-2mu \mathord- \mkern-2mu$}\hfill
  \mkern-6mu \mathord\rightharpoonup\hspace*{-2.5ex}\rightharpoonup$}
\def\Roverrightharpoonup#1{\vbox{\ialign{##\crcr
      \Rrightharpoonupfill\crcr\noalign{\kern-1pt\nointerlineskip}
      $\hfil\displaystyle{#1}\hfil$\crcr}}}
\def\Lleftharpoonupfill{$\mathsurround=0pt \mathord\leftharpoonup\hspace*{-2.5ex}\leftharpoonup \mkern-6mu
  \cleaders\hbox{$\mkern-2mu \mathord- \mkern-2mu$}\hfill
  \mkern-6mu \mathord-$}
\def\Loverleftharpoonup#1{\vbox{\ialign{##\crcr
      \Lleftharpoonupfill\crcr\noalign{\kern-1pt\nointerlineskip}
      $\hfil\displaystyle{#1}\hfil$\crcr}}}

\newcommand{\biRight}[1]{\Roverrightharpoonup{#1}} \newcommand{\biLeft}[1]{\Loverleftharpoonup{#1}} 
\newcommand{\means}[1]{\llbracket\, #1 \,\rrbracket}

\newcommand{\vlt}[1]{_{\bullet\!}#1}
\newcommand{\vrt}[1]{#1_{\!\bullet}}
\newcommand{\transL}[1]{_{\bullet\!\!\!}\llparenthesis\, #1 \,\rrparenthesis}
\newcommand{\transR}[1]{\llparenthesis\, #1 \,\rrparenthesis_{\!\!\!\bullet}}
\newcommand{\transB}[1]{\llparenthesis\, #1 \,\rrparenthesis}

\newcommand{\eqdef}{\mathrel{\,\hat{=}\,}}
\newcommand{\smalleqdef}{\mathord{\!\!\widehat{=}\!\!}}

\newcommand{\Z}{\mathbb{Z}} 
\newcommand{\ceval}{\mathrel{\Downarrow}}  \newcommand{\beval}{\mathrel{\Downarrow}}  \newcommand{\fail}{\lightning}

\newcommand{\update}[3]{#1[#2\mapsto#3]} 
\newcommand{\specSym}{\leadsto}
\newcommand{\bspecSym}{\ensuremath{\mathrel{\ddot{\leadsto}}}} \newcommand{\rspecSym}{\ensuremath{\mathrel{\mbox{\footnotesize$\stackrel{\forall}{\raisebox{-.07ex}{$\thickapprox$}\hspace*{-1.2ex}>}$}}}}
\newcommand{\aespecSym}{\ensuremath{\mathrel{\mbox{\footnotesize$\stackrel{\exists}{\raisebox{-.07ex}{$\thickapprox$}\hspace*{-1.2ex}>}$}}}}

\newcommand{\spec}[2]{#1\specSym #2} 
\newcommand{\rspec}[2]{#1\rspecSym #2} 
\newcommand{\aespec}[2]{#1\aespecSym #2} 
\newcommand{\Aespec}[2]{#1\;\aespecSym\; #2} \newcommand{\especSym}{\ensuremath{\mathrel{\mbox{\footnotesize$\stackrel{\exists}{\leadsto}$}}}}
\newcommand{\espec}[2]{#1\especSym #2}
\newcommand{\bspec}[2]{#1\bspecSym #2} 
\newcommand{\Bspec}[2]{#1\;\bspecSym\;#2} 

\newcommand{\lexop}{\text{\color{blue}\textbf{$\langle\hspace*{-2.2pt}[$}}} \newcommand{\lexcl}{\text{\color{blue}\textbf{$\langle\hspace*{-2.3pt}]$}}} \newcommand{\rexop}{\text{\color{blue}\textbf{$[\hspace*{-2.4pt}\rangle$}}} \newcommand{\rexcl}{\text{\color{blue}\textbf{$]\hspace*{-2.2pt}\rangle$}}} 
\lstset{language=whyrel}

\begin{abstract} 
\ifsubmit\else Dedicated to the memory of Akira. \fi

Relational verification \changed{is pertinent to}{encompasses} reasoning about data abstraction, security and privacy, secure compilation, and functional specification of tensor programs, among others. 
Several relational logics and tools support compositional reasoning about $\forall\forall$ (2-safety) properties and, generally, $k$-safety properties. In contrast, few logics and tools exist for reasoning about $\forall\exists$ properties which are critical in the presence of nondeterminism.  This paper's primary contribution is a technique for verifying a $\forall\exists$ judgment by way of a novel transformation. The filter-adequacy transformation adds assertions to a product program in such a way that the desired $\forall\exists$ property (of a pair of underlying unary programs) is implied by a $\forall\forall$ property of the transformed product. The paper develops bicoms, a form of product programs that facilitates expressing good alignments and caters for direct translation of $\forall\forall$ properties to unary correctness; proves that successful $\forall\forall$ verification of a transformed bicom implies the $\forall\exists$ spec for its underlying unary commands; and evaluates the technique using a prototype auto-active relational verifier. The technique enables users to work with ordinary assertions and assumptions, and a standard assertion language, so that existing
$\forall\forall$ and unary verification tools can be used for $\forall\exists$.
\end{abstract}

\maketitle

\section{Introduction}\label{sec:intro}

Many desirable properties of programs are naturally expressed as relational
properties, that is, conditions that relate multiple executions.  One commonly
occurring pattern specifies that for any pair of terminating runs, if the
initial states satisfy a specified relational precondition then the final states
satisfy a specified post-relation.  
This pattern is known as $\forall\forall$ or 2-safety.
The pairs of runs could be from two
different programs $c$ and $c'$, in which case we will write $c\sep c':
\rspec{\R}{\S}$ where $\R$ (resp.\ $\S$) is the pre- (resp.\ post-) relation.
Examples include observational equivalence, where $\R$ and $\S$ are the identity
on \changed{observable}{} states.  A well-known example that relates a program to itself
is noninterference, a security property with respect to a partition of variables into secret or public:
take $\R$ and $\S$ to be agreement on the values of public variables.

Nondeterminacy is important in imperative programming, even in the absence of
concurrency: it can represent unknown inputs and underspecified procedures, and
it can approximate randomization.  For nondeterministic programs a commonly
occurring relational pattern, which we write $c\sep c': \aespec{\R}{\S}$, says
that for $\R$-related initial states, any run of $c$ can be matched by some (terminated) run
of $c'$ that establishes $\S$ finally.  
Examples include trace refinement (taking $\R,\S$ to be the identity on states) and possibilistic noninterference.
We refer to this as a $\forall\exists$ property,\footnote{The $\forall\exists$ pattern also arises in more general form, where multiple
runs are universally and/or existentially quantified and the runs are
constrained at intermediate points not just final.  However, the case of two
runs and pre-post spec is ubiquitous and illustrates key issues; we focus on it
in this paper since it admits more streamlined notations.}
by contrast with the $\forall\forall$ property written $\rspec{\R}{\S}$.
Owing to the broad range of applications, relational verification is
an active research area. There have been considerable advances for
$\forall\forall$ verification (and $k$-safety in general) but relatively little
for $\forall\exists$ which is the focus of this paper.
\changed{}{In a nutshell, we introduce a practical reduction of $\forall\exists$ to $\forall\forall$.}

For context, let us recall the state of the art in \emph{unary} (non-relational)
verification, i.e., trace properties.
There are several established means of formal specification. 
Most relevant here are partial correctness assertions (Hoare triples)
given by pre- and post-conditions.
For such specifications, verification techniques include: fully automated tools
based on proof search or direct semantic reasoning; 
interactive proof assistants for direct semantic reasoning or application of
Hoare logic in its many variations~\cite{cao2018vst,JungKJBBD18};  and 
\emph{auto-active} (a.k.a.\ deductive) tools in which users provide loop invariants, procedure specs, and other annotations.
When formal verification of general programs is done in industry it often
relies on auto-active tools like Dafny~\cite{Leino10}, Why3~\cite{Why3}, and
Viper~\cite{MuellerSS16}.  These tools leverage SMT solvers for automation
and address \changed{}{the rather severe} limitations of \changed{}{full} automation by enabling users to write assertions in code that help decompose reasoning and guide solvers to the use of relevant lemmas.  
Code annotation integrates well with development processes and does not require specialized expertise in logic and proof tactics.

The state of the art in relational
verification is less advanced but some principles have emerged.
One principle is \emph{alignment}, in which two runs are viewed as
matched sequences of segments: put differently, relational assertions are
associated with aligned pairs of intermediate points.  For $\forall\forall$ this
directly generalizes the inductive assertion method~\cite{Floyd67}.  
When the related programs are similar, good alignment can greatly simplify the needed intermediate assertions (e.g., so they can be expressed in solvable fragments of first order logic~\cite{ShemerGSV19,ISV26}).
Often the aligned pairs can be designated in terms of program control structure, but in
addition they may be designated by state-dependent conditions
which is especially important for alignment of loops.  

Another key principle is \emph{product programs} to represent multiple runs by a
single one.  This has mostly been explored for $\forall\forall$ properties,
where a product of programs (or transition systems) serves straightforwardly to
reduce the problem to ordinary verification and thereby leverage existing tools.
Products in the form of code can enable users to express useful
alignments together with invariants~\cite{BartheCK16}, and can serve as a framework in which to search
for good alignments~\cite{DickersonMD25}.

This paper addresses $\forall\exists$ properties, for which the state of the art
is less advanced.  
\changed{We focus on imperative programs acting on general data structures (and thus infinite state spaces).}{}  
\changed{One verification approach is direct semantic reasoning~\cite{LamportS21}.}{}
For code
and specifications in solvable assertion languages there have been exciting
advances in fully automated verification, e.g., constraint solving to find
relational invariants and alignment
conditions~\cite{UnnoTerauchiKoskinen21,ItzhakySV24}, with programs represented
by transition relations.  
But solvability is \changed{a very strong restriction}{highly restrictive.}
\changed{For much richer languages, refinement-oriented logics have been developed based on the Iris separation logic which is implemented in the Rocq proof 
assistant~\cite{JungKJBBD18}.}{For programs in richer languages, refinement-oriented logics have been developed using Iris/Rocq~\cite{JungKJBBD18}
but these require high levels of user effort.}
For sequential imperative programs, the recent
$\forall\exists$ logics RHLE and FEHL have bespoke features to cater for
automated proof search~\cite{DickersonYZD22,Beutner24}.  
\changed{Even for the basic imperative commands, no $\forall\exists$ logic has emerged as a standard.}{} 

What about auto-active relational verification?
For $\forall\forall$ properties there are some tools
(e.g.,~\cite{NagasamudramBN25,EilersDM23}).
For $\forall\exists$, to our knowledge the only such tool is Hypra~\cite{DardinierLM24};
it handles general hyperproperties including $\forall\exists$, but not relations between different programs and its support for \changed{}{loop} alignment is very limited.
\changed{The work reported here fills those gaps---providing both requisite theory and a prototype.}{The work reported here introduces a new approach.}

Our starting point is the observation \Xchanged{(made in~ \cite{AntonopoulosEtal2022popl,NBN25lmcs})}{} that judicious use of \keyw{assume} 
statements in a product program can serve to reduce a $\forall\exists$ property
to a $\forall\forall$ property of the product by filtering out executions that
are not helpful to witness the existential.  
\changed{Free}{But free} use of assumptions is unsound,
so \changed{some}{} additional checks are needed to justify the assumptions and also to
ensure the right (existential) side does not diverge when the left side
terminates. We call this the \emph{adequacy} of the product with respect to the
specification and underlying programs. \changed{The cited works use the observation for
theoretical purposes.  In this paper we show how the observation can be used
directly in practical verification leveraging existing tools.}{}

Our \emph{conceptual contribution} is a \changed{notation for product programs that facilitates filtering, and a}{novel} transformation that adds assertions to a product program in such a way that the desired $\forall\exists$ property (of underlying
unary programs) is implied by a $\forall\forall$ property of the transformed
product.  The \emph{filter-adequacy transformation} is applicable so long as nondeterminism is made
explicit in the form of standard havoc statements.  It enables $\forall\exists$
verification in which the user works with ordinary assertions and assumptions,
and an ordinary 
\changed{assertion language---opening the door to use of existing tools including but not limited to auto-active verifiers.}{first order assertion language. 
Thus it is compatible with existing $\forall\forall$ techniques 
including automated invariant inference and reduction to 
unary verification tools,
without limitation to solvable fragments or need for custom 
verification algorithms.
}
This verification technique is developed through the following technical contributions.

\begin{description}

\item[Bicoms:] a \changed{novel}{} form of product program that \changed{supports}{facilitates} the filter-adequacy transformation and enables users to express filters and alignments including state-dependent alignment of loops
(Section~\ref{sec:bicoms}).  Bicom semantics is \changed{}{novel and} carefully crafted to support
 translation to intermediate verification languages supported by existing tools
 (such as Boogie and Why3).  We define and prove correctness of such a
 translation (Section~\ref{sec:unary}).

\item[The filter-adequacy transformation] which adds assertions to check $\forall\exists$-adequacy of a bicom (Section~\ref{sec:filt-adeq}).  The main result is a soundness theorem which says that successful $\forall\forall$ verification of a transformed bicom implies the $\forall\exists$ spec for its underlying commands.  

\item[Evaluation using a proof of principle prototype] that implements the \changed{filter-adequacy}{} transformation.  It \changed{}{also} compiles bicoms to ordinary programs with assertions and assumptions in accord with the translation in Section~\ref{sec:unary}.
The latter are verified using an existing \changed{verification}{$\forall\forall$} tool\changed{}{~\cite{NagasamudramBN25}} based on
verification conditions and SMT solving\changed{, in accord with our goal to advance auto-active relational verification.}{.}  Our technique was successful on all
relevant benchmarks from recent works as well as on benchmarks beyond their
reach (Section~\ref{sec:eval}).
\end{description}
\changed{Further discussion of related work is in Section~\ref{sec:related}.}{The observation about filtering for $\forall\exists$ can be found, in some form,
in the prior works~\cite{AntonopoulosEtal2022popl} and~\cite{NBN25lmcs},
as discussed under related work in Section~\ref{sec:related}.
But~\cite{AntonopoulosEtal2022popl} lacks assertions and 
only uses the observation for theoretical purposes.
In~\cite{NBN25lmcs}, filtering is embodied in a bespoke form of unstructured product automata and there are no assume statements.
By combining assert and assume statements with bicoms we 
can apply filtering in practical verification leveraging existing tools.
Once one sees it the transform may seem straightforward,
but that is a strength, giving it broad applicability.
}

To conclude this section we remark on challenges\changed{}{---not at all straightforward---}that were overcome to achieve the main result.

Although our bicoms have similarities with alignment products in prior works, failure invalidates the key commutativity law some rely on~\cite{AntonopoulosEtal2022popl,DickersonMD25}.
One established way to
handle failure in a product semantics~\cite{BNNN19} is precluded by the need for a
practical translation of bicoms to ordinary programs.

The filter-adequacy transformation adds variables that must be fresh (and
distinct from each other, in the case of nested loops).  Fresh variables are
easy to implement using mutable state, but that does not facilitate
rigorous reasoning about freshness; to make the soundness proof manageable we
take a different approach that relies on a custom theory of framing.

A key role in the soundness proof is played by a program logic for the
$\forall\exists$ judgement (Section~\ref{sec:progs}).
Unlike the aforementioned $\forall\exists$ logics it
supports fully general \changed{data }{state-}dependent loop alignments and the rules are
relatively simple.  The rules are adapted from a recent work on
alignment~\cite{NBN25lmcs}, but soundness had to be proved from scratch for a new interpretation of the judgment which addresses failure.
Failure is needed in the program semantics because the filter-adequacy transformation relies on assert commands that can fail, and the handling of failure in the $\forall\exists$ judgment needs to be compatible with the translation of bicoms to ordinary programs.
Unlike~\cite{NBN25lmcs}, we machine check soundness of the rules;
we see the logic as a modest contribution on its own.

To use rules of the $\forall\exists$ logic in the filter-adequacy soundness proof we develop a compositional theory of weakest preconditions for bicoms.  
Some technical challenges arise in this and the preceding technical developments, 
due to bicoms supporting general \changed{data}{state}-dependent loop alignment, by contrast with more restrictive product forms in most prior $\forall\forall$ and $\forall\exists$ works. 

All our technical results have been machine checked in Rocq, which was done
because some definitions are delicate and some proofs long and intricate.  But
this is for simple imperative programs whereas the technique is generally
applicable, and applied in our prototype to a language with dynamic allocation,
recursive procedures, etc.  With broader application in mind we provide a
readable, calculational style proof of the soundness theorem\changed{.}{  (in appendix).}  It was the
basis of the machine checked proof and appears suited for extension to richer programming languages.

\section{Overview}\label{sec:overview}

This section provides an informal overview of $\forall\exists$ verification using our technique.
For each example, we look at the verification problem, a bicom that captures an alignment, and the bicom produced by the filter-adequacy transformation.  
Correctness of the transformed bicom implies, via the paper's main theorem, the $\forall\exists$ spec for the original pair of commands.

\paragraph{An introductory example: illustrating filtering}

Let $c$ and $c'$ be the \emph{havoc} commands $\havc{x}$ and $\havc{y}$
that set integer variables $x,y$ to arbitrary values.
Consider this verification problem:
$c \sep c': \aespec{\mathit{true}}{x \agreeRel y}$.  
We refer to $c$ as the left program and $c'$ as the right.
The postcondition, $x \agreeRel y$, says that the value of $x$ in the left ($\forall$) execution equals the value of $y$ on the right ($\exists$).

The correctness judgment holds provided that for any states $s, s', t$, 
if $c$ executed in $s$ terminates in $t$, there exists state $t'$ such that $c'$
has a terminating run from $s'$ ending in some $t'$ such that $t(x) = t'(y)$.
\changed{(}{}There is also a constraint about failure which we discuss later; havoc never fails.\changed{)}{}
Given $c,c'$, the \emph{embed} bicom $\biEmb{c}{c'}$ is a product that 
represents pairs of their executions.  The bicom spec $\bspec{\mathit{true}}{x \agreeRel y}$
says that from any pair of states $s, s'$, 
if running $\biEmb{c}{c'}$ from $(s,s')$ terminates in the pair of states $(t, t')$ then $t(x) = t'(y)$.
Note that bicom correctness is essentially a $\forall\forall$ property. For our example, the spec clearly does not hold for $\biEmb{c}{c'}$: owing to nondeterminism in both $\havc{x}$ and $\havc{y}$, $t(x)$ and $t'(y)$ need not agree.

To achieve bicom correctness we introduce a \emph{filtering condition}, in this
case $x\agreeRel y$.  The idea is that such a ``filtered'' bicom captures all
executions of $c$ on the left but keeps only executions of $c'$ on the
right for which the filtering condition $x\agreeRel y$ is met, that is, the
havoc'd $y$ on the right equals the havoc'd value of $x$ on the left. Havocing
$x$ on the left is accomplished using the bicom
$\biEmb{\havc{x}}{\skipc}$. Incorporating the filtering condition is
accomplished using the \emph{havoc-filter} bicom $\havRt{y}{(x \agreeRel y)}$
which can be understood as
$\biEmb{\skipc}{\havc{y}}; \keyw{assume}\;(x \agreeRel y)$. Here is the overall
filtered bicom:
\[
\biEmb{\havc{x}}{\skipc} \,;\, \havRt{y}{(x \agreeRel y)}
\]
The \havRtKeyword\ bicom assumes $x \agreeRel y$, but free use of assumptions is unsound for our purpose: 
one could use $\biEmb{\havc{x}}{\skipc}; \havRt{y}{\mathit{false}}$
which satisfies $\bspec{\mathit{true}}{x \agreeRel y}$ at the cost of having no executions at all.  Its correctness \changed{}{certainly} does not imply the desired $\forall\exists$ spec! 
What we need is that there exists at least one value of $y$ on the right that agrees with $x$'s value on the left. This property---adequacy---can be ensured by adding an assertion preceding the \havRtKeyword, like this: 
\[
  \biEmb{\havc{x}}{\skipc} \,;\, 
  \assertc{(\some{\smSep y}{x \agreeRel y})} \,;\, 
  \havRt{y}{(x \agreeRel y)}
\]
Here $\some{\smSep y}{\ldots}$ expresses quantification over $y$ on the right side.
This bicom is correct, that is, it satisfies $\bspec{\mathit{true}}{x \agreeRel y}$
which in particular expresses that there is no assertion failure.
Moreover correctness is easily checked with any ordinary verification technique or tool: The bicom can be translated to an ordinary (unary) command, interpreting the embed construct as sequence and renaming the two sides apart. 
In this example no renaming is needed, as the commands act on distinct variables, 
so the translation is 
$\havc{x};
\assertc{(\some{y}{x = y})};
\havc{y};
\assumec{(x = y)}$.

In general, renaming and other artifacts of translation introduce clutter 
that is at a lower level of abstraction than the relational verification task,
which is avoided by custom notation like our bicoms.  More on this in Sections~\ref{sec:unary} and~\ref{sec:eval}.

The step of adding the assertion is performed by the
filter-adequacy transformation, which is called $\chk$ for short (defined later).  The main
result of the paper is that if a bicom spec such as
$\bspec{\mathit{true}}{x \agreeRel y}$ holds for a transformed bicom then the
underlying commands satisfy the $\forall\exists$ spec which is
$\aespec{\mathit{true}}{x \agreeRel y}$ in this case.  \changed{The connection with the
original commands, $\havc{x}$ and $\havc{y}$ in the example, is a matter of}{Justifying the connection with the original commands, $\havc{x}$ and $\havc{y}$ in the example, is done by} 
syntactically projecting out the unary parts,
discarding relational assertions,
and discarding the assumption part of \havRtKeyword.
This yields $\havc{x} \sep \havc{y}: \aespec{\mathit{true}}{x \agreeRel y}$.  

Now consider a variation on the example, using the same spec 
$\aespec{\mathit{true}}{x \agreeRel y}$ 
but relating $\havc{x}$ with the command $\havc{y};\,y:=2\cdot y$.  
Filtering is only helpful in connection with nondeterminism so we consider this bicom:
\( \biEmb{\havc{x}}{\skipc}; \havRt{y}{(\ldots)}\,; \biEmb{\skipc}{y:=2\cdot y} \).
What should the elided assumption be?  
If we use $x\agreeRel y$, the assertion added by $\chk$ will be as before, and will not fail.
But that assumption does not support successful verification of $\bspec{\mathit{true}}{x \agreeRel y}$ because from $x\agreeRel y$ the assignment to $y$ establishes $x\agreeRel y/2$.
This suggests the assumption should be $x\agreeRel 2\cdot y$ which would suffice to prove the postcondition \changed{}{$x\agreeRel y$}.
But now $\chk$ inserts the assertion $\some{\smSep y}{x \agreeRel 2\cdot y}$ which 
fails because not all integers are even.
Indeed, verification should fail because this judgment is invalid:
$\havc{x} \sep \havc{y};y:=2\cdot y: \aespec{\mathit{true}}{x \agreeRel y}$.

\paragraph{Conditionally aligned loops} 
\begin{figure}[t]
  \begin{subfigure}[t]{.25\textwidth}
    \begin{lstlisting}
  z := 0; w := 0;
  while x > 0 do
    if w = 0 then
      hav z;
      x := x - 1;
    end;
    w := (w + 1) mod n
  done;
     \end{lstlisting}
  \end{subfigure}
    \begin{subfigure}[t]{.35\textwidth}
    \begin{lstlisting}
  |_ z := 0 _|; |_ w := 0 _|;
  while x > 0 | x > 0 .
      *<| w <> 0 *<] | [> w <> 0 |> do 
      variant { [> (n - w) mod n |> }
    if w = 0 | w = 0 then
       < hav z | skip >;
       havF z { z =:= z };
       |_ x := x - 1 _|;
    end;
    |_ w := (w + 1) mod n _|;
    done;
   \end{lstlisting}
  \end{subfigure}
    \begin{subfigure}[t]{.36\textwidth}
    \hfill
\begin{tikzpicture}[baseline=0, x=2.2cm, y=-0.45cm, font=\scriptsize]

            \tikzstyle{state} = [inner sep=1pt]

                  \node[state] (L0) at (0, 0) {$\;$};
      \node[state] (L1) at (0, 1.2) {$\;$};          \node[state] (L2) at (0, 2.4) {$w{:}0$};        \node[state] (L3) at (0, 3.8) {$w{:}1$};        \node[state] (L4) at (0, 5.2) {$w{:}0$};  
            \node[state] (R0) at (1, 0) {$\;$};
      \node[state] (R1) at (1, 1.2) {$\;$};
      \node[state] (R2) at (1, 2.4) {$w{:}0$};   
      \node[state] (R3) at (1, 3.8) {$w{:}1$};   
      \node[state] (R4) at (1, 5.2) {$w{:}2$};   
      \node[state] (R5) at (1, 6.6) {$w{:}0$}; 

            \begin{scope}[every path/.style={thick}]
                \draw (L0) -- (L1) node[midway, left=2pt] {$z:=0$};
        \draw (L1) -- (L2) node[midway, left=2pt] {$w:=0$};
        \draw (L2) -- (L3) node[midway, left=2pt] {\textit{body}};
        \draw (L3) -- (L4) node[midway, left=2pt] {\textit{body}};

                \draw (R0) -- (R1) node[midway, right=2pt] {$z:=0$};
        \draw (R1) -- (R2) node[midway, right=2pt] {$w:=0$};
        \draw (R2) -- (R3) node[midway, right=2pt] {\textit{body}};
        \draw (R3) -- (R4) node[midway, right=2pt] {\textit{body}};
        \draw (R4) -- (R5) node[midway, right=2pt] {\textit{body}};
      \end{scope}

            \begin{scope}[every path/.style={dashed, thin}]
                \draw (L0) -- (R0) node[midway, above=-1pt] {$x \agreeRel x$};
        \draw (L1) -- (R1) node[midway, above=-1pt] {$x \agreeRel x \land z \agreeRel z$};
        \draw (L2) -- (R2);
        \draw (L3) -- (R3);

                \draw (L4) -- (R3);         \draw (L4) -- (R4);         \draw (L4) -- (R5) node[midway, below=-1ex] {$\vdots$} ;       \end{scope}

    \end{tikzpicture}
  \end{subfigure}

  \vspace*{-2.5ex}
  \caption{Program \graybox{$c1$} (on the left) and a bicom for two copies of $c1$. Notation such as \protect\lstinline{|_ z := 0 _|} abbreviates $\biEmb{z:=0}{z:=0}$.
The right side depicts the beginning of aligned executions for $n=2$ on the left and $n=3$ on the right, with $w\mathord{:}0$ indicating a state where $w=0$. 
}
  \label{fig:whyrelext:cond-align}
\end{figure}

Consider the
unary program $c1$ shown in \autoref{fig:whyrelext:cond-align} (using concrete
syntax of our prototype which has been used to verify the examples).  This
program havocs $z$ repeatedly, until $x$ is decremented to zero,
but some iterations stutter---leaving $x$ and $z$ unchanged---if $n>1$.
We require $n>0$ on both sides, but not $n\agreeRel n$.  We aim to show
\begin{equation}\label{eq:c1}
c1 \sep c1: \aespec{\leftF{n>0}\land\rightF{n>0}\land x \agreeRel x}{z \agreeRel z}
\end{equation}
where $\leftF{n>0}$ says the left state satisfies $n>0$
and $\rightF{n>0}$ says the same on the right.

To prove (\ref{eq:c1}) we must show, given a pair of initial states that 
agree on $x$ and 
both
satisfy $n>0$, that for any run of $c1$ from the left state there is a run of
$c1$ from the right state ending in states where $z\agreeRel z$ holds.  We align
the two runs by establishing a correspondence between steps where control is at
the top of the loop on both sides, such that $z\agreeRel z \land x\agreeRel x$
holds at corresponding pairs of states.  

An alignment of the first few iterations is depicted on the right of \autoref{fig:whyrelext:cond-align}.
The two initial states are aligned, as indicated by dashes, and related by $x\agreeRel x$,
and we assume $n=2$ on the left and $n=3$ on the right.
The assignments $z:=0$ are aligned in lockstep, meaning that their ending states are aligned.  Those states are related by $z\agreeRel z \land x\agreeRel x$. 
The assignments $w:=0$ are aligned in lockstep as are the first iterations 
of the loop body.  Then, for the aligned states where $w$ is 1, 
the alignment has a \emph{left-only iteration}
leading to aligned states which again satisfy 
$z\agreeRel z \land x\agreeRel x$ because $x$ and $z$ are unchanged when $w\neq 0$.
Next, the alignment has a \emph{right-only iteration}, leading to aligned states with $w:0$ on the left and $w:2$ on the right. For whatever value $z$ was assigned on the left, there is an execution where $z$ gets the same value on the right, in which case 
$z\agreeRel z \land x\agreeRel x$ holds at aligned states with $w:0,w:2$.
Another right-only iteration can lead to 
$z\agreeRel z \land x\agreeRel x$ holding at the last dashed line, which is poised to do another lockstep, i.e., \emph{joint iteration}.

This notional alignment can be designated by two relations called the left and right \emph{alignment conditions}. Here the conditions are $\leftF{w\neq 0}$ and $\rightF{w\neq 0}$ respectively.  (As often happens in practice, the left alignment condition \changed{happens to constrain}{constrains} only the left state and likewise for the right.)  Such alignment conditions appear in the relational logic rule for loops discussed in Section~\ref{sec:progs} but here our interest is with bicoms.
\autoref{fig:whyrelext:cond-align} shows a bicom for the pair $c1 \sep c1$.
The notation \lstinline{if...|...} indicates that two if-commands are aligned and are implicitly asserted to take the same branch.
The so-called bi-while represents two related loops, and the syntax includes not only the two loop tests but also the left and right alignment conditions.  
The semantics of bi-while is designed to embody the alignment designated by these conditions.  
Under the condition for left-only (resp.\ right-only) iteration,
the bi-while body is projected to its left (resp.\ right) component.
(Explained later, see Figures~\ref{fig:project} and~\ref{fig:biproject}.)
The bi-while body in \autoref{fig:whyrelext:cond-align}
includes a filtering condition for $\havc{z}$ that assumes $z\agreeRel z$, so all terminating runs of this bicom indeed maintain invariant $z\agreeRel z$.

The bi-while is also annotated with a variant.  Its significance is better
motivated in the upcoming example which is a program that can diverge.
But let us explain the issue now.  In an alignment of loops that intersperses left-only,
right-only, and joint iterations, divergence on the left, or
jointly, cannot falsify a $\forall\exists$ property, but divergence on the right
alone can do so.
A bi-while can embody such alignments, so the $\chk$
transformation must preclude diverging right-only iterations.  The
keyword \lstinline{variant} introduces an annotation used by the \chk\ function;
it has no effect on the meaning of the bicom.  In \autoref{fig:whyrelext:cond-align} the variant
$\rightex{(n - w) \mathbin{mod} n}$ is an integer expression evaluated in the
right state. Note that it is \emph{not} a conventional variant (termination
measure): execution on the right does not always decrease the value.  It does,
however, decrease when $\rightF{ w \neq 0 }$ is true, and that is the only
condition under which right-only iteration occurs in our chosen alignment.

\begin{figure}[t]
\begin{lstlisting}
  |_ z := 0 _|; |_ w := 0 _|;
  while x > 0 | x > 0 . *<| w <> 0 *<] | [> w <> 0 |> do variant { [> (n - w) mod n |> }
    var vsnap : int, rosnap : bool in
    vsnap := (n - w) mod n; /* added by chk: snapshot variant */
    rosnap := [> w <> 0 |> ; /* added by chk: snapshot of w <> 0 on right */
    if w = 0 | w = 0 then
       < hav z | skip >;
       assert { exists |z. z =:= z }; /* added by chk */
       havF z { z =:= z };
       |_ x := x - 1 _|;
    end;
    |_ w := (w + 1) mod n _|;
    /* added by chk: assert variant decreases under right-only iterations */
    assert { rosnap -> [> 0 <= (n - w) mod n < vsnap |> };
  done;
\end{lstlisting}
\vspace*{-2ex}
\caption{The $\chk$ function applied to the bicom in \autoref{fig:whyrelext:cond-align}.}
\label{fig:whyrelext:chkc1}
\end{figure}

The result of applying $\chk$ on the bicom in
\autoref{fig:whyrelext:cond-align} is shown in
\autoref{fig:whyrelext:chkc1}. Fresh variable \lstinline{rosnap} snapshots the
truth value of the condition $\rightex{w\neq 0}$ under which the current
iteration will be right-only, and \lstinline{vsnap} snapshots the variant.  The
added assertion requires the variant to decrease (and remain non-negative) during right-only iterations.
The checked bicom can be verified for $\bspec{\leftF{n>0}\land\rightF{n>0}\land
x \agreeRel x}{z \agreeRel z}$, 
\changed{}{using suitable invariants,} 
so the main theorem of the paper lets us conclude 
(\ref{eq:c1}).

Please note that this example does not admit proof based on lockstep alignment of the loops:  if the left value of $n$ exceeds that of the right, it takes more iterations on the left to reach termination, by which time the right has finished and cannot match the left havocs.  In this example only the last one needs to be matched, so it can be verified in systems that support sequential alignment.  But in more realistic cases the sequential alignment of loops has a high price because strong invariants are needed to ``remember everything''~\cite{Francez83,PickFG18} about the left execution.

\paragraph{Possibilistic noninterference: illustrating may termination (right-sided divergence)}

\newcommand{\vhigh}{\ensuremath{\mathit{high}}}
\newcommand{\vlow}{\ensuremath{\mathit{low}}}

The preceding examples are terminating programs.  The next example, adapted
from \citet{UnnoTerauchiKoskinen21}, involves loops that can diverge. 
We refer to it as \graybox{$c2$}:
\[
\begin{array}{l@{\,}l}
\keyw{if}\ \vhigh \neq 0\  & \keyw{then}\  \havc{x};\ 
                            \keyw{if}\ x \geq \vlow\ \keyw{then}\ \skipc\
                            \keyw{else}\ (\keyw{while}\ true\ \keyw{do}\ \skipc) \\
     & \keyw{else}\  x := \vlow; \, \havc{b}; 
       \, (\keyw{while}\ (b \neq 0)\ \keyw{do}\ x := x+1;\ \havc{b})
\end{array}
\]

We aim to show $c2 \sep c2: \aespec{\vlow\, \agreeRel\, \vlow}{x\,\agreeRel\,
x}$.  All four combinations of initial values of $\vhigh$ must be considered,
and different alignments are needed in different cases.  
The desired alignments are expressed using bi-if: a 4-way conditional
\[ \ifFourLong{(\vhigh\neq 0\sep \vhigh\neq 0)}{B_1}{B_2}{B_3}{B_4} \]
with bicoms $B_1,\ldots,B_4$ shown in \autoref{fig:unno}.  For the case that
both initial states satisfy $\vhigh\neq 0$ (resp.\  $\vhigh=0$) in both
initial states, we use mostly lockstep alignment: $B_1$ in \autoref{fig:unno}
(resp.\ $B_4$).  

For the case that $\vhigh\neq 0$ in the left initial
state but $\vhigh=0$ in the right initial state, we use sequential, left-first
alignment: $B_2$ in \autoref{fig:unno}. Note the possibility of nontermination
of the left execution when $x < \vlow$. 
If the left execution does
terminate, $x\geq low$ can be asserted. 
(We show that, but omit other user-provided assertions and invariants used for this example.)
In the right execution,
the \keyw{WhileR} bicom abbreviates a bi-while where the left loop test and left
alignment conditions are false, but the right alignment condition is
true. (Recall from the previous example that when a left (resp. right) alignment
condition holds, the bicom does a left (resp. right)-only iteration; for joint
iterations, both left and right alignment conditions are false.) In essence
this is just a loop on the right only, but written in a form that
allows \havRtKeyword\ to be used.

Not all executions of $c2$ terminate, but for the $\forall\exists$ property 
$\aespec{\vlow\, \agreeRel\, \vlow}{x\,\agreeRel\,x}$
we only need the existence of a terminating right execution (may
termination), when the left does. This is achieved in $B_2$ by filtering right executions to only allow those
where $\havc{b}$ sets the value of $b$ in the right state (written
$\rightex{b}$) to the difference of $x$'s values in the left and right states,
and using $\rightex{b}$ as variant\changed{: filtering}{. The filtering} forces $\rightex{b}$ to decrease
in every iteration. 
The \keyw{whileR} in $B_2$ declares $\rightex{b}$ as variant,
which causes $\chk$ to add snapshots and an assertion that the variant is decreased (like in \autoref{fig:whyrelext:chkc1}).

For lack of space we refrain from detailed descriptions of $B_3$ and $B_4$.
To \changed{validate}{justify} the assumptions in \havRtKeyword\
bicoms, $\chk$ introduces assertions 
\changed{(not shown in Figure~\ref{fig:unno}) preceding them.}{preceding them,
as shown in previous examples but not shown in Figure~\ref{fig:unno}.}
For example, $\assertc{\some{\smSep b}{\rightex{b}=\leftex{x}-\rightex{x}}}$ in $B_2$.

\begin{figure}[t]
  \begin{subfigure}[t]{.45\textwidth}
    \begin{lstlisting}[mathescape=true]
 ($B_1$) <hav x | skip>; havF x { x =:= x};
      If (x >= low) | (x >= low) then
        |_ skip _|
      else
        <while true do skip done | skip>;
        Assert { *<| false |> };
        <skip | while true do skip done>;
      end;

    \end{lstlisting}
  \end{subfigure}
    \begin{subfigure}[t]{.45\textwidth}
    \begin{lstlisting}[mathescape=true]
 ($B_2$) <hav x | skip>;
      <if x >= low then skip
       else while true do skip done
      | skip>;   Assert { *<| x >= low <] };
      <skip | x := low>;
      havF b { [> b |> = *<| x *<] - [> x |> };
      WhileR b <> 0 do
        variant { [> b |> }
        <skip | x := x+1>;
        havF b { [> b |> = *<| x *<] - [> x |> }
      done;
    \end{lstlisting}
  \end{subfigure}
    \begin{subfigure}[t]{.45\textwidth}
    \begin{lstlisting}[mathescape=true]
 ($B_3$) <x := low; hav b | skip>;
      <while b <> 0 do x := x+1; hav b done 
      | skip>;
      havF x { x =:= x };
      <skip | if x >= low then skip
              else assert {false};
                   while true do skip done end>
    \end{lstlisting}
  \end{subfigure}
  \begin{subfigure}[t]{.45\textwidth}
    \begin{lstlisting}[mathescape=true]
 ($B_4$) |_ x := low _|;
      <hav b | skip>; havF b { b =:= b };
      while b <> 0 | b <> 0 do
        |_ x := x+1 _|;
        <hav b | skip>; havF b { b =:= b };
      done;
    \end{lstlisting}
  \end{subfigure}
\vspace*{-3ex}
  \caption{Bicoms $B_1\ldots B_4$ capture different alignments of 
conditional branches for $c2 \sep c2$.}
  \label{fig:unno}
\end{figure}

\paragraph{Summary}
The filter-adequacy transformation supports the following novel technique
for verifying a $\forall\exists$ judgment $c\sep c':\aespec{\R}{\S}$.

\begin{enumerate}
\item\label{it:one} Find a bicom $B$ that represents $c,c'$ in the sense that its left and right projections are semantically equivalent to $c$ and $c'$.
Moreover $B$ should represent a helpful alignment in which right-side havocs are accompanied by filtering conditions 
and simple relational invariants can be used.
Finding $B$ is not the focus of this paper; it can be done with user guidance or automated search, possibly \changed{using equivalence-preserving}{by} rewriting of $\biEmb{c}{c'}$ as discussed under related work (Section~\ref{sec:related}).  
A simple conservative check for the semantic equivalences is that $c,c'$ should be the same as the syntactic projections of $B$, modulo simplification (see Def.~\ref{def:kateq}); this was sufficient in our evaluation\changed{.}{ (Section~\ref{sec:eval}).}
\item\label{it:two} Transform $B$ by applying $\chk$.  This is a linear time transformation.
\item\label{it:three} 
\changed{Verify}{For the transformed bicom $\chk(B)$, verify}
$\forall\forall$ property $\chk(B): \bspec{\R}{\S}$.  
If verification is successful then $c\sep c':\aespec{\R}{\S}$ holds by the main theorem of the paper (see Section~\ref{sec:main}).  
\end{enumerate}
One way to verify 
$\chk(B): \bspec{\R}{\S}$ is to reduce it to standard unary verification:
First, 
encode $\R$ and $\S$ as assertions $enc(\R),enc(\S)$ where variables are renamed to encode state pairs as single states. 
Then \changed{encode}{translate} $\chk(B)$ \changed{as}{to} a unary program $\transB{\chk(B)}$
and verify the ordinary Hoare triple which we write as 
$\transB{\chk(B)} : \spec{enc(\R)}{enc(\S)}$.  
For such \changed{an encoding}{a translation} see Section~\ref{sec:unary}.

\section{Programs, specifications, and a relational program logic}\label{sec:progs}

This section formalizes the syntax and semantics of commands
and the relational correctness properties of interest.
\changed{It also gives key $\forall\exists$ proof rules we rely on.}{It also gives $\forall\exists$ proof rules that are key to proving our 
main result in Section~\ref{sec:main}.}

We treat assertions and relations by shallow embedding,\footnote{Readers not
   familiar with shallow embedding can consult~\cite{NielsonBook92} or 
   the chapter \texttt{Hoare.v} in~\cite{PierceSF2}. See also~\cite{WangCT24}.} 
i.e., as sets of stores and relations
on stores rather than syntactic formulas.  This is done both in specifications
(Section~\ref{sec:correctness}) and in code, specifically in \keyw{assert} commands
and the \havRtKeyword\ construct.  Use of shallow embedding provides generality,
in that our results are not tied to a specific assertion language. A key benefit
is that weakest preconditions ($wlp$) can be defined in the ambient logic, which avoids
the need to confront expressiveness issues~\cite{AptOld3}.\footnote{Application of our results does not rely on wlp being expressible in the assertion language; that is only a consideration for proving completeness,
and is irrelevant here owing to shallow embedding.}
In the paper the
ambient logic is ordinary set theory.  In the mechanization the
ambient logic is that of Rocq.\footnote{Making use of these standard libraries:
\texttt{Classical}, \texttt{FunctionalExtensionality}, and
\texttt{PropExtensionality}.}  Given that the program syntax has ``semantic''
assertions in this manner, we also choose shallow embedding of expressions in
code.  This precludes reasoning by structural induction on expressions or
assertions but we have no need for that.

We do need logical operations like quantification and substitution for assertions and relations. We sketch key definitions, which amount to the usual semantics of predicate logic formulas, and any missing details can be found in the Rocq development.
We use conventional syntactic notations for expressions and assertions, 
for example $p\land q$ instead of $p\intersect q$.
So the casual reader can ignore the finer points about this shallow embedding. 

\subsection{Assertions and commands}

We assume given denumerable sets $\keyw{boolVar}$ and $\keyw{intVar}$ that are disjoint.
Metavariables $x,y,\ldots$ range over
$\keyw{boolVar}\union\keyw{intVar}$.  A \dt{store} is a total function from
$\keyw{boolVar}\union\keyw{intVar}$ to values that maps integer variables to
$\Z$ and boolean variables to $\{\True,\False\}$.  The set of all stores is named
$\Store$.  
An integer (resp.\ boolean) \dt{expression} $e$ is a total function $\Store\to\Z$ (resp.\ $\Store\to\{\True,\False\}$).  
Metavariables $s,t,u\dots$ range over stores.
We write $\update{s}{x}{v}$ for the update of store
$s$ to map $x$ to value $v$.  
An \dt{assertion} $p$ is a set $p\subseteq\Store$.
We write \graybox{$s\models p$} for $s\in p$.

As mentioned earlier, we use ordinary formula notation for logical operations.
For example we write $p\imp q$ instead of $(\Store\setminus p)\union q$.  
Let $\imp$ bind less tightly
than $\land$ and $\lor$.  
For assertion $p$ we write $\subst{p}{x}{e}$ for
semantic substitution: $\subst{p}{x}{e}$ is defined to be $\{ s \mid
\update{s}{x}{e(s)} \in p \}$.  This enjoys standard properties of
syntactic substitution, e.g., distributing through boolean operators.
Quantifiers, as operators on store sets, are defined by $\graybox{$\all{x}{p}$}
\eqdef \{ s \mid s\in \subst{p}{x}{v} \mbox { for all }v\in\Z \}$ and
$\graybox{$\some{x}{p}$} \eqdef \{ s \mid s\in \subst{p}{x}{v} \mbox { for some
}v\in\Z \}$.  Here and in the sequel we use words to distinguish metalanguage
from the operators on assertions; but sometimes we use symbols for metalanguage
when needed for succinctness (e.g., in Definition~\ref{def:semframe}).
In some contexts we implicitly convert a boolean expression $e$ to the assertion $\{ s
\mid e(s)=\True \}$.

Commands are given by this grammar,
where $e$ ranges over expressions and $p$ over assertions:
\( c::= \skipc \mid x:=e \mid \havc{x} \mid \assertc{p} 
     \mid c;c \mid \ifc{e}{c}{c} \mid \whilec{e}{c} 
\). 
We assume without further comment that commands are type correct, e.g., conditional tests are
boolean expressions and in $x:=e$ the expression has the right type for the variable $x$.  

We use an evaluation (big-step) semantics for commands.
The outcome $\fail$ (pronounced \dt{fail}) represents assertion failure which is the only form of runtime failure in our idealized language.
The relation \graybox{$c/s\ceval \phi$} says that from initial store $s$ the command yields outcome $\phi$, where $\phi$ ranges over the disjoint union $\Store\union\{\fail\}$.
A few cases are in \autoref{fig:commandSemSelected}.  

\begin{figure}[t]
\begin{footnotesize}
\begin{mathpar}
\inferrule{s\models p}{\assertc{p} / s \ceval s}

\inferrule{s\not\models p}{\assertc{p} / s \ceval \fail}

\inferrule{n\in\Z}{\havc{x} / s \ceval \update{s}{x}{n}}

\inferrule{c/s \ceval t\\ d/t \ceval \phi}{c;d/ s \ceval \phi}

\end{mathpar}
\end{footnotesize}
\vspace*{-3ex}
\caption{Semantics of selected commands.}
\label{fig:commandSemSelected}
\end{figure}

\subsection{Relations and relational correctness judgments}\label{sec:correctness}

A \dt{store relation} (relation, for short) is a set $\R\subseteq\Store\times\Store$.
Relations are ranged over by identifiers $\P,\Q,\R,\S,\ldots$.  
The requisite operations include the following:\footnote{The funny
notations using $x\smSep$ and  $\smSep x$ are taken from~\cite{BNNN19,NBN25lmcs} and cater for two-side forms like $\subst{\R}{x\smSep x'}{e\smSep e'}$.}
quantification over variables on the left (written 
$\all{x\smSep }{\R}$ and $\some{x\smSep }{\R}$), on the right
($\all{\smSep x}{\R}$ and $\some{\smSep x}{\R}$), and 
of course $\land,\imp,\dots$.
We also use the form  $E \oplus E$
where $E$ ranges over two-state expressions and $\oplus$ ranges over primitive relations including equality and inequality.
\dt{Two-state expressions} are integer-valued expressions that depend on a pair of states.  To be precise, we let $E$ range over functions $\Store\times\Store\to\Z$, 
and we write two-state expressions using the left and right embeddings of unary expressions:
$\graybox{$\leftex{e}$} \eqdef (s,t)\mapsto e(s)$ and 
$\graybox{$\rightex{e}$} \eqdef (s,t)\mapsto e(t)$.

For unary assertions $p$, the form $\leftF{p}$ (resp.\ $\rightF{p}$) says the left 
(resp.\ right) store satisfies $p$.
That is, $\graybox{$\leftF{p}$} = \{(s,t)\mid s\in p\}$ and 
$\graybox{$\rightF{p}$} = \{(s,t)\mid t\in p\}$.
As an example, $\leftF{x>0}$ says $x$ on the left is positive,
which can also be written $\leftex{x} > 0$.
We write \graybox{$e \agreeRel e'$} as abbreviation of $\leftex{e}=\rightex{e'}$.

For relation $\R$ on stores we write 
$\subst{\R}{x\smSep}{e\smSep}$ for substitution of $e$ for $x$ in the left store,
which is defined by 
$(s,t)\in \subst{\R}{x \smSep }{e \smSep }$ iff $(\update{s}{x}{e}, t)\in \R$.
Similarly $\subst{\R}{\smSep x}{\smSep e}$ substitutes on the right.
Now we can define quantification, for example 
$\graybox{$\some{x\smSep }{\R}$} 
\eqdef\{(s,s')\mid (s,s')\in\subst{\R}{x\smSep }{v\smSep } \mbox{ for some } v\in\Z \}$.
These operations enjoy the usual properties of operations on formulas, e.g., $\subst{\R}{|x}{|e} = \R$ if $\R$ does not depend on $x$.

An assertion is \dt{valid}, written \graybox{$\models p$}, iff $s\models p$ for
all $s\in \Store$.  Note that $p\subseteq q$ iff $\models (p\imp q)$.  We write
$s,s'\models\R$ for $(s,s')\in\R$.  (We often use primed identifiers, simply as
mnemonic for things on the right.)  Validity of a relation $\R$,
written \graybox{$\models \R$}, means $s,s'\models\R$ for all $s,s'$.  Owing to
shallow embedding, $\models (\R\iff\S)$ says we have $\R=\S$, i.e., the same
sets of pairs.

\begin{definition}\label{def:relcorrect}
A $\forall\forall$ correctness judgment $c\sep c':\rspec{\R}{\S}$ is
\dt{valid}, written \graybox{$\models c\sep c':\rspec{\R}{\S}$}, iff
for all $s,s',\phi,\phi'$, if $s,s'\models\R$
and $c/s\ceval \phi$ and $c'/s'\ceval \phi'$ then
$\phi,\phi'\models\S$.

A $\forall\exists$ correctness judgment $c\sep c':\aespec{\R}{\S}$ is \dt{valid}, written
\graybox{$\models c\sep c':\aespec{\R}{\S}$}, iff
for all $s,s'$, if $s,s'\models\R$ then 
(i)  $c/s\not\ceval\fail$, and 
(ii) for all $t$, if $c/s\ceval t$ 
then there is $t'$ such that $c'/s'\ceval t'$ and $t,t'\models\S$. 
\end{definition}
Please note, apropos $\rspecSym$, that $\phi,\phi'\models\S$ implies $\phi\neq\fail$ and  $\phi'\neq\fail$
because $\S$ is a relation on stores, not on outcomes;
so neither side is allowed to fault unless the other side diverges.

\begin{figure}[t!]
\begin{footnotesize}
\begin{mathpar}

\inferrule[eHavSkip]{}{ 
  \havc{x}\sep\skipc : \aespec{(\all{x\smSep}{\R})}{\R}
}

\inferrule[eSkipHav]{}{ 
  \skipc\sep\havc{x}:\aespec{(\some{\smSep x}{\R})}{\R}
}

\inferrule[eRewrite]{ c\sep c': \aespec{\Q}{\R} \\ c\kateq d \\  c'\kateq d' }
{ d\sep d': \aespec{\Q}{\R}  }

\inferrule[eDo]{ 
  c\sep \skipc : \aespec{\Q\land \leftF{e}\land\P }{\Q}
   \\
\mbox{$ \skipc\sep c' : \aespec{\Q\land \rightF{e'}\land\P'\land (n=E) }{\Q\land (0\leq E < n)}$ for all $n\in\Z$}
   \\
  c\sep c' : \aespec{\Q\land \leftF{e}\land\rightF{e'} \land \neg\P \land \neg\P'}{\Q}
  \\
  \Q\imp (\leftex{e} = \rightex{e'}
          \lor (\P \land \leftF{e})\lor (\P' \land \rightF{e'})) 
}{ \whilec{e}{c} \Sep \whilec{e'}{c'} : 
  \aespec{\Q}{\Q\land \neg\leftF{e}\land\neg\rightF{e'}}
}

\end{mathpar}
\end{footnotesize}
\vspace*{-6ex}
\caption{ERHL: selected core rules for the $\aespecSym$ judgment
(others are in appendix \autoref{fig:ERHLcoreaddl}) }
\label{fig:ERHL}
\end{figure}
\autoref{fig:ERHL} gives proof rules for the judgment 
$c\sep c':\aespec{\R}{\S}$.  
We refer to these rules as the logic \dt{ERHL}, as they
are adapted from the logic named ERHL+~\cite{NBN25lmcs} (see Section~\ref{sec:related}).

Rule \rn{eDo} supports general \changed{data }{state-}dependent alignments.
Recall from the discussion in connection with example $c1$ (\autoref{fig:whyrelext:cond-align})
that any pair of loop runs can be considered with the iterations aligned 
in lockstep or with some iterations considered to advance on one
side only.  Relation $\Q$ in \rn{eDo} is the invariant
(like $x\agreeRel x\land z\agreeRel z$ in \autoref{fig:whyrelext:cond-align}).
Relation $\P$ is the left alignment condition; it 
designates the condition under which the alignment advances left-only (resp.\ $\P'$ for right-only).  
The premises for $c\sep\skipc$ and $\skipc\sep c'$ cover these cases. 
The premise for $c\sep c'$ covers joint iteration.
The fourth premise ensures that the three cases cover all possibilities
unless and until both sides have terminated.
That is, the alignment designated by $\P,\P'$ is adequate.

Instantiating $\P,\P'$ with false gives lockstep alignment
(like in~\cite{Benton:popl04}).
Instantiating both $\P,\P'$ with true gives alignment with arbitrary interleaving of one-sided iterations (rarely helpful), subject to a restriction on right-only iterations.
In \rn{eDo} the premise for right-only requires a variant, $E$, to preclude the
unsound use of right-only alignment if the right side diverges when the left
terminates.  The condition is weaker than termination; it merely bounds the
number of right-only iterations that can be aligned after a given number of
iterations on the left.\footnote{In rule \rn{eDo} the right-only premise is a
family of premises indexed by integers $n$.  Put differently, $n$ is a
universally quantified variable in the metalanguage.  It serves to snapshot the
value of the variant $E$ at the start of an iteration.  It is possible to
formalize the rule using instead a single premise with a fresh program variable
in place of $n$.  We choose this version (following~\cite{NBN25lmcs} and many
other works) because it avoids technicalities about freshness that would add
complications to the already intricate proofs of Lemma~\ref{lem:uChk_term} and
Theorem~\ref{thm:main}.  }

The following is used in rule \rn{eRewrite} and later to connect commands with bicoms.
 
\begin{definition}\label{def:kateq}
Let \dt{kat equivalence},\graybox{$\kateq$}, be the relation on commands 
that is the least congruence that satisfies the following (for all $c,d$):
\(
\begin{array}[t]{ll}
\skipc;c \kateq c 
\quad
& \ifc{\True}{c}{d}\kateq c \\
c;\skipc \kateq c
& \whilec{\False}{c}\kateq \skipc 
\end{array}\).
\end{definition}
Here congruence is with respect to the command combinators: sequence, if, and while.
We use the term \dt{kat equivalence} with reference to KAT~\cite{Kozen97}, a theory of imperative control structure that has been used to minimize the number of core rules needed for a relational logic~\cite{NBN25lmcs}.   
Results in this paper hold if $\kateq$ is defined with additional equations
(as in \cite{NBN25lmcs}).
What matters is that $\kateq$ implies semantic equality (appendix Lemma~\ref{lem:bieq_bsem_equiv})
because the correctness judgment is about semantics.

\begin{theorem} \label{thm:soundness}
The proof rules in Figure~\ref{fig:ERHL} (and appendix figure~\ref{fig:ERHLcoreaddl}) are all sound.
\end{theorem}

This section presented commands, their correctness specs and judgments,
and the $\forall\exists$ logic used later to prove the main theorem
(Section~\ref{sec:main}). On this basis we proceed to introduce bicoms.

\section{Bicoms}\label{sec:bicoms}

Our bicoms are inspired by \changed{similar}{} forms of alignment product used in prior work
(see Section~\ref{sec:related}). 
\changed{But the}{The} details are carefully designed in order to satisfy three criteria.
First, the syntax and semantics should facilitate straightforward translation into unary code 
\changed{such as an intermediate verification language} (see Section~\ref{sec:unary}).
Second, there should be full support for \changed{data }{state-}dependent alignment of loop iterations
and general alignment of conditionals (as opposed to the same-branch conditionals 
in works like~\cite{BartheCK16,BNNN19}).
Third, there should be a big-step semantics with clear connection to 
command semantics---since the ultimate goal is to reason about actual program behavior\changed{.}{ as summarized at the end of Section~\ref{sec:overview}.}

\subsection{Bicoms and projections}\label{sec:syntax}

The syntax of bicoms is as follows, overloading some keywords without ambiguity. 
\[ B ::= 
  \biEmb{c}{c} \:\bigmid\:  \assertc{\R} \:\bigmid\:  \havRt{x}{\R}    \:\bigmid\:  B;B \:\bigmid\:  \ifFour{e|e}{B}{B}{B}{B}  
   \:\bigmid\:  \biwhile{e|e}{\R|\R}{B}
\]
We let $B,C,D$ range over bicoms but reserve $E$ for two-state expressions (Section~\ref{sec:correctness}).
In the bi-while construct, $\biwhile{e|e'}{\P|\P'}{B}$, we call $\P$ (resp.\ $\P'$) the left (resp.\ right) \dt{alignment condition}; examples are $\leftF{w\neq 0}$ and $\rightF{w\neq 0}$ 
in \autoref{fig:whyrelext:cond-align}.
As introduced in Section~\ref{sec:overview}, the \dt{havoc-filter}
$\havRt{x}{\R}$ has the effect of havoc on $x$ on the right side, followed by assuming $\R$.

The grammar shows a compact notation for bi-if, but it has an alternate notation that is mnemonic and close to the syntax 
in our prototype:
$\graybox{$\ifFourLong{e|e'}{B_1}{B_2}{B_3}{B_4}$}$
is the same as $\ifFour{e|e'}{B_1}{B_2}{B_3}{B_4}$.
The form allows different alignments to be used for different combinations of the branch conditions. 
The benefit of this is evident in example $c2$ in Section~\ref{sec:overview}.
A common idiom is alignment of two if-commands that are expected to follow the same branch.  For $\ifc{e}{c}{d}$ and $\ifc{e'}{c'}{d'}$ this can be written 
\begin{equation}\label{eq:alignIf}
  \assertc(\leftex{e}=\rightex{e'});
  \ifFourLong{e|e'}{\biEmb{c}{c'}}{\biEmb{c}{d'}}{\biEmb{d}{c'}}{\biEmb{d}{d'}}
\end{equation}
This form allows to then replace $\biEmb{c}{d'}$ with $\assertc{\False};\biEmb{c}{d'}$
which is easy to verify, and likewise for $\biEmb{d}{c'}$.
Our prototype's syntax sugar for (\ref{eq:alignIf}) is used 
in \autoref{fig:whyrelext:cond-align} and in $B_1$ in \autoref{fig:unno}.

\begin{figure}[t]
\begin{footnotesize}
\( \begin{array}{l@{\;}c@{\;}l}   
\Left{\biEmb{c}{c'}} &=& c \\
\Left{\assertc(\P)} &=& \skipc \\
\Left{\havRt{x}{\R}} &=& \skipc \\
\Left{\ifFour{e|e'}{B_1}{B_2}{B_3}{B_4}}  
&=& \ifc{e}{\Left{B_1}}{\Left{B_3}} \\
\Left{B_1;B_2} &=& \Left{B_1};\Left{B_2} \\
\Left{\biwhile{e|e'}{\P|\P'}{B}} &=&
      \whilec{e}{\Left{B}} 
\end{array}
\)
\quad
\( \begin{array}{l@{\;}c@{\;}l}   
\multicolumn{3}{l}{
  \mbox{Right projection \graybox{$\Right{B}$} is mirror image, except in these cases:}} \\
\\
\Right{\havRt{x}{\R}} &=& \havc{x} \\
\Right{\ifFour{e|e'}{B_1}{B_2}{B_3}{B_4}}  
&=& \ifc{e'}{\Right{B_1}}{\Right{B_2}} \\
\\
\\
\end{array}\)
\end{footnotesize}
\vspace*{-2ex}
\caption{Left ($\protect\Left{B}$) and right ($\protect\Right{B}$) syntactic projections of bicoms $B$}
\label{fig:project}
\end{figure}

\autoref{fig:project} defines left and right projections from bicoms to commands. 
One use of projections is to ensure that bi-if is used coherently to reason about a pair of unary if-commands.

\begin{definition} \label{def:ifFourWF}
A bicom is \dt{well-formed} just if for each
sub-bicom of the form 
$\ifFour{e|e'}{B_1}{B_2}{B_3}{B_4}$
we have that   $\: \Left{B_1} \kateq \Left{B_2}$,
$\: \Left{B_3} \kateq \Left{B_4}$,
$\: \Right{B_1} \kateq \Right{B_3}$, and
$\: \Right{B_2} \kateq \Right{B_4}$.
\end{definition}
The reader should check that these conditions hold for the pattern in
(\ref{eq:alignIf}).  In fact they hold in (\ref{eq:alignIf}) as syntactic
equalities.  But in general we allow that, say $B_1$ has a relational assertion
that is not included in $B_2$, so $\Left{B_1}$ will include an extra skip that
$\Left{B_2}$ lacks.  For example, for $B_3$ and $B_4$ in \autoref{fig:unno}, the
following are $\Left{B_3}$ and $\Left{B_4}$ which are different
(with extra $\skipc$'s corresponding to the different positions of \havRtKeyword)
but related by $\kateq$.
\[
\begin{array}{ll}
x := \vlow;\ \havc{b}; & \qquad x := \vlow;\ \havc{b};\ \skipc;\\
\keyw{while}\ b\neq 0\ \keyw{do}\ x := x+1; \havc{b}\ \keyw{done}; & \qquad
\keyw{while}\ b\neq 0\ \keyw{do}\ x := x+1; \havc{b}; \skipc\ \keyw{done}\\
\skipc; \skipc & \\
\end{array}
\]

Define the \dt{bi-left projection} \graybox{$\biLeft{B}$} by
\(\biLeft{B} \eqdef \biEmb{\Left{B}}{\skipc}
\).
The \dt{bi-right projection}, \graybox{$\biRight{B}$},
is defined in Figure~\ref{fig:biproject}.
Its purpose is to retain filters associated with havoc on the right. 
Both bi-projections preserve well-formedness.
They have useful properties including 
$\Left{(\biLeft{B})} = \Left{B}$
and 
$\Left{(\biRight{B})} \kateq \skipc$.

\begin{figure}[t]
\begin{footnotesize}
\( \begin{array}{lcl}
\biRight{\biEmb{c}{c'}} &=&  \biEmb{\skipc}{c'} \\
\biRight{\assertc{\P}} &=&  \assertc{\P} \\
\biRight{\havRt{x}{\R}} &=&  \havRt{x}{\R} 
\end{array}
\)
\hspace*{3em}
\( \begin{array}{lcl}
\biRight{B_1;B_2} &=&  \biRight{B_1};\biRight{B_2} \\
\biRight{\ifFour{e|e'}{B_1}{B_2}{B_3}{B_4}}  
&=&  
\ifFour{\True|e'}{\biRight{B_1}}{\biRight{B_2}}{\biRight{B_3}}{\biRight{B_4}} \\
\biRight{\biwhile{e|e'}{\P|\P'}{B}} &=&
\biwhile{\False|e'}{\False|\P'}{\biRight{B}} 
\end{array}
\)
\end{footnotesize}
\vspace*{-2ex}
\caption{Bi-right projection \graybox{$\protect\biRight{B}$}}\label{fig:biproject}
\end{figure}

\subsection{Semantics of bicoms and their correctness}\label{sec:semantics}

\begin{figure}[t!]
\begin{small}
\begin{mathpar}
\mprset{sep=1.4em}
\inferrule{c / s \ceval \fail}
          {\biEmb{c}{c'} / (s,s')\beval \fail}

\inferrule{c / s \ceval t \\ c' / s'\ceval \fail}
          {\biEmb{c}{c'} / (s,s')\beval \fail}

\inferrule{c / s\ceval t \\ c' / s'\ceval t'}
          {\biEmb{c}{c'},(s,s')\beval (t,t')}

\inferrule{s,s'\models \Q}
          {\assertc{\Q} / (s,s')\beval (s,s')}

\inferrule{s,s'\not\models \Q}
          {\assertc{\Q} / (s,s')\beval \fail}

\inferrule{t'=\update{s'}{x}{n}\\ s,t'\models \Q}
          {\havRt{x}{\Q} / (s,s')\beval (s,t')}

\inferrule{s\not\models e\\ s'\not\models e'}
          {W / (s,s') \beval (s,s')}

\inferrule[Left-only fail now]{s\models e \\ s,s'\models\Lrel \\ \biLeft{B} / (s,s') \beval \fail}
          {W / (s,s') \beval \fail}

\inferrule[Left-only iter]{s\models e \\ s,s'\models\Lrel \\ \biLeft{B} / (s,s') \beval (t,t')
           \\ W / (t,t') \beval \varphi}
          {W / (s,s') \beval \varphi}

\inferrule[Right-only iter]{(s \not\models e \mbox{ or } s,s'\not\models\Lrel) 
              \\ s'\models e' \\ s,s'\models\R
              \\ \biRight{B} / (s,s') \beval (t,t')
              \\ W / (t,t')\beval \varphi}
          {W / (s,s') \beval \varphi}

\inferrule{s\models e \\ s'\models e' \\
          s,s'\not\models\Lrel \\ s,s'\not\models\R \\
           B / (s,s') \beval (t,t') \\
           W / (t,t')\beval \varphi}
          {W / (s,s') \beval \varphi}

\inferrule{(s\models e \mbox{ and } s'\not\models e' \mbox{ and } s,s'\not\models\Lrel)
           \mbox{ or } 
            (s\not\models e \mbox{ and } s'\models e' \mbox{ and } s,s'\not\models\R)}
        {W / (s,s') \beval \fail}

\end{mathpar}
\end{small}
\vspace*{-2ex}
\caption{Bicom semantics (excerpts). 
Here $W$ abbreviates $\biwhile{e|e'}{\Lrel|\R}{B}$. 
(Full semantics in \autoref{fig:fullbicomSem}.) 
}
\label{fig:bicomSem}
\end{figure}

Bicoms are given an evaluation semantics in \autoref{fig:bicomSem}.  Here
$\varphi$ ranges over outcomes of two forms: either a store pair $(s,s')$ or
$\fail$.  The handling of failure in bicom semantics 
is motivated by two
considerations.  First, it must support the main theorem.  That is, the
$\forall\forall$ property of a bicom to which the filter-adequacy transformation
has been applied must imply the $\forall\exists$ property of its underlying
commands.  Second, it should be implementable by translation to commands
(including assertions), using straightforward encoding both to leverage tool
automation and to facilitate user interaction.  
Our semantics of $\biEmb{c}{c'}$ can
be implemented by translating to $c;c''$ where $c''$ is $c'$ with its variables
renamed apart from those of $c$ (Formalized in Section~\ref{sec:unary}). This means, for example, that
$\biEmb{\whilec{\True}{\skipc}}{\assertc{\False}}$ does not fail.  

Several evaluation rules are used for the semantics of the form
$\biwhile{e|e'}{\Lrel|\R}{B}$.  The rules labeled \rn{Left-only} use bi-left
projection to have the intended effect.  Following those is one of the rules for
right-only iteration; it determinizes the choice between left-only and right-only
iterations when $\Lrel$ and $\R$ (and $e$ and $e'$) hold.  This loses no
generality\footnote{Any judgment proved using \rn{eDo} with $\P,\P'$
instantiated as $\Lrel,\R$ can also be proved instantiating $\P'$ as
$\R\land\neg(\leftF{e}\lor\Lrel)$.}  and accords with a translation that uses
if-commands for the loop body, which works provided $\Lrel$ and $\R$ can be
written as boolean expressions---which is the case in all examples we have seen.
Instantiating both $\Lrel$ and $\R$ as true yields a left-first sequential
alignment.  
\changed{}{Both false is lockstep.}
It is possible for one-sided divergence to mask failure on the other
side,
e.g., $\biwhile{\True|e'}{\True|\R}{\biEmb{\skipc}{\assertc{\False}}}$.

The last rule in \autoref{fig:bicomSem} 
bakes an adequacy condition into the bi-while semantics: 
if either $e$ or $e'$ is true but none
of the other transition rules apply then the bicom fails.  (Compare the side
condition in rule \rn{eDo} in \autoref{fig:ERHL}.)  

For any $B,s,s'$, the possible results from $B/(s,s')$ are failure, normal termination, or no outcome.
The latter can happen only due to a divergent loop or blockage by $\havRt{x}{\Q}$ (in case there is no value for $x$ that makes $\Q$ true). 
The bi-while semantics uses $\biRight{\mystrut\missingArg}$ to retain such blockage for right-only iterations.
It can be proved that $\biRight{B} / (s,s') \beval (t,t')$ implies  $s=t$,
and similarly for $\biLeft{B}$.  
Bicom semantics models the effects of commands, so for example 
$B / (s,s') \beval (t,t')$ implies  
$\Left{B} / s \ceval t \mbox{ and } \Right{B} / s' \ceval t'$.
The converse may fail if $B$ includes \havRtKeyword\  (the assumption of which may not hold) 
or nontrivial loop alignment conditions or relational assertions (which can lead to failures) ---because these are discarded by $\Left{\mystrut\missingArg}$ and
$\Right{\mystrut\missingArg}$.  (See appendix Lemma~\ref{lem:eval_eval_beval}.)

Let \graybox{$\means{c}$} be the relation from stores to outcomes defined by 
$\means{c} = \{ (s,\phi) \mid c / s \ceval \phi \}$.
Let \graybox{$\means{B}$} be the relation from store pairs to outcomes defined by 
$\means{B} = \{ ((s,s'),\varphi) \mid B / (s,s') \beval \varphi \}$.
Now for semantic equivalence of bicoms $B$ and $C$ we can simply write $\means{B}=\means{C}$. For the semantic equivalences mentioned in
item (\ref{it:one}) of the summary of our technique at the end of Section~\ref{sec:overview},
we can rely on the fact that $c\kateq d$ implies $\means{c}=\means{d}$.

Finally, we define \dt{bicom correctness}.  It is about individual bicom executions but it is an $\forall\forall$ property in the sense that it implies 
$\models \Left{B}\sep\Right{B} : \rspec{\P}{\Q}$
for $B$ without $\havRtKeyword$
(see Theorem~\ref{thm:bcorrect_no_assume_adeq}).

\begin{definition}\label{def:bcorrect} The judgment $B:\bspec{\P}{\Q}$ is \dt{valid}, 
written \graybox{$\models B: \bspec{\P}{\Q}$}, 
iff for all $s,s'$ such that $s,s'\models \P$ we have
$B / (s,s') \not\beval \fail$ and moreover 
$t,t'\models \Q$ for all $t,t'$ with $B / (s,s')\beval(t,t')$.
\end{definition}

\subsection{Semantic framing and variants}\label{sec:framing}

For the filter-adequacy transformation to serve its purpose we need the $\chk$
function to use fresh variables to snapshot variants as in the example
of \autoref{fig:whyrelext:chkc1}.  They must be fresh with respect to the spec
of interest and with respect to the bicom including any variants in it.  In an
implementation, freshness is easily achieved using syntactic checks and gensym.
But a stateful gensym cannot easily be used in proofs; rather, we need to define
the transformation as a pure function in the ambient logic.  Moreover shallow
embedding precludes naive syntactic analysis to find the variables used in a
given command or bicom.  So, the main theorem assumes given a set of variables
that includes any that the bicom reads or writes.  This section formalizes
what's needed for that assumption, using semantic frame conditions and a
syntactic analysis involving variants.
We choose to represent variable sets by lists (without any requirement of ordering or uniqueness) and write $x\in vs$ to say $x$ is in the list $vs$.
\changed{}{On first reading, we suggest you jump to Section~\ref{sec:filt-adeq} which begins with the definition of $\chk$ in Section~\ref{sec:chk}.  Return here when you want to study the details that underly      Section~\ref{sec:main}.
}

\begin{definition}[semantic frames]\label{def:semframe}
For list $vs$ of variables, define relation $\agree{vs}$ on stores by
$ \graybox{$ s \agree{vs} t $} \eqdef \all{x\in vs}{s(x)=t(x)} $ (pronounced
``$s$ \dt{agrees with $t$ on $vs$}'').
For expressions, assertions, two-state expressions, and relations we define
a proposition \graybox{$\semframe{\missingArg}{vs}$}, that says the entity depends only on the variables in list $vs$.
For an assertion $p$, define
\(
\begin{array}{lcl}
\semframe{p}{vs} & \eqdef & \all{s,t}{s\agree{vs} t \imp (s\models p \iff t\models p) }
\end{array}
\).
The definitions are similar for expressions and relations.
\end{definition}

For commands, semantic framing says that the \emph{effect} only depends on the initial values of variables in $vs$. 
\qquad
\( \graybox{$\semframe{c}{vs}$} \eqdef
\begin{array}[t]{l}
(\all{s,s',t}{s\agree{vs}s' \land c/s\ceval t \imp
  \some{t'}{c/s'\ceval t' \land t\agree{vs}t'}})  \\
\land(\all{s,s'}{s\agree{vs}s' \land c/s\ceval\fail \imp c/s'\ceval\fail }) 
\end{array}\)
\\
Note that this considers the effect on $vs$.
This allows that if the command terminates normally it may have an effect on other variables, and that effect may depend on other variables.
Although the phrasing of these conditions seems asymmetric, the relation $\agree{vs}$ is symmetric. So an informal reading of $\semframe{c}{vs}$ is that from initial states $s,s'$ that agree on $vs$, $c$ has the same behaviors (modulo $\agree{vs}$).
For bicoms the definition is similar.
\[ \graybox{$\semframe{B}{vs}$} \eqdef
\begin{array}[t]{l}
(\all{s,s',t,t',u,u'}{
    \begin{array}[t]{l}
      s\agree{vs}t \land s'\agree{vs}t'\land B/(s,s')\ceval (u,u') \\
      \imp \some{v,v'}{B/(t,t')\beval (v,v') \land u\agree{vs}v \land u'\agree{vs}v'}})
    \end{array} \\
\land (\all{s,s',t,t'}{s\agree{vs}t \land s'\agree{vs}t' \land B/(s,s')\beval\fail 
  \imp B/(t,t')\beval\fail })
\end{array}\]

Until now our formal syntax for commands and bicoms omitted variants.
Henceforth we work with annotated command and bicom syntax.  
Bicoms already feature alignment conditions, which are a kind of annotation that 
accords with the proof rule \rn{eDo} in \autoref{fig:ERHL}.  Now we add variant annotations (as in Section~\ref{sec:overview}). 
\changed{Unlike alignment conditions, variants}{Variants} have no effect on command or bicom semantics; they just serve in defining the filter-adequacy transformation (Section~\ref{sec:chk}).
\changed{}{(By contrast, it is crucial for alignment conditions, assertions, and $\havRtKeyword$ to have semantics: they embody alignment and filtering out of unhelpful executions so an $\forall\forall$ property can hold.)}
The syntax of while and bi-while is henceforth 
$\whilev{e}{e_1}{c}$ and $\biwhilev{e|e'}{\P|\P'}{E}{B}$ 
where $e_1$ is an integer expression
and $E$ is a two-state integer expression.
All preceding definitions and results are applicable to the revised syntax.

\begin{figure}[t!]
\begin{small}
\[\begin{array}{ll}
B              & \biFrame(B,vs)                         \\\hline

\biEmb{c}{c'}  & \comFrame(c,vs)\land\comFrame(c',vs)   \\

\assertc{\P}   & \semframe{\P}{vs}                      \\

\havRt{x}{\P}  & x\in vs \land\semframe{\P}{vs}         \\

B_1; B_2       & \biFrame(B_1,vs) \land \biFrame(B_2,vs) \\

\ifFour{e\sep e'}{B_1}{B_2}{B_3}{B_4} 
     & \semframe{(e,e')}{vs}\land\quant{\land}{i}{1\leq i\leq 4}{\biFrame(B_i,vs)} \\

\biwhilev{e\sep e'}{\P\sep\P'}{E}{B_1} 
     & \semframe{(e,e',\P,\P',E)}{vs}\land \biFrame(B_1,vs) 
\end{array}\]
\end{small}
\vspace*{-3ex}
\caption{Defining $\biFrame$ for bicoms. Here $\semframe{(e,e',\ldots)}{vs}$ abbreviates $\semframe{e}{vs}\land\semframe{e'}{vs}\land\ldots$ There is a similar function $\comFrame$ for commands,
e.g.,
$\comFrame(x:=e,vs) = (x \in vs \land \semframe{e}{vs})$.
(See appendix \autoref{fig:cFrame}.)
}
\label{fig:frame}
\end{figure}

\autoref{fig:frame} defines a truth-valued function $\biFrame$ that we refer to as a syntactic framing check because it recurses over syntax.   
The condition $\biFrame(B,vs)$ ensures $\semframe{B}{vs}$ and moreover 
$B$ has no effect or dependence outside $vs$.  

\begin{lemma}\label{lem:frame-of-bicom-unchanged}
\mbox{If $\:\biFrame(B,vs)$, $x\notin vs$, and $\,B/(s,s')\beval(t,t')$
then $B/(s,\update{s'}{x}{n})\beval (t,\update{t'}{x}{n})$.}
\end{lemma}

Note that $\biFrame(B,vs)$ considers variant expressions in the code, even though variants do not influence the semantics.
This ensures the strong freshness needed for $\chk$ to serve its purpose,
so $\biFrame$ is the condition assumed in the main Theorem~\ref{thm:main}.
The key property is a difficult Lemma~\ref{lem:chk_frame_sem}.

\ifchange
\begingroup 
\color{red} 

\subsection{Weakest preconditions for bicoms}\label{sec:wlp}

We use the phrase ``weakest precondition'' but we are concerned with partial correctness and so use the standard name $\wlp$ that refers to weakest liberal precondition~\cite{Dijkstra76}.
For commands, define
\(\graybox{$\wlp(c,p)$} \eqdef \{ s \mid c/s\not\ceval\fail \mbox{ and }
                             t\models p \mbox{ for all $t$ such that } c/s\ceval t \} \).
This satisfies characteristic equations including 
$\wlp(\havc{x},p) = \all{x}{p}$
and $\wlp(\whilec{e}{c},p) = \gfp(F(e,c,p))$
where $F(e,c,p)(X)\eqdef (e\imp\wlp(c,X))\land(\neg e\imp p)$. 
Note that $F(e,c,p)(X)$ is monotonic in $X$ with respect to the subset  ordering on assertions, and $\gfp$ gives the greatest fixpoint.

We overload the name $\wlp$, using it for a map from bicoms and relations to relations:
\begin{equation}\label{eq:def:rwlp}
\graybox{$\wlp(B,\P)$} \eqdef \{ (s,s') \mid B/(s,s')\not\beval\fail \land
\all{t,t'}{ B/(s,s')\ceval (t,t') \imp (t,t')\models \P \} }
\end{equation}
Like wlp of commands, this has nice properties like  $\models \P\imp\wlp(B,\Q)$ iff $\, \models B:\bspec{\P}{\Q}$.
For embedded commands, $\wlp$ satisfies equations like
$\wlp(\biEmb{x:=e}{\skipc},\Q) = \subst{\Q}{x|}{e|}$.
We rely on characteristic equations for other bicom constructs,
e.g., $\wlp(\havRt{x}{\P},\Q)  =  \all{\smSep x}{(\P\imp\Q)}$.
Also $\wlp(\biwhile{e|e'}{\Lrel|\R}{B},\Q)   =  \gfp(G(e,e',\Lrel,\R,B,\Q))$ where 
\[ G(e,e',\Lrel,\R,B,\Q)(\X) \eqdef 
\begin{array}[t]{ll}
(\neg\leftF{e} \land \neg\rightF{e'} \imp \Q) \land                                    \\
(\leftF{e} \land \Lrel \imp \wlp(\biLeft{B}, \X)) \land                                \\
(\rightF{e'}\land \R \land \neg(\leftF{e}\land\Lrel) \imp \wlp(\biRight{B}, \X)) \land \\
(\leftF{e} \land \rightF{e'} \land \neg\Lrel \land \neg\R \imp \wlp(B, \X)) \land                \\
((e \agreeRel e') \lor (\leftF{e} \land \Lrel) \lor (\rightF{e'} \land \R))         
\end{array}\]
The last four conjuncts correspond with the premises in rule \rn{eDo} (in \autoref{fig:ERHL}),
and perhaps illuminate what is achieved by the evaluation semantic rules of \autoref{fig:bicomSem}.

\endgroup
\fi

\section{The filter-adequacy transformation}\label{sec:filt-adeq}

Section~\ref{sec:chk} \changed{formalizes the transformation as a function, $\chk$, on bicoms.}{
defines the transformation ($\chk$) on bicoms.}
Section~\ref{sec:main} gives the main result.

\subsection{The check functions}\label{sec:chk}

For the transformation to serve its purpose, the instrumentation added by $\chk$
should not interfere with the underlying executions unless (and until) an
assertion gets violated.  As shown in Section~\ref{sec:overview}, the
instrumentation uses fresh variables in loop bodies to snapshot the value of the
variant in order to assert that it gets decreased.  
The added variables should be fresh with respect to both the program and its specification.
As discussed in Section~\ref{sec:framing} we handle freshness by assuming a given 
frame, i.e., finite \changed{set}{list} of variables that should be avoided.
We parameterize $\chk$ and the helping function $\uChk$ by $vs$, 
which should be chosen to frame the spec and bicom or command of interest.

\begin{wrapfigure}{R}{0.57\textwidth}  \begin{small}
\(\begin{array}{ll}
c                   & \uChk(c,vs) \\\hline
\skipc              & \skipc \\
x:=e                & x:=e \\
\havc{x}            & \havc{x} \\
\assertc{p}         & \assertc{p} \\
c_1;c_2               & \uChk(c_1,vs);\uChk(c_2,vs) \\
\ifc{e}{c_1}{c_2}     & \ifc{e}{\uChk(c_1,vs)}{\uChk(c_2,vs)} \\
\whilev{e}{e_1}{c_1} \!\! & \whilev{e}{e_1}{c_2} \\
\multicolumn{2}{l}{
\qquad \begin{array}[t]{l}
   \mbox{where } c_2 \mbox{ is } x := e_1; \uChk(c_1,vs); \assertc(0\leq e_1 < x) \\
   \mbox{and } x \mbox{ is not in } vs\catenate\modVars(\uChk(c_1,vs))
    \end{array}}
\end{array}\)
\end{small}
\vspace*{-2ex}
\caption{The $\uChk$ transformation on commands}\label{fig:uChk}
\end{wrapfigure}

The \dt{unary check} function maps a command $c$ and list $vs$ of variables to a
command \graybox{$\uChk(c,vs)$} that is equivalent except that each loop is
instrumented so it asserts that the body decreases the given variant expression.
It is defined in \autoref{fig:uChk}.  In the loop case, a variable $x$ is chosen
that is fresh with respect to any instrumentation variables added to the body
$c_1$ and with respect to the given list $vs$.  This is achieved by the
condition $x$ not in $vs\catenate\modVars(\uChk(c_1,vs))$ in the case for
while.\footnote{For any $c$ we define \graybox{$\modVars(c)$} to be the list of
variables modified in $c$, namely variables that are assigned or havoc'd. The
omitted definition is straightforward recursion on syntax.}  This works because
the variables added by $\uChk$ are explicitly assigned and so caught by
$\modVars$.\footnote{In our use of $\uChk$, $vs$ will frame the command as well
as the specification of interest, so $vs$ already includes the command's
assigned variables.  Thus the catenation ($\catenate$) can create
duplicates. This is harmless because we allow duplicates in $vs$.}

For $\uChk$ to be a function one could determinize the choice of $x$, for
example by ordering variables and letting $x$ be the least variable that does
not occur in $vs\catenate\modVars(\uChk(c_1,vs))$.  In our mechanization we
instead use Hilbert's indefinite choice operator.

The definition of $\uChk$ for sequential composition is slightly delicate.  Both recursive
calls use the same list $vs$, which means the instrumentation variables added to
$c_1$ may well be the same as those added in $c_2$.  This is harmless because
the variables are initialized in each loop body (the assignment $x:=e_1$ where
$e_1$ is framed by $vs$).  One might guess to change the definition to use
$\uChk(c_2,vs\catenate\modVars(\uChk(c_1,vs))$, which prevents re-use of an
instrumentation variable for two loops.  The same idea can be used for the
branches of an if.  However, the chosen definition has the nice feature that the
given $vs$ is used unchanged in every recursive call.  This pays off well in the
proof of the main theorem and its supporting Lemma~\ref{lem:uChk_term}.

\begin{figure}[t!]
\begin{small}
\(\begin{array}{ll}
B               & \chk(B,vs) \\\hline

\biEmb{c}{c'}    & \biEmb{c}{\uChk(c',vs)} \\

\assertc{\P}     & \assertc{\P}     \\

\havRt{x}{\P}  & \assertc{(\some{\smSep x}{\P})};\havRt{x}{\P} \\ 

B_1; B_2,      & \chk(B_1,vs);\chk(B_2,vs) \\

\ifFour{e\sep e'}{B_1}{B_2}{B_3}{B_4} 
  \quad &  
\ifFour{e\sep e'}{\chk(B_1,vs)}{\chk(B_2,vs)}{\chk(B_3,vs)}{\chk(B_4,vs)} \\

\biwhilev{e\sep e'}{\P\sep\P'}{E}{B_1} & \biwhilev{e\sep e'}{\P\sep\P'}{E}{B_2} \\
     &\quad\mbox{where } B_2 \mbox{ is } 
       \begin{array}[t]{l}
       \havRt{x_1}{(\rightex{x_1}=E)}; \\
       \havRt{x_2}{(\rightex{x_2} = (\rightex{e'}\land\P'))}; \\ 
       \chk(B_1,vs); \\
       \assertc{(\rightex{x_2} \imp 0 \leq E < \rightex{x_1})}
       \end{array} \\
     & \quad\mbox{and } x_1,x_2 \mbox{ are distinct and not in } vs\catenate\modVarsR(\chk(B_1,vs)) 
\end{array}\)
\end{small}
\vspace*{-2ex}
\caption{The $\chk$ transformation on bicoms}
\label{fig:chk}
\end{figure}

The definition of \graybox{$\chk(B,vs)$} is in \autoref{fig:chk}.  
As discussed in Section~\ref{sec:overview}, loops are instrumented to assert that the right-side execution decreases the declared variant expressions.  Filtered havocs are also guarded by an existence assertion.  
The instrumentation only adds variables on the right
side, so it is convenient to define \graybox{$\modVarsR(B)$} to be the list of
variables modified on the right side of $B$, namely those in \havRtKeyword\ and
on the right side of $\biEmb{\missingArg}{\missingArg}$.
This is used for freshness in the way already discussed for $\uChk$.

In the bi-while case, the two-state expression $E$ is used as a variant. It is only
relevant for right-only iterations where it must decrease (due to changes on the
right side since any left variables will remain unchanged).  Integer variable $x_1$
snapshots the value of $E$ and boolean variable $x_2$ snapshots the truth value of the condition under which
the current iteration will be right-only, i.e., the relation
$\rightex{e'}\land\P'$.  In the definition of
$\uChk$ an assignment command can be used for the snapshot, but here $E$ and the
condition $\P'$ can depend on both states so we cannot use
assignments for $x_1$ and $x_2$.  
(In our experience they often depend only on the right, 
in which case assignment could be used.)

As with \uChk, the definition of \chk\ has the nice feature that $vs$ is
unchanged in recursive calls. This is achieved by allowing that instrumentation
variables may be re-used between conditional branches and between bicoms in
sequence, while ensuring that nested loops have distinct variables. 

A key technical result is that the bi-projections commute with $\chk$.  On the
left side this is only up to a syntactic equivalence, $\bieq$, derived from $\kateq$,
because $\Left{\mystrut\missingArg}$ replaces assertions and assumptions with
$\skipc$, hence the same for $\biLeft{\mystrut\missingArg}$.  In fact we have
$\Left{\chk(B,vs)} \kateq \Left{B}$ but the two are not identical.

\begin{lemma}
\label{lem:chk_commute} $\biLeft{\chk(B,vs)} \bieq \chk(\biLeft{B},vs)$ and 
$\biRight{\chk(B,vs)} = \chk(\biRight{B},vs)$.
\end{lemma}

\begingroup \ifchange
\color{red}
In general, $\chk(B,vs)$ acts on variables that are not in $vs$,
so $\biFrame(B,vs)$ does not imply $\biFrame(\chk(B,vs),vs)$---but 
it does imply semantic framing of $\chk(B,vs)$.
\begin{lemma}
\label{lem:chk_frame_sem} (i) $\comFrame(c,vs)$ implies $\semframe{\uChk(c,vs)}{vs}$.
(ii) $\biFrame(B,vs)$ implies $\semframe{\chk(B,vs)}{vs}$.
(iii) $\biFrame(B,vs)$ implies $\semframe{\biRight{\chk(B,vs)}}{vs}$.
\end{lemma}
Why are these true?  Both $\uChk$ and $\chk$ can alter behavior,
because the added assertions can fail. However, whether failure happens
is influenced only by the assumed relation (for $\havRtKeyword$) or the variant (for loop), and those are framed by $vs$ according to $\biFrame$.  That influence goes via snapshot variables that are outside $vs$ (and outside the snapshot variables of inner loop bodies, which must be shown to preserve their values using Lemma~\ref{lem:frame-of-bicom-unchanged}).
So the proofs require more than simple application of induction hypotheses.
A key point is that for a bi-while $B$, the instrumented body
($B_2$ in \autoref{fig:chk}) is semantically framed by $vs$.
Indeed, this motivates the definition of $\semframe{\missingArg}{vs}$.
\fi \endgroup 

\subsection{Main result}\label{sec:main}

The main result justifies item (\ref{it:three}) of the summary of our technique at the end of Section~\ref{sec:overview}, that is,  
if the bicom $\chk(B,vs)$ satisfies a spec $\bspec{\R}{\S}$
then the projections $\Left{B}$ and $\Right{B}$
satisfy $\aespec{\R}{\S}$.

\begin{restatable}{theorem}{thmmain}\label{thm:main} 
Suppose $B,vs,\R$, and $\S$ satisfy the following.
(i) $B$ is well-formed.
(ii) $\biFrame(B,vs)$.
(iii) $\semframe{\R}{vs}$ and $\semframe{\S}{vs}$.
(iv) $\models \chk(B,vs) : \bspec{\R}{\S}$.
Then $\models \Left{B} \sep \Right{B} : \aespec{\R}{\S}$.
\end{restatable}

Finding a frame $vs$ that satisfies (ii) and (iii) is, in practice, easy: just collect all free or assigned variables of $\R,\S,B$. 
This is true even for a richer programming language with dynamically allocated objects: 
the requisite snapshots are simple variables, so no heap framing is needed.

\changed{Because $\chk$ uses $\uChk$ for embedded commands,
Theorem~\ref{thm:main} rests on a similar result for $\uChk$.}{}

\begin{restatable}{lemma}{lemuChkterm}\label{lem:uChk_term}
Suppose $\comFrame(c,vs)$ and  
$\semframe{\R}{vs}$ and $\semframe{\S}{vs}$.
If $\,\models \biEmb{\skipc}{\uChk(c,vs)}: \bspec{\R}{\S}$
then $\models \skipc\sep c: \aespec{\R}{\S}$.
\end{restatable}

Detailed proofs of \changed{both}{} Lemma~\ref{lem:uChk_term} and Theorem~\ref{thm:main} are provided in the appendix.
\changed{}{Here we sketch highlights, beginning with framing.}
\begingroup \ifchange
\color{bleudefrance}
\fi
In general, $\chk(B,vs)$ acts on variables that are not in $vs$,
so $\biFrame(B,vs)$ does not imply $\biFrame(\chk(B,vs),vs)$---but 
it does imply semantic framing. \begin{lemma}
\label{lem:chk_frame_sem} (i) $\comFrame(c,vs)$ implies $\semframe{\uChk(c,vs)}{vs}$.
(ii) $\biFrame(B,vs)$ implies $\semframe{\chk(B,vs)}{vs}$.
(iii) $\biFrame(B,vs)$ implies $\semframe{\biRight{\chk(B,vs)}}{vs}$.
\end{lemma}
Why are these true?  Both $\uChk$ and $\chk$ can alter behavior,
because the added assertions can fail. However, whether failure happens
is influenced only by the assumed relation (for $\havRtKeyword$) or the variant (for loop), and those are framed by $vs$ according to $\biFrame$.  That influence goes via snapshot variables that are outside $vs$ (and outside the snapshot variables of inner loop bodies, which must be shown to preserve their values using Lemma~\ref{lem:frame-of-bicom-unchanged}).
So the proofs require more than simple application of induction hypotheses.
A key point is that for a bi-while $B$, the instrumented body
($B_2$ in \autoref{fig:chk}) is semantically framed by $vs$.
Indeed, this motivates the definition of $\semframe{\missingArg}{vs}$.

\endgroup 
\changed{}{To prove Theorem~\ref{thm:main}, the high level idea is that 
if $\chk(B,vs) : \bspec{\R}{\S}$ is valid  we can prove 
$\models \skipc\sep c: \aespec{\R}{\S}$ using rules of the logic.
To do so we must obtain loop invariants and assertions for 
intermediate control points.  For that purpose we use 
semantic weakest liberal preconditions:
\[ 
\graybox{$\wlp(B,\P)$} \eqdef \{ (s,s') \mid B/(s,s')\not\beval\fail \land
\all{t,t'}{ B/(s,s')\ceval (t,t') \imp (t,t')\models \P \} }
\]
In satisfies characteristic equations familiar from unary $\wlp$; the details are in the appendix.
}
The semantics of bi-while involves projections of the loop body,
which in general are not sub-terms of the bicom.
So some results that one might expect to prove by structural induction (with inner induction on loop evaluation) must instead go by strong induction on a measure, $\size$, 
defined so that $\size(\biLeft{B}) \leq \size(B)$ and $\size(\biRight{B}) \leq \size(B)$.

Theorem~\ref{thm:main} is proved by induction on $\size(B)$.
In the induction hypothesis the list $vs$ is fixed but $\R,\S$ are general and required to be framed by $vs$.
For each syntactic form $B$, key consequences of the assumption $\models \chk(B,vs) : \bspec{\R}{\S}$ are derived using $\wlp$ equations, 
and used to establish the premises of rules 
(\autoref{fig:ERHL} and appendix \autoref{fig:ERHLcoreaddl}) used to obtain 
$\models \Left{B} \sep \Right{B} : \aespec{\R}{\S}$.  
For the bi-while case we instantiate rule \rn{eDo} 
with invariant $\Q \eqdef \wlp(\chk(B),\S)$.
We get $\semframe{\Q}{vs}$ from the assumptions 
$\biFrame(B,vs)$ and  $\semframe{\S}{vs}$,
using Lemmas~\ref{lem:chk_frame_sem} 
(and appendix Lemma~\ref{lem:frame_sem_rwlp} for $\wlp$).

\changed{By definitions, 
$\Q$ satisfies $\models\Q\imp G(e,e',\P,\P',B,\S)(\Q)$ for $G$ 
defined in Section~\ref{sec:wlp}
(with alignment guards $\P,\P'$).
Thus $\Q$ implies each of the conjuncts of $G(\ldots)(\Q)$.
These yield, using \rn{eConseq}, the premises of rule \rn{eDo}.
}{}

With this we have justified item (\ref{it:three}) of the technique summarized at the end of Section~\ref{sec:overview}.
Before turning to the technique's evaluation in Section~\ref{sec:eval},
we address how the key hypothesis (iv) of Theorem~\ref{thm:main}
can be established by existing verification techniques.

\section{Reduction to unary verification}\label{sec:unary}

\begin{wrapfigure}{R}{0.55\textwidth}  \(\begin{array}{ll}
B             & \transB{B} \\\hline
\biEmb{c}{c'} & \transL{c}; \transR{c'} \\
\assertc{\P}  & \assertc{\enc(\P)} \\             \havRt{x}{\P} & \havc{\vrt{x}}; \assumec{\enc(\P)} \\
B;C           & \transB{B}; \transB{C} \\
\ifFour{e|e'}{B_1}{B_2}{B_3}{B_4} &  
    \keyw{if}\ \transL{e}\land\transR{e'} \ \keyw{then}\ \transB{B_1} \\
&   \keyw{else~if}\ \transL{e}\land\transR{\neg e'} \keyw{then}\ \transB{B_2} \\
&   \keyw{else~if}\ \transL{\neg e}\land\transR{e'} \keyw{then}\ \transB{B_3} \\
&   \keyw{else}\ \transB{B_4} \\
\biwhile{e|e'}{\Lrel|\R}{B} 
& 
\assertc{q} ;                            
\whilec{\transL{e}\lor\transR{e'}}{c_0} \\
\multicolumn{2}{l}{
        \qquad
         \mbox{where } c_0 \mbox{ is} 
         \begin{array}[t]{l}
         \keyw{if}\ \transL{e}\land \enc(\Lrel)\ \keyw{then}\ \transB{\biLeft{B}} \\
         \keyw{else~if}\ \transR{e'}\land \enc(\R)\ \keyw{then}\ \transB{\biRight{B}} \\
         \keyw{else}\ \transB{B} ; \\
         \assertc{q} 
         \end{array} } \\ 
\multicolumn{2}{l}{
        \qquad
         \mbox{and } q \mbox{ is} 
         \begin{array}[t]{l}
         \enc((\leftF{e}\mathord{\land}\rightF{e'})
         \mathord{\lor} (\leftF{\neg e}\mathord{\land}\rightF{\neg e'}) \\
         \qquad
         \mathord{\lor} (\leftF{e}\mathord{\land}\Lrel)
         \mathord{\lor} (\rightF{e}\mathord{\land}\R)) 
         \end{array} }
\end{array}\)
\vspace*{-1ex}
\caption{Definition of \graybox{$\transB{B}$}.}\label{fig:transB}
\end{wrapfigure}

This section sketches the translation of bicoms to commands,
whereby $\forall\exists$ and $\forall\forall$ verification can be reduced to unary verification.
We define a translation $\transB{-}$ that takes a bicom $B$ to a command
$\transB{B}$ that acts on stores suitably interpreted as encoding pairs of
stores.  The first step is to extend the grammar of commands to include
$\assumec{p}$ with its standard semantics: it never faults and if $s\models p$
then $\assumec{p}/s\ceval s$.  The results in the preceding sections, and
Theorem~\ref{thm:bcorrect_no_assume_adeq}, hold for the extended language where
$\keyw{assume}$ may occur in the $\biEmb{\missingArg}{\missingArg}$
construct---with the exception that both Lemma~\ref{lem:uChk_term} and
Theorem~\ref{thm:main} are restricted to bicoms that do not have embedded unary
$\keyw{assume}$.  

The next steps involve use of variable renaming to bijectively
encode a pair of stores as a single store.  This \changed{in turn} lets us define the
encoding of a command $c$ as acting on either the left or the right encoded
store, written $\transL{c}$ and $\transR{c}$ respectively,
\changed{}{and the encoding $\enc(\R)$ of a relation as store predicate.}
Finally, the
definition of $\transB{B}$ is in \autoref{fig:transB}.  It is a well-founded
recursion on $\size(B)$, not \changed{simply} structure of $B$, owing to the recursive
calls $\transB{\biLeft{B}}$ and $\transB{\biRight{B}}$.

\begin{theorem}\label{thm:bicomToCom_reflect_correct}
If $B$ is well-formed and 
$\transB{B} : \spec{enc(\R)}{enc(\S)}$
then $B: \bspec{\R}{\S}$.
\end{theorem}

For $\forall\exists$ verification the upshot is as follows.
To prove $c\sep c':\aespec{\P}{\Q}$, where $c,c'$ are $\keyw{assume}$-free, it suffices to find
$B$ such that $\means{\Left{B}}=\means{c}$ and $\means{\Right{B}}=\means{c'}$,
and to use whatever unary verification technique you like 
to prove $\transB{\chk(B)}:\spec{\enc(\P)}{\enc(\Q)}$.
This is by Theorems~\ref{thm:bicomToCom_reflect_correct} and~\ref{thm:main}.

The translation provides a theoretical account of the prototype relational verifier discussed in Section~\ref{sec:eval}, including its use for $\forall\forall$ verification
which is justified by the following.

\begin{theorem}\label{thm:bcorrect_no_assume_adeq}
If $B$ is well-formed and contains no $\havRtKeyword$
then $B:\bspec{\P}{\Q}$ implies
$\Left{B}\sep\Right{B}:\rspec{\P}{\Q}$.
\end{theorem}

Combining Theorem~\ref{thm:bcorrect_no_assume_adeq} with
Theorem~\ref{thm:bicomToCom_reflect_correct}, the upshot is that 
to prove $c\sep c':\rspec{\P}{\Q}$ it suffices to find
$B$ such that $\means{\Left{B}}=\means{c}$ and $\means{\Right{B}}=\means{c'}$,
and prove $\transB{B}:\spec{\enc(\P)}{\enc(\Q)}$.

Not only does the translation reflect correctness (Theorem~\ref{thm:bicomToCom_reflect_correct}), it also preserves correctness.

\begin{theorem}\label{thm:bicomToCom_preserve_correct}
If $B$ is well-formed and $B: \bspec{\R}{\S}$
then $\transB{B} : \spec{enc(\R)}{enc(\S)}$.
\end{theorem}

\section{Related work}\label{sec:related}

Relational verification is an active area of research encompassing 
\changed{secure compilation~\cite{AbateBCDGHPTT21},}{security
and compilation~\cite{LamportS21,AbateBCDGHPTT21},}
probabilistic reasoning for security and
privacy~\cite{AvanziniBDG25,GregersenAHTB24}, 
regression verification~\cite{StrichmanV2016}, 
and 
functional specification of tensor
programs~\cite{GladshteinZAAS24}, just to name a few directions. Here we focus
on works close to our goals and contributions\changed{.}{, e.g., excluding model checking~\cite{BeutnerF25}.}

Several relational Hoare logics have been proposed for $\forall\forall$ and
$k$-safety~\cite{MaillardHRM20,DosualdoFD2022} but few for $\forall\exists$,the
main thrust of our work.  One important line of work has developed \changed{}{$\forall\exists$} relational
separation logics, based on Iris~\cite{JungKJBBD18}, for refinement of
concurrent programs~\cite{TuronDB13,FruminKB18,GaherSSJDKKD22}.  Iris is
implemented in Rocq.  These logics, while expressive and powerful, are very
different from systems based on first order assertions and amenable to SMT-based
automation\changed{ in auto-active tools like Dafny, Why3, and Viper}{}.

Our focus is on sequential progams and alignment-oriented reasoning\footnote{As
opposed to taking a global view of traces~\cite{BartheEGGKM19}.}  for which one
of the first $\forall\exists$ logics is RHLE~\cite{DickersonYZD22}.  It
addresses nondeterminism in the form of havoc commands and in modular reasoning
about underspecified procedure calls.  The logic uses three judgments: in
addition to the $\forall\exists$ judgment (like ours but without failure), there
is the standard partial correctness judgment
and the less common forward under-approximation judgment\footnote{They call it
an existential judgment; it has also been called possible
correctness~\cite{Hoare:wp}, sufficient incorrectness~\cite{Ascari24}, etc.}
that can be written $\skipc\sep c:\aespec{\rightF{p}}{\rightF{q}}$ in our
notation.  RHLE is designed to cater for automated proof search.  The \changed{primary
relational}{} rules provide for unary reasoning on one side or the other in forward symbolic execution style, relying on the unary logics and unary over- and
under-approximate procedure specifications.  Alignment of loops\changed{, for which users must provide loop invariants,}{} is achieved using rules
adapted from~\citet{SousaD2016}.  
These only provide eager lockstep-until-terminate alignment~\cite{Francez83}\changed{.}{ and not data-dependent alignment.}
\changed{Procedure}{Under-approximate procedure} specs include logical variables (termed ``choice variables'') to facilitate reasoning about existential witnesses.  This can be avoided by
using aligned calls and relational specs as we show in Section~\ref{sec:eval}.

Building on RHLE, Beutner develops FEHL (forall-exist Hoare
logic)~\cite{Beutner24}.  It lacks procedures but is more general than RHLE in that it
handles $\forall^k\exists^l$ judgments (i.e., more than two-run properties).  Like RHLE, FEHL decomposes reasoning
about $\forall\exists$ properties into one-side reasoning about 
using ordinary Hoare logic and a unary forward under-approximation logic.  Like RHLE, the core rules cater for proof search.
FEHL features a rule for reasoning about loops in non-lockstep alignment: it
aligns $n$ iterations of one loop with $m$ iterations of another loop, for fixed
$n$ and $m$. There are naturally occurring examples of this~\cite{ChurchillP0A19},
and by bounding right-side iterations it avoids the need for variants.
But \changed{it does not support}{ForEx and RHLE cannot handle} more general alignments
of loops as needed for \changed{data }{state-}dependent alignments like example $c1$ in Section~\ref{sec:overview} 
(noting that $n$ there is not a fixed constant)
and those considered in~\cite{ShemerGSV19,BNNN19} (for $\forall\forall$) and in~\cite{UnnoTerauchiKoskinen21}.

Wu et al.\cite{WuWuCao25} introduce a $\forall\exists$ relational Hoare logic in
which, following the approach of~\cite{TuronDB13,FruminKB18,GaherSSJDKKD22}, the
existential (``abstract'') program appears in the pre- and post-relations with a
higher order predicate related to $wlp$.  The treatment of loop alignment is
limited. The main focus of the work is an encoding into ordinary Hoare logic,
using an assertion operator derived from $wlp$ so the $\forall\exists$ semantics
gets encoded using existential quantification over states within pre/post.
\changed{This encoding is used to verify some examples in Rocq~\cite{WuShushuICFP26} 
and also to derive Hoare logic rules that correspond to some basic rules of relational logic.}{}
The work is carried out in Rocq. It is not clear that the encoding is amenable
to automated theorem provers because application of the logic relies on use of
higher order assertions.  We use $wlp$ in this paper, but only in the
metatheory; applying our technique does not require higher order predicates.

\citet{NBN25lmcs} introduce ERHL+, a logic which uses only the single $\forall\exists$ judgment.  It features a general \changed{data}{state}-dependent
loop alignment rule attributed to~\citet{Beringer11} and a rewrite rule
attributed to~\cite{BNN16}.  
\changed{}{(See \rn{eDo} and \rn{eRewrite} in \autoref{fig:ERHL}.)}
Similar rules can also be found
in \cite{BartheGHS17}.  The rewrite rule of~\cite{NBN25lmcs} 
is based on \changed{full}{} KAT equivalence and so
can be used to derive rules including the $n,m$-fixed-iteration rule
of \citet{Beutner24}.  The focus of the paper is on the logic's completeness
with respect to alignments that can be described by alignment automata.
\changed{The authors}{They} introduce a form of
unstructured product automaton called filtered, from which we borrow the term.
\changed{}{Their filters are a bespoke automaton feature; the language does not include assumptions or syntax for product programs, and their treatment of adequacy is semantic.} 
Their filters can occur at any control point but our evaluation shows the
efficacy of putting them at the control point where nondeterminism occurs.
\changed{}{This observation, together with syntactic representation of product programs, was key to discovering the filter-adequacy transformation.}
\changed{}{Their semantic ``adequate filtering'' condition involves a decreasing measure on right-only iterations to support general state-dependent alignment.  
Our transformation embodies that idea in the form of tool-oriented annotations.}
Rules of ERHL (\autoref{fig:ERHL}) are directly
adapted from ERHL+ but the semantics of its judgments is more complicated because it
handles failure; so our soundness proofs took substantial work.
\changed{We conjecture that alignment completeness can be adapted to our logic.}{}

For alignment products, early work is found
in~\cite{BartheCK16,BartheDR11,EilersMH18} for $\forall\forall$.  Works that
represent products as automata (i.e., transition systems)
include~\citet{ChurchillP0A19}, \changed{}{\citet{UnnoTerauchiKoskinen21},}
and~\citet{ShemerGSV19}.  The latter uses
constrained Horn clauses (CHC)~\cite{Gurfinkel22} to encode the transition
system and alignment condition adequacy, using CHC solving to simultaneously
infer inductive relational invariants and state-dependent alignment conditions
for $\forall\forall$ properties.
\citet{UnnoTerauchiKoskinen21} introduce a solver for a special class of 
constraints in order to verify $\forall\exists$ properties by also inferring
well founded relations for termination.  
By contrast with our filtering approach,
\citet{UnnoTerauchiKoskinen21} also solve for Skolem functions that
witness existentials.
\citet{ItzhakySV24} show how CHC solvers can be used for $\forall\exists$, building on the
representation of witnesses as strategies in a game~\cite{BeutnerF22CAV}.  
Solvability strongly restricts both the programs and the specifications that can be verified.

\changed{}{An early work on $\forall\exists$ is \citet{BartheCK13}; their asymmetric products are similar to the alignment automata 
of \cite{ShemerGSV19,UnnoTerauchiKoskinen21,NagasamudramBN25} but without state-dependent alignment and with more restrictive adequacy conditions.
One of our benchmarks is based on their loop tiling example.}

Syntactic representation of product programs is convenient as a user interface
for auto-active tools and as a way to reduce relational verification to unary
verification in a form to which many tools apply.  For our technique, in
principle one could avoid the use of bicoms and carry out all three of the steps
at the end of Section~\ref{sec:overview} working directly with encoded product
programs.  That would spare the need for some custom notations
(Sections~\ref{sec:progs} and~\ref{sec:bicoms}).  However, custom syntax is more
convenient in an auto-active tool (Section~\ref{sec:eval}) and facilitates the
theoretical development and its exposition. Readers who are skeptical on this
point should look at the filter-adequacy transformation in \autoref{fig:chk} and
contemplate its formulation in terms of the encoding in \autoref{fig:transB}.

For direct manipulation of products as such, \citet{AntonopoulosEtal2022popl}
give an algebraic formulation called BiKAT, in reference to KAT~\cite{Kozen97})
that is shown to subsume $\forall\forall$ relational logic.  It is also used to
express $\forall\exists$ by combining a $\forall\forall$ condition with
inequations that express adequacy.  
\changed{Our embed notation $\biEmb{c}{c'}$ is inspired by that work.}{}
The idea is to use equational reasoning to
manipulate $\biEmb{c}{c'}$ into a better aligned form by inserting assumptions
(i.e., filters). \changed{But the adequacy check is via
equations that do not correspond to a standard property for which tools
exist.}{But their $\forall\exists$ adequacy check provides no verification method;
it is via equations that are not a standard property for which tools exist.}\footnote{Aside from decision procedures for KAT, which generally do not
support semantic interpretation of commands or expressive
assertions~\cite{Kozen1996,GreenbergBC22}.}  By contrast, our transformation
reduces adequacy to \changed{}{the standard}
$\forall\forall$ \changed{.}{property.}
\changed{}{The transformation is not expressible in BiKAT which 
can express neither assert nor reasoning about termination.}
The KestRel
tool~\cite{DickersonMD25} uses an algebra similar to BiKAT with e-graphs in
data-driven automatic search for good alignments for $\forall\forall$
verification (though without supporting \changed{data}{state}-dependent loop alignment). 
This exciting work is complementary to the problem addressed in this paper.

\Xchanged{BiKAT and KestRel do not address failure, but 
KAT has been extended to model failure~\cite{Mamouras17,ZhangAG22}.
However, laws of the form \( \biEmb{c}{c'};\biEmb{d}{d'} \bieq \biEmb{c;d}{c';d'} \)
found in~\cite{AntonopoulosEtal2022popl,DickersonMD25}
are invalid for programs with failure.}{}

The $\forall\forall$ logic of \citet{BNNN19} combines deductive rules with a
rule of rewriting (like \rn{eRewrite} in \autoref{fig:ERHL}) using a bicom-like
notation that includes an embed construct like our
$\biEmb{\missingArg}{\missingArg}$.  Their $\forall\forall$ property disallows
failure entirely, \changed{and the phenomena mentioned in the preceding paragraph are
avoided by using}{which leads them to use smallstep semantics with dovetailed execution of the embed
construct.}
That semantics, however, is not easily reconciled with our goal that
bicoms have a straightforward translation to unary code.  Indeed,
\citet{NagasamudramBN25} implement the logic in a prototype verifier\changed{}{, WhyRel,} that uses
the sequential translation of $\biEmb{\missingArg}{\missingArg}$
and thus with respect to failure it is verifying a
property like our Definition~\ref{def:relcorrect} for $\rspecSym$. 
\changed{}{So our bicom semantics and the theorems in Section~\ref{sec:unary} 
provide an accurate account of that verifier, unlike \cite{BNNN19}.}
(Aside from
the handling of failure, \changed{the prototype}{WhyRel} is very close to the logic\changed{.}{ of \cite{BNNN19}.}) 
The weaker treatment of failure does not seem disadvantageous in practice, since absence of
failure is a unary property that can be checked as such.  The programs handled
by~\cite{BNNN19,NagasamudramBN25} have procedures and act on dynamically
allocated \changed{}{mutable} object structures; specifications feature \changed{}{dynamic} frame conditions to support local reasoning and encapsulated state.

\citet{DardinierM24} develop a Hoare style logic for hyperproperties, in which pre- and post-condition are second order predicates.
The logic applies to a single program, not relations between different programs.
The correctness judgment says that for any set of initial states satisfying the precondition, the direct image (collecting semantics) is a set that satisfies the postcondition.  Many hyperproperties including $\forall\exists$ can be expressed owing to the ability to quantify over states in the specification. 
\citet{DardinierLM24} show that, remarkably, this approach is amenable to SMT-based automation, by encodings that track both over and under-approximations of the collecting semantics.  
\changed{The logic is also extended with explicit quantification over failure states
which is more expressive than the fixed treatment of failure in our $\forall\exists$ judgment.}{}
Their prototype, Hypra, is an auto-active tool based on Viper~\cite{MuellerSS16,EilersSSM25} in which the user provides invariants and other annotations for the single program under consideration.
Hypra has been used to verify (or refute) correctness 
of many examples from the literature, and relies on several loop rules
and an automatic selection mechanism based on the quantifier alternations of the spec.
For $\forall\exists$ properties a single loop rule is applicable
and it supports only eager lockstep alignment of iterations
(as in~\cite{Francez83,SousaD2016,DickersonYZD22}),
not conditioned on data.  
Data structures are not supported.  
A novel hint annotation is used for witnesses of existentials,
to help with nondeterministic choices.
The underlying logic~\cite{DardinierM24} has been machine checked
in Isabelle, including the $\forall\exists$ loop rule~\cite{DardinierLM24}.
Some recent works provide machine checked components of auto-active
tools \changed{}{for unary verification}~\cite{CohenJF24,ParthasarathyDBMS24,DardinierEtalPOPL25}.

\nocite{murray2023assume}

\section{Evaluation}\label{sec:eval}

\ifchange
\begingroup 
\color{red} \input{oldeval}
\endgroup
\fi 

\begingroup \ifchange
\color{bleudefrance}
\fi

To investigate the effectiveness of the filter adequacy transformation we
consider three criteria:
\begin{itemize}
\item[(EC1)] Can existing $\forall\exists$ benchmarks be verified by our technique, with the same or comparably simple invariants and alignments?
\item[(EC2)] Beyond existing benchmarks, does our technique work with dynamic data structures and procedures? Can it be effective with reasonable user effort?
\item[(EC3)] Is our technique tool agnostic?
\end{itemize}

To address these criteria we implemented a proof-of-principle prototype,
based on an existing $\forall\forall$ tool, WhyRel~\cite{NagasamudramBN25},
chosen because it provides significant solver-based automation but also supports a relatively rich class of programs including heap data structures.
As discussed in Section~\ref{sec:related}, WhyRel supports a bicom-like syntax
and in the following we use our term bicom.  The tool translates bicoms and specs to unary code and annotations in the intermediate
verification language \dt{WhyML} of the Why3 verifier~\cite{why3url}.  In turn, Why3
generates verification conditions (\dt{VC}) and dispatches them to SMT solvers.
Why3 provides an IDE in which users can apply tactics such as splitting conjunctive VCs, selecting a specific solver for a VC, applying a lemma, and introducing a case split.  Why3 also provides a command line \emph{prove mode} that can be used to automatically apply splitting, invoke one solver on all VCs, and report which VCs are valid, invalid, or result in solver timeout.

WhyRel interprets pre-post specifications as $\forall\forall$.  Our prototype, 
dubbed \dt{\WhyRelE} in the following text, 
adds a  $\forall\exists$ mode in which the filter adequacy transformation is applied after desugaring and typechecking, just before translation to WhyML.
This addresses (EC3): by source-to-source transformation 
we turn a $\forall\forall$ verifier into one for $\forall\exists$.
That said, \WhyRelE\ adds some features.
It supports variant declarations and the $\havRtKeyword$ syntax (keyword \dt{HavocR}) as opposed to pattern-matching on havoc followed by assume.
This makes it easy to check that no unjustified assumptions are used.
We also add the general 4-way bi-if construct 
including its well-formedness check (keyword \dt{If4} in the prototype).\footnote{The $\forall\forall$ benchmarks of~\cite{NagasamudramBN25} had only the same-branch aligned form (\ref{eq:alignIf}); they use sequential alignment to relate conditionals that do not follow the same branches. Example $c2$ in Section~\ref{sec:overview} shows why this is unsatisfactory for $\forall\exists$.}
To address (EC1), we also add \WhyRelE\ \dt{prove mode}: after transformation and translation to WhyML the Why3 prove mode is invoked and there is no user interaction.

\WhyRelE\ takes as input unary code together with user-provided bicoms.  It implements
projection checks $\Left{B}\kateq c$ and $\Right{B}\kateq c'$ to confirm the
user-provided bicom $B$ represents the given programs $c,c'$ of interest.
(See item (\ref{it:one}) in the summary of our technique at the end of Section~\ref{sec:overview}.)  This is
done by implementing the semantic equivalences of Def.~\ref{def:kateq} as
simplification rules.  We added one simplification to handle ForEx
benchmarks with finite loop unrollings:
$\whilec{e}{(c;\ifc{e}{c}{\skipc})} \kateq \whilec{e}{c}$.  We also added 
this simplification: $\ifc{e}{c}{c} \kateq c$.
It facilitates use of the 4-way bi-if as an
idiom to introduce case splits as annotations in support of prove mode.
By contrast, some case studies of~\cite{NagasamudramBN25} relied on Why3 IDE to
manually apply a case split tactic.  

Neither \WhyRelE\  nor the original WhyRel implements back-translation from WhyML; thus when using Why3 IDE the user sees unproved verification conditions in encoded form.  They can then update the input bicom and its annotations and retry, or in the worst case apply Why3 tactics manually. 
Of course this is not state of the art---tools like Dafny and Viper present unproved VCs to the user at source level, but that would be a major engineering effort. A practical tool could also use known techniques to infer invariants/variants, and automate search for aligned bicoms~\cite{DickersonMD25}.

\paragraph{\textbf{EC1}} 
Table~\ref{tab:Catalog} lists benchmark examples used in the evaluation.
All verified successfully. 
They are grouped by the repo from which we got them from:
ForEx~\cite{ForExRepo} (based on FEHL~\cite{Beutner24})
ORHLE~\cite{RHLErepo} (based on RHLE~\cite{DickersonYZD22}),
Veracity~\cite{VeracityURLmain,ChenFKP22},
and WhyRel~\cite{NagasamudramBN25},
further grouped in the way those repos are organized. 
Those under heading ForEx/PCSat are originally from~\citet{UnnoTerauchiKoskinen21}.
The ForEx and ORHLE examples only involve integer variables, not even arrays.\footnote{The ORHLE examples perm-inv-refinement and add3-sorted use three variables to represent a list of that length; our version uses mathematical type intList.}
These are all small synthetic programs on the order of ten lines of code, with a few slightly larger.
The names in the table are close to filenames in the respective repos to facilitate comparison.\footnote{Some are misnomers: examples with GNI in their
name are for what is usually called possibilistic noninterference whereas ``generalized noninterference'' in the literature usually refers to a 3-run $\forall\forall\exists$ property~\cite{ClarksonSchneiderHyper10}.}

We verified all 24 benchmarks in the ForEx repo that are in scope, meaning they are 2-run properties. 
Some of these are verified in ForEx with $2$-safety ($\forall\forall$) specs,
but the programs are havoc-free and terminating, so the corresponding $\forall\exists$ property holds and we verify them as such.
This includes some using loop unrolling to give the alignments used by ForEx (as discussed in Section~\ref{sec:related}).  
The ForEx tool searches for alignments, within its limited repertoire,
and with user-provided hints about loop unrolling
designated by its custom pragma ``asynchronous''.
It searches for invariants, using the Houdini algorithm~\cite{FlanaganL01},
with another pragma ``numC'' to bound the search.

We verified all 15 benchmarks in the ORHLE repo that are in scope,
i.e., omitting the 3- and 4-run properties.
Some involve right-side calls that require under-approximate unary procedure specs, a feature not currently supported by \WhyRelE,
but the bodies of those procedures are very simple so they can be soundly inlined in the form $\assertc{p};\havc{x};\assumec{q}$ where $x$ is the frame (modifiable variables or result), $p$ the precondition and $q$ the post.  This admits use of a filter and since we avoid ever using raw $\assumec{}$ our inlined versions look like $\assertc{p};\havRt{x}{(q\land\R)}$ where $\R$ is a chosen filter.
These are noted in Table~\ref{tab:Catalog} as ``HavocR for under-approx call''.
The ORHLE tool requires user-provided invariants and variants
but it searches for loop alignments and instantiations of the ``choice variables'' for witnessing existentials
in under-approximate calls. 

The ForEx and ORHLE benchmarks include examples from prior 
work including PCsat~\cite{UnnoTerauchiKoskinen21} and HyPa~\cite{hypa25}. 
There is a lot of overlap with the benchmarks of 
Itzhaky et al~\cite{ItzhakySV24}, 
from which we also took an example with multiple loops and a static array.
Hypra~\cite{DardinierLM24} was also evaluated against some self-relation benchmarks from ORHLE, HyPa, PCsat (and Descartes~\cite{SousaD2016} for $\forall\forall$).
We add from PCsat the possibilistic noninterference example
featured in~\cite{UnnoTerauchiKoskinen21} and discussed in Section~\ref{sec:overview}, which has divergent loops.

In Table~\ref{tab:Catalog} the \textbf{Spec} column indicates what kind of property is proved.  
\dt{FR} is functional refinement, i.e., trace refinement (for pre-post traces).
\dt{Data ref't} is refinement modulo change in data representation.
The information flow properties relate a program to itself:
\dt{NI} is possibilistic noninterference, and 
\dt{Delim} is delimited release (declassification).

The \textbf{Time} column indicates verification time (in seconds) using \WhyRelE\ in prove mode: including time of
the filter-adequacy transformation, the translation to WhyML,
VC generation and solving by Why3's prove mode.\footnote{Evaluation was done on 2023 MacbookPro, M2 CPU, 16G RAM.}
\WhyRelE\ inherits WhyRel's rich state model including the heap and typed objects,  
which necessitates VCs not needed in the other tools
(and which are visible to user when in IDE mode).
Our times in prove mode are comparable to those reported for the other tools~\cite{DickersonYZD22,Beutner24,ItzhakySV24,DardinierLM24}
but keep in mind the tools differ in what is automated.

The \textbf{Loop align} column indicates what kind of alignment is used,
or NA if there are no loops.

All the tools require the user to provide specs.  For the ForEx and ORHLE benchmarks the specs are all short, roughly one or two lines per precondition or postcondition.  There is minor variation in size. 
In \WhyRelE\ we tend to use multiple keyword-declared lines where some tools have a single conjunctive formula.  The other tools need additional syntax for differing numbers of $\forall$ and $\exists$ runs, which we do not need.  On the other hand, catering for heap data structures our prototype requires and verifies frame conditions (based on dynamic frames~\cite{Leino10,BNNN19}) which the other tools do not.

As a measure of user effort we focus on annotations in the code,
because as noted the specs are similar in size across tools. 
The \textbf{Anno} column gives the count of annotations for \WhyRelE: unary and relational assertions, invariants, variants, loop frame conditions, HavocR, and the biUpdate used to indicate a relation between freshly allocated references.  The prototype can report fine grained statistics but we just sum them all here.
The \textbf{Orig} column counts the non-spec annotations for the ForEx or ORHLE 
version of the example (ForEx: asynchronous and numC, ORHLE: variant, invariant, and choiceVars). 

We did not encounter any cases that posed difficulties for our technique.
With respect to the original sources we used very similar, and simple, loop invariants. 
For example in \autoref{fig:unno} the loop invariant in $B_2$ is 
$\leftex{x}-\rightex{x} = \rightex{b}$ (Paper\_GNI\_Example).

\WhyRelE\ requires the user to provide the alignment in the form of a bicom.  For these benchmarks simple heuristics suffice: lockstep alignment of similar syntax and otherwise left-first sequential.
In ongoing work we automated these transformations and have encouraging preliminary results using  modern frontier models find alignments 
but this is beyond the scope of this paper.

One issue to evaluate is the sufficiency of the havoc-assume pattern 
as opposed to custom hints etc.: 
simplicity of the filter conditions and ease of finding suitable alignments to
enable use of filters. The havoc-filters used in our verifications are simple
equations, similar to the user-provided hints in the other works.  The asserted
existentials verified automatically, except for a few cases where a manual tactic
was used to give the instantiation.
There can be a tradeoff between weak filters that are sufficient and perhaps intuitive
for the example and precise ones for which solvers easily find witnesses.

The Hypra tool~\cite{DardinierLM24} features a custom hint mechanism 
which associates havocs with suggested existential witnesses.  
Pragmatically our filters are used similarly to Hypra's hints, 
though the latter are in the context of a bespoke hyperproperty verification system rather than a source-level transformation atop a $\forall\forall$ system that can relate different programs.

\begin{table}[h!]
\centering
\resizebox{\textwidth}{!}{
\begin{tabular}{|l|c|c|c|c|c|c|c|}
\hline
\textbf{\#} & \textbf{Example} & \textbf{Spec} & \textbf{Loop Align} & \textbf{Time} & \textbf{Anno} & \textbf{Orig} & \textbf{Remarks}  \\ \hline
\hline
\multicolumn{8}{|c|}{\textbf{ForEx/HyPa}} \\ \hline
1 & Asynch\_Opt & FR & Lockstep & 0.6 & 3 & 0 & HavocR  \\ \hline
2 & Compiler\_Opt & FR & Lockstep & 1.0 & 3 & 0 & HavocR  \\ \hline
3 & Compiler\_Opt2 & FR & Lockstep & 1.0 & 4 & 0 & HavocR  \\ \hline
4 & Counter\_Diff & Custom & Lockstep & 0.8 & 4 & 0 &  HavocR  \\ \hline
5 & Non\_Det\_Add & FR & Lockstep & 0.8 & 3 & 0 & HavocR  \\ \hline
6 & P1\_GNI & NI & NA & 0.6 & 1 & 0 & high variable in low assignment, HavocR  \\ \hline
7 & P2\_GNI & NI & NA & 0.6 & 1 & 0 & HavocR  \\ \hline
8 & P3\_GNI & NI & NA & 6.0 & 1 & 0 & HavocR  \\ \hline
9 & P4\_GNI & NI & NA & 7.9 & 5 & 0 & If4, HavocR  \\ \hline
10 & Paper\_Example\_Fig3 & FR & Lockstep & 0.9 & 2 & 0 & HavocR  \\ \hline
11 & Refine & Custom & Sequential & 42.1 & 33 & 0 & If4, HavocR, WhileR, If4 case split, Returns object  \\ \hline
12 & Refine2 & Custom & Sequential & 36.9 & 21 & 0 & If4, HavocR, WhileR, If4 case split, Returns Object  \\ \hline
13 & Smaller & Custom & Lockstep & 6.6 & 3 & 0 & HavocR   \\ \hline
\hline
\multicolumn{8}{|c|}{\textbf{ForEx/K\_Safety (verified as $\forall\exists$)}} \\ \hline
14 & Coll\_Item\_Sym & Custom & NA & 4.7 & 0 & 0 & global vars (keyword ``public'' in code)  \\ \hline
15 & Counter\_Det & FR & Lockstep & 0.7 & 2 & 0 & global vars  \\ \hline
16 & Double\_Square\_NI & FR & 2:1 & 1.6 & 2 & 3 & loop unrolling for $m:n$ alignment \\ \hline
17 & Fig\_2 & FR & 2:1 & 1.1 & 2 & 2 & loop unrolling for $m:n$ alignment \\ \hline
\hline
\multicolumn{8}{|c|}{\textbf{ForEx/PCSat} (but 18 
is directly from~\cite{UnnoTerauchiKoskinen21})} \\ \hline
18 & Paper\_GNI\_Example & NI & Seql\&Lockst & 2.0 & 17 & NA & If4, HavocR, WhileR,  while trues (see \autoref{fig:unno}) \\ \hline
19 & Ti\_GNI\_HFF & NI & Lockstep & 6.0 & 3 & 0 & HavocR  \\ \hline
20 & Ti\_GNI\_HFT & NI & Sequential & 0.6 & 2 & 1 & HavocR  \\ \hline
21 & Ti\_GNI\_HTT & NI & NA & 0.6 & 1 & 0 & HavocR  \\ \hline
22 & Ts\_GNI\_HFF & NI & Lockstep & 6.0 & 3 & 0 & HavocR, same as Ti\_GNI\_HFF  \\ \hline
23 & Ts\_GNI\_HFT & NI & Sequential & 0.7 & 3 & 2 & HavocR, while trues  \\ \hline
24 & Ts\_GNI\_HTT & NI & Sequential & 0.6 & 2 & 0 & HavocR, while trues  \\ \hline
\hline
\multicolumn{8}{|c|}{\textbf{RHLE/API\_Refinement}} \\ \hline
25 & Add3\_Sorted & FR & NA & 0.7 & 6(4) & 3 & HavocR for under-approx call, Havocs intList  \\ \hline
26 & Conditional\_Refinement & FR & NA & 0.6 & 2 & 1 & HavocR for under-approx call  \\ \hline
27 & Loop\_Refinement & FR & Lockstep & 0.6 & 3 & 5 & HavocR for under-approx call  \\ \hline
28 & Loop\_Refinement2 & FR & Lockstep & 0.6 & 3 & 5 & HavocR for under-approx call  \\ \hline
29 & Perm\_Inv\_Refinement & FR & NA & 0.8 & 4(5) & 3 & HavocR for under-approx call, Havocs intList  \\ \hline
30 & Simple\_Refinement & FR & NA & 0.4 & 2 & 1 & HavocR for under-approx call  \\ \hline
\hline
\multicolumn{8}{|c|}{\textbf{RHLE/Delimited\_Release}} \\ \hline
31 & Parity & Delim & NA & 0.7 & 0 & 0 & -  \\ \hline
32 & Parity2 & Delim & NA & 0.7 & 0 & 0 & -  \\ \hline
33 & Parity\_Fun & Delim & NA & 0.6 & 2 & 0 & HavocR for under-approx call  \\ \hline
34 & Wallet & Delim & NA & 0.6 & 0 & 0 & -  \\ \hline
\hline
\multicolumn{8}{|c|}{\textbf{RHLE/GNI}} \\ \hline
35 & Denning1 & NI & Lockstep & 1.4 & 3 & 2 & global vars (we fixed div-by-zero bug) \\ \hline
36 & Nondet\_Nonleak & NI & NA & 0.7 & 2 & 1 & HavocR for under-approx call  \\ \hline
37 & Nondet\_Nonleak2 & NI & NA & 5.3 & 6 & 1 & HavocR for under-approx call, If4 case split  \\ \hline
38 & Simple\_Nonleak & NI & NA & 0.4 & 0 & 0 & -  \\ \hline
\hline
\multicolumn{8}{|c|}{\textbf{RHLE/Param\_Usage}} \\ \hline
39 & Det\_Unused & Custom & NA & 0.3 & 0 & 0 & -  \\ \hline
\hline
\multicolumn{8}{|c|}{\textbf{WhyRel/all\_all (heap data structures, adapted to $\forall\exists$)}} \\ \hline
40 & tabulate & FR & Lockstep & 3284 * & 36(2) & NA &  allocates list nodes, synced meth call [size 24], $\forall\forall$ anno: 36(2) \\ \hline
41 & tabulate\_sampling & FR & Lockstep & 550 * & 42(9) & NA & HavocR [size 58]\\ \hline
42 & sumpub & NI & Data dep & 7.9 * & 7(7) & NA &  allocates list nodes, synced meth call [size 22], $\forall\forall$ anno: 6(7)   \\ \hline
43 & stack & Data ref't & Lockstep & 2546 * & 13(5) & NA & verifies representation independence [size 47], $\forall\forall$ anno: 13(5) \\ \hline
44 & tiling & Data ref't & Lockstep & 87.7 * & 8(3) & NA & nums for first of two versions [size 15], $\forall\forall$ anno: 8(2) (Fig.~3 of~\cite{BartheCK13}) \\ \hline
\hline
\multicolumn{8}{|c|}{\textbf{Veracity}} \\ \hline
45 & Fizzbuzz & Refinement & Lockstep & 12.6  & 25 & NA & commutativity of two array loops 
[size 42], $\forall\forall$ anno: 23 \\ \hline
46 & Simple\_IO & Refinement & Seql\&Lockst & 4.4  & 16 & NA & commutativity of two array loops, WhileR 
[size 22], $\forall\forall$ anno: 14 \\ \hline 
47 & Simple\_Vector & Refinement & Seql\&Lockst & 346 * & 25 & NA & commutativity of two array loops, WhileR
[size 23], $\forall\forall$ anno: 22 \\ \hline
\hline
\multicolumn{8}{|c|}{\textbf{Additional}} \\ \hline
48 & Itzhaky\_Example & Custom & Seql\&Lockst & 7.4 & 15 & NA & If4 w/diff loop align, HavocR [size 63] \\ \hline
49 & Conditional\_Loop & FR & Data dep & 2.4 & 6 & NA & HavocR (see \autoref{fig:whyrelext:cond-align})  \\ \hline
50 & Havoc\_Test & FR & Sequential & 0.8 & 5 & NA & HavocR, WhileR  \\ \hline
51 & CCR\_fig\_1 & Data ref't & NA & 9.1 & 1(3) & NA & HavocR, synced meth call (Fig.1 of~\cite{Song23})  \\ \hline
52 & hiccupSum & Refinement & Data dep & 6.6 & 5 & NA & HavocR  \\ \hline
53 & lowError & Custom & NA   & 3.4 & 0 & NA & 
fault-sensitive noninterference 
(lowError in Fig.~3 of~\cite{DardinierLM24}) \\ \hline

\end{tabular} }
\caption{Verified $\forall\exists$ examples.    
For Time, * indicates time from a saved session file
instead of prove mode.
For Anno, $n(m)$ indicates $m$ auxiliary definitions/lemmas, if any.
}
\label{tab:Catalog}
\end{table}

\paragraph{\textbf{EC2}} 
We made several additional benchmarks involving unbounded data structures, dynamic allocation, and relational procedure calls. 
Five benchmarks adapted from the $\forall\forall$ examples of~\cite{NagasamudramBN25} traverse
and dynamically allocate object structures.  These have aligned
procedure calls verified using a relational specification, and allocations are
also aligned.  
The relational assertion language and frame conditions of~\cite{NagasamudramBN25} address fine grained relations between pointer structures.  
For example, for tabulate we use per-node agreement of values in two linked lists.  
In addition to having more complex data these examples are a little larger than those for (EC1), around 40 lines of code.
In the table, ``[size N]'' is the number of nodes in the bicom's AST, leaving aside annotations. That's roughly proportional to lines of code, and convenient for the prototype to compute.  
Under Anno, $n(m)$ indicates $m$ auxiliary definitions/lemmas and $n$ code annotations.

We did multiple state-dependent loop alignments including with
conditions dependent on heap objects
and including for a program aligned with itself (sumpub).
We also verified an example adapted
from~\cite{Song23}, a refinement between two representations of a key-value map;
this was done using aligned calls and relational procedure specs.
(Although procedures are not in our formalization, the rule for aligned calls is sound for $\forall\exists$.)
For some of the verifications we did a few user invocations of tactics in Why3, for lemmas about the data structures; these were easily adapted from the existing
$\forall\forall$ verification of similar examples~\cite{NagasamudramBN25}.
These cases are marked by * in the Time column, where the time is given by Why3's session replay instead of prove mode.  
Some of these have very long prover times, because there are many VCs about heap framing and typing and high prover timeouts were used to lessen the need for additional annotations or user guidance.  
The authors of~\cite{NagasamudramBN25} say the VC encodings leave room for improvement to better leverage solvers.
For comparison the remarks give annotation counts for some $\forall\forall$ versions; the times are also comparable.
The time difference between tabulate and tabulate\_sampling is because in the latter session tactics were used
to clear irrelevent hypotheses in VCs; more user work, less solver work.

The remarks for a few examples say ``If4 case split''. These examples were
originally verified with some user tactics to help solvers find witnesses for the existentials introduced by $\chk$.  They were rewritten using If4 to express the cases as code annotations which eliminated the need for tactics. 

As an additional source of examples involving data structures we verified 
three commutativity problems, i.e., equivalence (mutual refinement) of $c;d$ with $d;c$ for some $c,d$ involving loops.  These came from the Veracity project~\cite{ChenFKP22,VeracityURLmain}, specifically ones that the current version of their tool is unable to verify according to the authors.
These typify situations where alignment is not a panacea: if actions take place in different order, they cannot be aligned to maintain simple equalities.
We used alignments of the form 
$\biEmb{c}{\skipc};B;\biEmb{\skipc}{c}$ where $B$ is a lockstep alignment of 
$\biEmb{d}{d}$.
If $d$ is independent from $c$ then $B$ can use very simple invariants.  
But a functional description of $c$ is needed in order to ``remember''
the effect of $c$ on the left to connect it with the right $c$.  
And functional description means strong loop invariants are needed
with quantification over the data structure.  

The prior works, and ours, are research prototypes with inessential limitations.  Their evaluations do not include user effort to find sufficient annotations and neither do we.  Anecdotally we can say the following.  For the prior benchmarks an effective alignment is easily seen along with the simple invariants etc.\ that suffice for its verification.  For examples involving nontrivial data structures, stronger unary invariants are sometimes needed. The effort to find those is comparable to doing so using familiar tools like Dafny and Why3.  
As the authors of a recent Viper project say,
``The overhead of 3.6 lines of annotation per line of code is typical for SMT-based
deductive verification and would be substantially higher for verification using an interactive theorem prover.''~\cite{PereiraKGLSWE0B25}.
This includes specifications, which we leave out of our annotation counts.

\paragraph{\textbf{EC3}} 
To evaluate that the technique can be used with other verifiers  
we manually wrote alignment products in Boogie~\cite{BoogieRepo}
for all of the benchmarks except the ones involving heap data structures.
We manually inserted the filter-adequacy checks.
We used essentially the same alignments and invariants as with \WhyRelE.
We used the default Boogie configuration with z3 as solver.
All the benchmarks verified fast.
In two cases where the filter equates the havoc'd variable with a nontrivial expression (Paper\_GNI\_Example and Itzhaky\_Example),
we needed to add a trigger pattern~\cite{BugariuTM21} for the existential quantifier, a common occurrence in verifications using Boogie.

\endgroup 

\section{Conclusion}\label{sec:concl}
  
We introduced a source level transformation \changed{}{on alignment product programs (bicoms)} which reduces verification of
$\forall\exists$ properties of imperative programs to $\forall\forall$.  The
soundness theorem connects the validity of a $\forall\forall$ property with
validity of the desired $\forall\exists$ property.  It is applicable
regardless of how the $\forall\forall$ property is verified.  For that matter,
one could use testing of the transformed bicom as evidence of the
$\forall\forall$ property which would be evidence of the $\forall\exists$
property. (Direct testing of $\forall\exists$ is
problematic~\cite{CorrensonNFW24}.)
The transformation works on bicoms, a form of product program that enables the
user to express intended alignments and intermediate assertions for 
verification, and which can be translated to unary code to leverage existing tools.

\paragraph{Statement on AI usage}

Aside from generating some tikz code for \autoref{fig:whyrelext:cond-align}
no generative AI was used in the writing of the document. 
Some generative AI helped with the evaluation and a
few small parts of the Rocq mechanization.

\bibliographystyle{ACM-Reference-Format}

\bibliography{biblio}

\newpage

\appendix

\section{Appendix: glossary of key notations}\label{sec:app:gloss}

\ifchange
\bigskip
\Xchanged{NOTE: APPENDIX IS NOT CHANGEMARKED}{} 
\bigskip
\fi

\begin{tabular}{llll}
notation & types involved  & purpose or meaning & where \\\hline
Semantic  \\\hline
$c/s\ceval\phi$ & com $c$, store $s$ & $c$ evaluates to outcome $\phi$ & \\ 
$B/(s,s')\beval\varphi$ & bicom $B$, stores $s,s'$ & $B$ evaluates to outcome $\varphi$ & \autoref{fig:bicomSem} \\ 
$c:\spec{p}{q}$ & judgment & partial correctness, i.e., $\{p\}c\{q\}$ \\
$c\sep c':\rspec{\R}{\S}$ & judgment & $\forall\forall$ relation between $c$ and $c'$ & Def.~\ref{def:relcorrect} \\
$c\sep c':\aespec{\R}{\S}$ & judgment & $\forall\exists$ relation between $c$ and $c'$ & Def.~\ref{def:relcorrect} \\
$B:\bspec{\R}{\S}$ & judgment & correctness of bicom $B$ & Def.~\ref{def:bcorrect} \\
$\means{c}$ & relates stores to outcomes & semantics of $c$ & before Def.~\ref{def:bcorrect} \\
$\means{B}$ & relates store pairs to outcomes & semantics of $B$ & before Def.~\ref{def:bcorrect} \\
$\wlp(B)$ & maps relations to relations & weakest liberal preconditions & Appendix~\ref{sec:wlp} \\
$\semframe{thing}{vs}$ & property; $vs$ is list of vars & $thing$ depends only on vars in $vs$ & Def.~\ref{def:semframe} \\
$\enc$ & maps stores to stores & encode store pairs as stores & Appendix~\ref{sec:app:unary} \\
$\enc$ & maps relations to relations  & encode rels as unary predicates & Appendix~\ref{sec:app:unary} \\
\hline

Syntactic \\\hline
$\chk$ & maps bicoms to bicoms & the filter-adequacy transformation & \autoref{fig:chk} \\
$\uChk$ & maps commands to commands & helper for $\chk$ & \autoref{fig:uChk} \\
$\Right{\missingArg}$ & maps bicoms to coms & right syntactic projection & \autoref{fig:project} \\
$\biRight{-}$ & maps bicoms to bicoms & right syntactic bi-projection & \autoref{fig:biproject} \\
$\transB{-}$ & maps bicoms to coms & encode bicom semantics as unary & \autoref{fig:transB} \\
$\transR{-}$ & maps coms to coms & com acting on right of encoded state & 
\autoref{fig:transR} \\ 
$\biFrame(B,vs)$ & bicom B, variable list $vs$ & $B$ is syntactically framed by $vs$  & \autoref{fig:frame} \\ 
$\comFrame(c,vs)$ & com c, variable list $vs$ & $c$ is syntactically framed by $vs$  & \autoref{fig:cFrame} \\ 
$\kateq$ & relates commands & kat equivalence & Def.~\ref{def:kateq} \\

\end{tabular}

\section{Appendix: brief guide to Rocq development}\label{sec:app:rocq}

This section provides a brief guide to the Rocq development.

As anonymized supplementary material we submit the following files
which were developed using Rocq 8.20.1.  
\begin{itemize}

\item\texttt{Comsem.v} gives syntax and semantics of commands and specs.
A comment at the beginning of the file describes minor differences from the paper
and the mapping of definitions/results in the paper to names in Rocq files.  

\item\texttt{LogicAE.v} gives the $\forall\exists$ logic and proves Theorem~\ref{thm:soundness};
it depends on \texttt{Comsem.v}.

\item\texttt{Bicom.v} gives syntax and semantics of bicoms, as well as results 
about framing and wlp; it depends on \texttt{Comsem.v}.

\item\texttt{FiltAdeq.v} gives the filter-adequacy transformation ($\uChk$ and $\chk$) and proves the main Theorem~\ref{thm:main}
(along with supporting results like Lemmas~\ref{lem:chk_commute}
and~\ref{lem:chk_frame_sem}); it depends on
\texttt{Bicom.v}, \texttt{LogicAE.v}, and \texttt{Comsem.v}.

\item\texttt{Transl.v} gives the unary encoding and Theorems~\ref{thm:bicomToCom_reflect_correct} and~\ref{thm:bicomToCom_preserve_correct};
it depends on \texttt{Bicom.v} and \texttt{Comsem.v}.

\item\texttt{AAbicom.v} gives Theorem~\ref{thm:bcorrect_no_assume_adeq};
it depends on \texttt{Bicom.v} and \texttt{Comsem.v}.

\item\texttt{LogicAA.v} gives some $\forall\forall$ rules; they depend on 
\texttt{Comsem.v}.  A ghost elimination rule is included, taking advantage of framing (Comsem, Bicom)
to improve on the rule used in~\cite{NagasamudramLMCS25} to prove alignment completeness.
There is a $\forall\exists$ ghost elimination rule in \texttt{LogicAE.v}.  
This bonus file is not needed for the paper's results.
\end{itemize}

\section{Appendix for Section~\ref{sec:progs}: Proof rules for ERHL}\label{sec:app:ERHL}

\begin{figure}[t]
\begin{footnotesize}
\begin{mathpar}

\inferrule[eAsgnSkip]{}{
x:=e \sep \skipc  : \aespec{\subst{\R}{x|}{e|}}{\R}
}

\inferrule[eSkipAsgn]{}{
\skipc\sep x:=e : \aespec{\subst{\R}{|x}{|e}}{\R}
}

\inferrule[eAsrtSkip]{}{ 
  \assertc{q}\sep\skipc : \aespec{\R\land\leftF{q}}{\R}
}

\inferrule[eSkipAsrt]{}{ 
  \skipc\sep\assertc{q}: \aespec{\R\land\rightF{q}}{\R}
}

\inferrule[eSkip]{}{
\skipc \sep \skipc : \aespec{\R}{\R}
}

\inferrule[eSeq]{
  c\sep c' : \aespec{\P}{\Q} \\
  d\sep d' : \aespec{\Q}{\R}
}{
  c;d \Sep c';d' : \aespec{\P}{\R}
}

\inferrule[eIf4]{
  c\sep c' : \aespec{\Q\land \leftF{e}\land\rightF{e'}}{\R} \\
  c\sep d' : \aespec{\Q\land \leftF{e}\land\neg\rightF{e'}}{\R} \\
  d\sep c' : \aespec{\Q\land \neg\leftF{e}\land\rightF{e'}}{\R} \\
  d\sep d' : \aespec{\Q\land \neg\leftF{e}\land\neg\rightF{e'}}{\R} 
}{
  \ifc{e}{c}{d} \Sep \ifc{e'}{c'}{d'} : \aespec{\Q}{\R}
}

\inferrule[eConseq]{
  \P\imp \R \\
  c\sep c' : \aespec{\R}{\S} \\
  \S \imp \Q 
}{
  c\sep c' : \aespec{\P}{\Q} \\
}

\inferrule[eDisj]{
  c\sep c' : \aespec{\P}{\R} \\
  c\sep c' : \aespec{\Q}{\R} \\
}{
  c\sep c' : \aespec{\P\lor\Q}{\R} \\
}

\inferrule[eFalse]{}{
  c\sep c' : \aespec{\False}{\R} 
}

\inferrule[eGhost]{
  c \sep c' : \aespec{\P}{\Q} \\
  g\notin vs \\
  \semframe{vs}{\Q} \\
  \comFrameGhost(c,vs,g) \\
  \comFrameGhost(c',vs,g) 
}{ 
   \eraseGhost(c,g) \sep \eraseGhost(c',g) :\aespec{\P}{\Q}
}

\end{mathpar}
\end{footnotesize}
\vspace*{-2ex}
\caption{ERHL: core rules for the $\aespecSym$ judgment
in addition to \autoref{fig:ERHL}}
\label{fig:ERHLcoreaddl}
\end{figure}

\autoref{fig:ERHLcoreaddl} gives core rules of ERHL not included in Figure~\ref{fig:ERHL}. 
The proof of soundness for \rn{eDo} is particularly intricate.  For both the failure and success cases, the proof is by induction on the total number of iterations of the two loops, 
with a key lemma to account for the relative progress of right-only iterations as a consequence of the variant decrease condition in the right-only premise.

Rules \rn{eGhost} and \rn{eDisj} are not needed for the proof of Theorem~\ref{thm:main} but are included because the analogous rules are used in~\cite{NagasamudramLMCS25} to prove alignment completeness.
We conjecture alignment completeness holds for ERHL too.
In \rn{eGhost}, $\comFrameGhost(c,vs,g)$ is like $\comFrame(c,vs)$ but with an additional check that $g$ is an auxiliary variable in the sense that it does not occur in assertions or in any expressions except those in assignments/havocs of $g$.
The function $\eraseGhost$ changes all such assignments/havocs to skips.
The proof of \rn{eGhost} relies on a notion of equivalence relative to the commands' effect on a set of variables, which generalizes the framing theory in Section~\ref{sec:framing}, so in particular $\semframe{c}{vs}$ is the same as $c$ being self-equivalent modulo $vs$.
For details see the Rocq mechanization.

\begin{figure}[t!]
\begin{footnotesize}
\begin{mathpar}

\inferrule[eSkipDo]{ 
\mbox{$\skipc\sep c : 
       \aespec{\rightF{e}\land\R\land (n=E)}{\Q\land (0\leq E < n)}$ for all $n\in\Z$} \\
  }{ \skipc \Sep \whilec{e}{c} : 
    \aespec{\Q}{\Q\land \neg\rightF{e}}
}

\inferrule[eSkipIf]{ 
  \skipc\sep c : \aespec{\rightF{e}\land\R}{\S} \\
  \skipc\sep d : \aespec{\neg\rightF{e}\land\R}{\S} 
  }{ \skipc \Sep \ifc{e}{c}{d} : \aespec{\R}{\S}
}

\inferrule[eDoX]{
  c\sep \skipc : \aespec{\I\land \leftF{e}\land\P }{\I} \\ 
  c\sep c' : \aespec{\I\land \leftF{e}\land\rightF{e'} \land \neg\P \land \neg\P'}{\I} \\
  \mbox{$\skipc\sep c' : \aespec{\I\land \rightF{e'}\land\P'\land \neg(\leftF{e}\land\P)\land(n=E) }{\I\land (0\leq E < n)}$ for all $n\in\Z$}
 \\ 
   \I\imp (\leftex{e} = \rightex{e'} \lor (\P \land \leftF{e})\lor (\P' \land \rightF{e'})) \\
   \R\imp\I \\ \I\imp\S 
  }{
   \whilec{e}{c} \Sep \whilec{e'}{c'} : \aespec{\R}{\S\land \neg\leftF{e}\land\neg\rightF{e'}} 
}
\end{mathpar}
\end{footnotesize}
\vspace*{-3ex}
\caption{Derived rules for $\aespecSym$.}
\label{fig:ERHLadditional}
\end{figure}

\autoref{fig:ERHLadditional} gives additional rules which are derivable from
those in Figures~\ref{fig:ERHL} and~\ref{fig:ERHLcoreaddl}.

To derive \rn{eSkipDo}, apply \rn{eDo}
to prove $\whilec{\False}{\skipc} \Sep \whilec{e}{c}
: \aespec{\Q}{\Q\land \neg\leftF{\False}\land\neg\rightF{e'}}$, then
use \rn{eRewrite} with $\whilec{\False}{\skipc}\kateq\skipc$ from
Definition~\ref{def:kateq}, and \rn{eConseq} to eliminate $\neg\leftF{\False}$
(which is just $\True$).  By instantiating \rn{eDo} with alignment conditions
$\P,\P':=\False,\True$, the left-only and joint premises are easily proved
using \rn{eFalse} and \rn{eConseq}.  The derivation of \rn{eSkipIf} is similar,
using the equation $\ifc{\True}{\skipc}{\skipc}\kateq\skipc$ with \rn{rIf},
noting that two of the premises can then be proved using \rn{eFalse}
and \rn{eConseq}.

Rule \rn{eDoX} is derived from \rn{eDo} simply using \rn{eConseq} to allow general pre- and post-relations.  It also has a formally stronger precondition for the right-only premise,
i.e., with added conjunct that negates the left-only condition. 
That is derived by instantiating the right alignment condition $\P'$ in \rn{eDo} with $\P'\land \neg(\leftex{e}\land\P)$.

Item (i) in the following lemma could as well be presented as a proof rule but, unlike
the rules in \autoref{fig:ERHL} and \autoref{fig:ERHLcoreaddl}, it involves a
$\forall\forall$ judgment so we keep it separate.

\begin{lemma}\label{lem:allall_allexists_term}
(i) If $\models c\sep \skipc : \rspec{\Q}{\R} $ then
$\models  c\sep \skipc : \aespec{\Q}{\R}$. 
(ii) More generally, suppose $c'$ can terminate normally, i.e.,
$\all{s}{\some{t}{ c'/s\ceval t}}$.  
Then $c\sep c': \rspec{\R}{\S}$ implies $c\sep c': \aespec{\R}{\S}$.
\end{lemma}

\section{Appendix for Section~\ref{sec:bicoms}: Bicoms}\label{sec:app:bicom}

Here is a key fact about projections.  
\begin{equation}\label{eq:proj_biproj}
\Left{(\biLeft{B})} = \Left{B} \qquad         \Right{(\biLeft{B})} = \skipc \qquad          \Right{(\biRight{B})} = \Right{B} \qquad      \Left{(\biRight{B})} \kateq \skipc            \end{equation}

\begin{definition}\label{def:bieq}
Define \graybox{$\bieq$} to be the least equivalence relation on bicoms such that 
\[\begin{array}{l}
\biEmb{\skipc}{c';d'} \bieq \biEmb{\skipc}{c'};\biEmb{\skipc}{d'} \hspace*{6em}
\biEmb{c;d}{\skipc} \bieq \biEmb{c}{\skipc};\biEmb{d}{\skipc} \\
\biEmb{c}{c'} \bieq \biEmb{c}{\skipc}; \biEmb{\skipc}{c'} \hspace*{6em}
c\kateq d \mbox{ and } c'\kateq d' \mbox{ imply } \biEmb{c}{c'}\bieq\biEmb{d}{d'}
\end{array}
\]
\end{definition}
In other words, $\bieq$ is the reflexive, symmetric, transitive closure of the displayed relations.  The relation $\bieq$ implies equivalence in the semantics (Lemma~\ref{lem:bieq_bsem_equiv}).

Apropos variant annotations: For technical reasons, the right bi-projection $\biRight{\mystrut\missingArg}$ (\autoref{fig:biproject}) keeps the variant. For $\kateq$ and $\bieq$,
our Rocq mechanization leaves variants unchanged, but ignoring variants would work as well. 

\begin{figure}[t]
\begin{small}
\begin{mathpar}
\inferrule{s\models p}{\assertc{p} / s \ceval s}

\inferrule{s\not\models p}{\assertc{p} / s \ceval \fail}

\inferrule{n\in\Z}{\havc{x} / s \ceval \update{s}{x}{n}}

\inferrule{s(e)=n}{x:=e / s \ceval \update{s}{x}{n}}

\inferrule{c/s \ceval t\\ d/t \ceval \phi}{c;d/ s \ceval \phi}

\inferrule{c/s \ceval \fail}{c;d/ s \ceval \fail}

\inferrule{s\models e \\ c/s\ceval \phi}
          {\ifc{e}{c}{d} / s \ceval \phi}

\inferrule{s\not\models e \\ d/s\ceval \phi}
          {\ifc{e}{c}{d} / s \ceval \phi}

\inferrule{s \not\models e}
          {\whilec{e}{c} / s \ceval s}

\inferrule{s\models e \\ c/s\ceval \fail}
          {\whilec{e}{c} / s \ceval \fail}

\inferrule{s\models e \\ c/s\ceval t
           \\ \whilec{e}{c} / t \ceval \phi}
          {\whilec{e}{c} / s \ceval \phi}
\end{mathpar}
\end{small}
\caption{Semantics of commands}
\label{fig:commandSem}
\end{figure}

The definition of command semantics is in \autoref{fig:commandSem}.
\autoref{fig:fullbicomSem} gives the full bicom semantics.

\begin{figure}[t!]
\begin{small}
\begin{mathpar}
\mprset{sep=1.4em}
\inferrule{c / s \ceval \fail}
          {\biEmb{c}{c'} / (s,s')\beval \fail}

\inferrule{c / s \ceval t \\ c' / s'\ceval \fail}
          {\biEmb{c}{c'} / (s,s')\beval \fail}

\inferrule{c / s\ceval t \\ c' / s'\ceval t'}
          {\biEmb{c}{c'},(s,s')\beval (t,t')}

\inferrule{s,s'\models \Q}
          {\assertc{\Q} / (s,s')\beval (s,s')}

\inferrule{s,s'\not\models \Q}
          {\assertc{\Q} / (s,s')\beval \fail}

\inferrule{t'=\update{s'}{x}{n}\\ s,t'\models \Q}
          {\havRt{x}{\Q} / (s,s')\beval (s,t')}

\inferrule{B_1 / (s,s')\beval\fail}
          {B_1;B_2 / (s,s')\beval\fail}

\inferrule{s\models e \\ s'\models e' \\ B_1 / (s,s')\beval\varphi}
{\ifFour{e|e'}{B_1}{B_2}{B_3}{B_4} / (s,s') \beval \varphi}

\inferrule{B_1 / (s,s')\beval(t,t') \\ B_2 / (t,t')\beval\varphi}
          {B_1;B_2 / (s,s')\beval\varphi}

\inferrule{}{\parbox{35ex}{Also three similar bi-if rules\\
for one, other, or neither test true}}

\inferrule{s\not\models e\\ s'\not\models e'}
          {W / (s,s') \beval (s,s')}

\inferrule{s\models e \\ s,s'\models\Lrel \\ \biLeft{B} / (s,s') \beval \fail}
          {W / (s,s') \beval \fail}

\inferrule{s\models e \\ s,s'\models\Lrel \\ \biLeft{B} / (s,s') \beval (t,t')
           \\ W / (t,t') \beval \varphi}
          {W / (s,s') \beval \varphi}

\inferrule{(s\not\models e \mbox{ or } s,s'\not\models\Lrel) \\ s'\models e' 
              \\ s,s'\models\R
              \\ \biRight{B} / (s,s') \beval \fail}
          {W / (s,s') \beval \fail}

\inferrule{(s \not\models e \mbox{ or } s,s'\not\models\Lrel) 
              \\ s'\models e' \\ s,s'\models\R
              \\ \biRight{B} / (s,s') \beval (t,t')
              \\ W / (t,t')\beval \varphi}
          {W / (s,s') \beval \varphi}

\inferrule{s\models e \\ s'\models e' \\ 
           s,s'\not\models\Lrel \\ s,s'\not\models\R \\
           B / (s,s') \beval \fail}
          {W / (s,s') \beval \fail}

\inferrule{s\models e \\ s'\models e' \\
          s,s'\not\models\Lrel \\ s,s'\not\models\R \\
           B / (s,s') \beval (t,t') \\
           W / (t,t')\beval \varphi}
          {W / (s,s') \beval \varphi}

\inferrule{(s\models e \mbox{ and } s'\not\models e' \mbox{ and } s,s'\not\models\Lrel)
           \mbox{ or } 
            (s\not\models e \mbox{ and } s'\models e' \mbox{ and } s,s'\not\models\R)}
        {W / (s,s') \beval \fail}

\end{mathpar}
\end{small}
\vspace*{-2ex}
\caption{Bicom semantics. Here $W$ abbreviates $\biwhile{e|e'}{\Lrel|\R}{B}$. }
\label{fig:fullbicomSem}
\end{figure}

\begin{figure}[t!]
\begin{small}
\(\begin{array}{ll}
c               & \size(c) \\\hline
\skipc           & 0 \\
x:=e             & 1 \\ 
\havc{x}         & 1 \\
c_1; c_2          & 1 + \size(c_1) + \size(c_2)  \\
\ifc{e}{c_1}{c_2} & 1+ \size(c_1) + \size(c_2) 
\hspace*{-1ex} \\
\whilec{e}{c_1}   & 1+\size(c_1)
\end{array}\)
\quad
\(\begin{array}{ll}
B               & \size(B) \\\hline
\biEmb{c}{c'}   & 1 + \size(c_1) + \size(c_2)  \\
\assertc{\P}    & 1 \\
\havRt{x}{\P}   & 1 \\
B_1;B_2         & 1 + \size(B_1) + \size(B_2)  \\
\ifFour{e|e'}{B_1}{B_2}{B_3}{B_4} 
    & 1 + \quant{+}{i}{1\leq i \leq 4}{B_i} \\
\biwhile{e|e'}{\Lrel|\R}{B_1} & 1 + \size(B_1) 
\hspace*{-1ex} \end{array}\)
\end{small}
\caption{Size of commands and bicoms}
\label{fig:size}
\end{figure}

The semantics of bi-while involves projections of the loop body, and those
projections are not sub-terms of the bicom.  Thus some results that one might
expect to prove by structural induction on syntax, must instead go by strong
induction on a $\size$ measure.  We define $\size(\skipc)=0$, size 1 for each
other primitive command and bicom, and otherwise the size is one more than the
sum of sizes of the constituent parts (see \autoref{fig:size}).  The treatment
of $\skipc$ is what ensures the following important fact.

\begin{lemma}[size of projection]
$\size(\biLeft{B}) \leq \size(B)$ and
$\size(\biRight{B}) \leq \size(B)$.
\end{lemma}

The following two results have straightforward proofs.

\begin{lemma}\label{lem:bieq_bsem_equiv} $c\kateq d$ implies $\means{c}=\means{d}$,
and 
$B\bieq C$ implies $\means{B}=\means{C}$.
\end{lemma}

\begin{lemma}\label{lem:bsem_equiv_bcorrect} 
Suppose $\means{B}=\means{C}$.
Then $\models B:\bspec{\P}{\Q}$ iff $\models C:\bspec{\P}{\Q}$.
\end{lemma}

\begin{remark}
\normalshape
Definition~\ref{def:bieq} of $\bieq$ does not include congruence clauses.  It
may seem natural to include them but we have no need for them.
In fact congruence with respect to sequence and bi-if is no problem, but
congruence for bi-while would falsify Lemma~\ref{lem:bsem_equiv_bcorrect}.  This
is an artifact of the current definition of $\biLeft{\mystrut\missingArg}$ which
is used in the semantics of bi-while.  Because $\biLeft{\mystrut\missingArg}$
discards bi-assertions and $\havRtKeyword$, semantic equivalence is not a
congruence with respect to bi-while.  For example, let $B_0$ be
$\assertc{\False};\biEmb{x:=0}{\skipc}$ and $B_1$ be
$\assertc{\False};\biEmb{x:=1}{\skipc}$.  They are semantically equivalent,
i.e., $\means{B_0}=\means{B_1}$.  Consider their use in this context:
$\biwhile{x<0|\False}{\True|\False}{B_i}$.  From initial stores where $x$ on the
left is negative, this iterates once and sets $x$ to 0 or 1 depending on whether
$B_0$ or $B_1$ is used.  This peculiarity can probably be avoided by defining
$\biLeft{\mystrut\missingArg}$ to keep assertions like
$\biRight{\mystrut\missingArg}$ does, but the current definitions slightly streamline
the proof of our main result.
\qed
\end{remark}

\subsection{Weakest preconditions for bicoms}\label{sec:wlp}

We use the phrase ``weakest precondition'' but we are concerned with partial correctness and so use the standard name $\wlp$ that refers to weakest liberal precondition~\cite{Dijkstra76}.
For commands, define
\begin{equation}\label{eq:def:wlp}
\graybox{$\wlp(c,p)$} \eqdef \{ s \mid c/s\not\ceval\fail \mbox{ and }
                             t\models p \mbox{ for all $t$ such that } c/s\ceval t \} 
\end{equation}
This satisfies characteristic equations including 
$\wlp(\havc{x},p) = \all{x}{p}$
and $\wlp(\whilec{e}{c},p) = \gfp(F(e,c,p))$
where $F(e,c,p)(X)\eqdef (e\imp\wlp(c,X))\land(\neg e\imp p)$. 
Note that $F(e,c,p)(X)$ is monotonic in $X$ with respect to the subset  ordering on assertions, and $\gfp$ gives the greatest fixpoint.

We overload the name $\wlp$, using it for a map from bicoms and relations to relations:
\begin{equation}\label{eq:def:rwlp}
\graybox{$\wlp(B,\P)$} \eqdef \{ (s,s') \mid B/(s,s')\not\beval\fail \land
\all{t,t'}{ B/(s,s')\ceval (t,t') \imp (t,t')\models \P \} }
\end{equation}
Like wlp of commands, this has nice properties like  $\models \P\imp\wlp(B,\Q)$ iff $\, \models B:\bspec{\P}{\Q}$.
For embedded commands, $\wlp$ satisfies equations like
$\wlp(\biEmb{x:=e}{\skipc},\Q) = \subst{\Q}{x|}{e|}$.
We rely on characteristic equations for other bicom constructs,
e.g., $\wlp(\havRt{x}{\P},\Q)  =  \all{\smSep x}{(\P\imp\Q)}$.
Also $\wlp(\biwhile{e|e'}{\Lrel|\R}{B},\Q)   =  \gfp(G(e,e',\Lrel,\R,B,\Q))$ where 
\[ G(e,e',\Lrel,\R,B,\Q)(\X) \eqdef 
\begin{array}[t]{ll}
(\neg\leftF{e} \land \neg\rightF{e'} \imp \Q) \land                                    \\
(\leftF{e} \land \Lrel \imp \wlp(\biLeft{B}, \X)) \land                                \\
(\rightF{e'}\land \R \land \neg(\leftF{e}\land\Lrel) \imp \wlp(\biRight{B}, \X)) \land \\
(\leftF{e} \land \rightF{e'} \land \neg\Lrel \land \neg\R \imp \wlp(B, \X)) \land                \\
((e \agreeRel e') \lor (\leftF{e} \land \Lrel) \lor (\rightF{e'} \land \R))         
\end{array}\]
The last four conjuncts correspond with the premises in rule \rn{eDo} (in \autoref{fig:ERHL}),
and perhaps illuminate what is achieved by the evaluation semantic rules of \autoref{fig:bicomSem}.

\begin{lemma}\label{lem:rwlp_prop}
(i) $\models \P\imp\wlp(B,\Q)$ iff $\, \models B:\bspec{\P}{\Q}$
\qquad (ii) $\models B: \bspec{\wlp(B,\Q)}{\Q}$
\\
(iii) If $\, \means{B}=\means{C}$ then $\wlp(B,\Q)=\wlp(C,\Q)$
\end{lemma}

It is convenient to define the following abbreviations:
\[  \graybox{$\wlpL(c,\Q)$} \eqdef \wlp(\biEmb{c}{\skipc},\Q) 
\qquad 
    \graybox{$\wlpR(c,\Q)$} \eqdef \wlp(\biEmb{\skipc}{c},\Q) \]
We use the following characteristic equations for $\wlpR$ which are similar to those for $\wlp$ of commands.

\begin{lemma}[wlpR equations]
\label{lem:wlpR_equations}
\[\begin{array}{lcl}
\wlpR(\skipc,\Q)            & = & \Q                     \\
\wlpR(x:=e,\Q)              & = & \subst{\Q}{|x}{|e}        \\
\wlpR(\havc{x},\Q)          & = & \all{|x}{\Q}             \\
\wlpR(\assertc{p},\Q)       & = & \rightF{p}\land \Q               \\
\wlpR(c_1;c_2,\Q)           & = & \wlpR(c_1,\wlpR(c_2,p))   \\
\wlpR(\ifc{e}{c_1}{c_2},\Q) & = & (\rightF{e}\imp\wlpR(c_1,\Q))\land (\neg \rightF{e}\imp\wlpR(c_2,\Q)) \\
\wlpR(\whilec{e}{c},\Q)     & = & \gfp(F(e,c,\Q)) \\ 
\multicolumn{3}{l}{
\quad \mbox{where } F(e,c,\Q)(\X)\eqdef (\rightF{e}\imp\wlpR(c,\X))\land(\neg \rightF{e}\imp \Q)
}
\end{array}\]
\end{lemma}
Store relations form a complete lattice with respect to subset ordering.
Note that  $\Q\subseteq\R$ iff $\models \Q\imp\R$, 
and $F(e,c,\Q)(\X)$ is monotonic in $\X$ with respect to $\subseteq$.

The following equations characterize $\wlp$ on
bicoms. In the equation for $\biEmb{c}{c'}$, the order reflects the
asymmetry in the failure semantics of $\biEmb{c}{c'}$
(\autoref{fig:bicomSem}).  Readers accustomed to using such equations to define
$\wlp$, as in~\cite{DijkstraScholten}, may see apparent circularity in this
equation.  But this is not a definition, the equations are derived from big-step
semantics---see (\ref{eq:def:rwlp})---like in~\cite{AptOld3}.  

The loop case is given by
greatest fixpoint a function, $G$.  It is complicated for two reasons: there are
three kinds of bi-while loop iteration, and bi-while can fail due to its dequacy
condition (i.e., the last rule in \autoref{fig:bicomSem}).

\begin{lemma}[bicom wlp equations]
\label{lem:rwlp_equations}
\begin{small}
\[\begin{array}{lcl}
\wlp(\biEmb{c}{c'},\Q)       & = & \wlpL(c,\wlpR(c',\Q))               \\
\wlp(\assertc{\P},\Q)        & = & \P\land \Q               \\
\wlp(\havRt{x}{\P},\Q)       & = & \all{\smSep x}{(\P\imp\Q)}        \\
\wlp(B_1;B_2,\Q)             & = & \wlp(B_1,\wlp(B_2,p))   \\
\wlp(\ifFour{e|e'}{B_1}{B_2}{B_3}{B_4},\Q) & = & 
        (\leftF{e}\land \rightF{e'} \imp \wlp(B_1,\Q)) \,\land\,
        (\leftF{e}\land \neg \rightF{e'} \imp \wlp(B_2,\Q)) \,\land\,         \\
     && (\neg\leftF{e}\land \rightF{e'} \imp \wlp(B_3,\Q)) \,\land\,
        (\neg\leftF{e}\land \neg\rightF{e'} \imp \wlp(B_4,\Q))           \\
\wlp(\biwhile{e|e'}{\Lrel|\R}{B},\Q)  & = & \gfp(G(e,e',\Lrel,\R,B,\Q)) \quad\mbox{where $G$ is defined below}\\ 
\end{array}\]
\end{small}
\end{lemma}
The definition of $G(e,e',\Lrel,\R,B,\Q)$ 
involves five conjuncts to which we make later reference, so we label them (G1)--(G5).
Note that $G(e,e',\Lrel,\R,B,\Q)(\X)$ is monotonic in $\X$.
\[ G(e,e',\Lrel,\R,B,\Q)(\X) \eqdef 
\begin{array}[t]{ll}
(\neg\leftF{e} \land \neg\rightF{e'} \imp \Q) \land                                    & \mbox{(G1)}\\
(\leftF{e} \land \Lrel \imp \wlp(\biLeft{B}, \X)) \land                                & \mbox{(G2)}\\ 
(\rightF{e'}\land \R \land \neg(\leftF{e}\land\Lrel) \imp \wlp(\biRight{B}, \X)) \land & \mbox{(G3)}\\
(\leftF{e} \land \rightF{e'} \land \neg\Lrel \land \neg\R \imp \wlp(B, \X)) \land                & \mbox{(G4)}\\
((e \agreeRel e') \lor (\leftF{e} \land \Lrel) \lor (\rightF{e'} \land \R))            & \mbox{(G5)} 
\end{array}\]

\subsection{More on framing}\label{sec:app:framing}

\begin{figure}[t!]
\begin{small}
\[\begin{array}{ll}
c                    & \comFrame(c,vs)                          \\\hline

\skipc               & \True                                    \\
x:=e                 & x \in vs \land \semframe{e}{vs}          \\
\havc{x}             & x \in vs                                 \\
\assertc{p}          & \semframe{p}{vs}                         \\
c_1;c_2              & \comFrame(c_1,vs)\land\comFrame(c_2,vs)   \\ 
\ifc{e}{c_1}{c_2}    & \semframe{e}{vs}\land\comFrame(c_1,vs)\land \comFrame(c_2,vs) \\
\whilev{e}{e_1}{c_1} & \semframe{e}{vs}\land\semframe{e_1}{vs} \land \comFrame(c_1,vs)
\end{array}\]
\end{small}
\vspace*{-3ex}
\caption{Defining $\comFrame$ for commands.}
\label{fig:cFrame}
\end{figure}

\autoref{fig:cFrame} defines the $\comFrame$ function used by $\biFrame$
defined in \autoref{fig:frame}.

Semantic framing for expressions, assertions, and relations is defined as follows:
\[
\begin{array}{lcl}
\semframe{e}{vs} & \eqdef & \all{s,t}{s\agree{vs} t \imp e(s) = e(t) } \\
\semframe{p}{vs} & \eqdef & \all{s,t}{s\agree{vs} t \imp (s\models p \iff t\models p) } \\
\semframe{E}{vs} & \eqdef & \all{s,t,s',t'}{s\agree{vs} s' \land t\agree{vs} t'
  \imp E(s,s') = E(t,t') } \\
\semframe{\R}{vs} & \eqdef & \all{s,t,s',t'}{s\agree{vs} s' \land t\agree{vs} t'
  \imp (s,s'\models \R \iff t,t'\models \R) }
\end{array}\]

By contrast with the semantic property $\semframe{c}{vs}$, the condition $\comFrame(c,vs)$ used in $\biEmb{c}{c'}$ considers 
variant expressions in the code, even though variants do not influence the semantics.  
The same for $\biFrame$.
The functions also check that all assigned/havoc'd variables are in $vs$. 
Observe that $\comFrame(c,vs)$ implies $\semframe{\missingArg}{vs}$ for every expression
and assertion in $c$.

The most difficult and subtle results involving these functions
is Lemma~\ref{lem:chk_frame_sem}.
But the following is also important.
 
\begin{lemma}\label{lem:frame_of_com_sem}
$\comFrame(c,vs)$ implies $\semframe{c}{vs}$ and 
$\biFrame(B,vs)$ implies $\semframe{B}{vs}$.
\end{lemma}
The first implication is proved by induction on $c$.  The second implication is
proved by induction on $\size(B)$, because in the case $B$ is a loop the
semantics involves projection of the loop body, and projections are not in
general subterms of the bicom.  Both proofs involve reasoning about all details
in the semantics.

For no particular reason we state the following two equivalences 
in different but equivalent ways.  

\begin{lemma}\label{lem:rwlp_subst_under} 
If $\biFrame(B,vs)$ and $x\notin vs$ then, for any $n\in\Z$,
$\subst{(\wlp(B,\R))}{|x}{|n} = \wlp(B,\subst{\R}{|x}{|n})$.          \end{lemma}

\begin{lemma}\label{lem:subst_equal_ante}
If $\models \P\imp\rightF{x=e}$ 
then 
$\models (\P\imp\subst{\Q}{|x}{|e}) \iff
(\P\imp\Q)$.
\end{lemma}

The preceding two results, as well as the following three, have straightforward proofs.

\begin{lemma}\label{lem:sr_valid_metaR}
$\models\R$ iff for all $n\in\Z$, 
$\models\subst{\R}{x}{n}$.
\end{lemma}

\begin{lemma}\label{lem:substR_outside_frame}
$\semframe{\R}{vs}$ and $x\notin vs$ implies
$\subst{\R}{|x}{|e} = \R$.
\end{lemma}

\begin{lemma}[framing $\wlp$]
\label{lem:frame_sem_rwlp}
(i) If $\, \semframe{c}{vs}$ and $\semframe{\Q}{vs}$ then $\semframe{\wlpR(c,\Q)}{vs}$. \\
(ii) If $\, \semframe{B}{vs}$ and $\semframe{\Q}{vs}$ then $\semframe{\wlp(B,\Q)}{vs}$.
\end{lemma}

\section{Appendix for Section~\ref{sec:filt-adeq}}\label{sec:app:filt-adeq}

This section provides detailed proofs for the results in Section~\ref{sec:filt-adeq} including the main theorem.

The key results about syntactic framing are the following for $\uChk$ and 
Lemma~\ref{lem:chk_frame_sem} for $\chk$.

\begin{lemma}[framing $\uChk$]
\label{lem:uChk_frame_sem}
If $\, \comFrame(c,vs)$ then $\semframe{\uChk(c,vs)}{vs}$.
\end{lemma}
This holds because $\uChk$ adds assignments to fresh variables outside $vs$,
and asserts involving those variables, 
but does not alter the underlying computation except insofar as the asserts can fail.
And those asserts depend only on $vs$, via the variant which is taken into account by $\comFrame$.  
The proof is by induction on $c$ with inner inductions for the loop evaluation and fault cases, and involves intricate semantic reasoning.

Lemma~\ref{lem:chk_frame_sem}, given in the body of the paper, says 
$\, \biFrame(B,vs)$ implies $\semframe{\chk(B,vs)}{vs}$.
The proof is by induction on $B$ with inner inductions for bi-while evaluation and fault.  
The rough idea is similar to the proof of Lemma~\ref{lem:uChk_frame_sem} but the transformed bicom is more complicated (e.g., each loop gets two snapshot variables).
The $\biFrame$ condition takes into account bi-while variants,
and again the behavior only depends on the initial values of $vs$.

A key consequence is the following.
\begin{lemma}[framing $\wlp$ of $\chk$]
\label{lem:rwlp_chk_frame}
If $\,\biFrame(B,vs)$ and $\semframe{\Q}{vs}$
then $\semframe{\wlp(\chk(B,vs))}{vs}$.
\end{lemma}

\subsection{Proof of Lemma~\ref{lem:uChk_term}}

A note about the proofs to follow:
In the hints we say ``predicate calculus'' both for reasoning in the ambient
logic and for manipulating shallow embedded assertions and relations which 
obey the laws of predicate calculus.

The following is straightforward to prove.  (And it is a consequence of Theorem~\ref{thm:bcorrect_no_assume_adeq}.)
\begin{lemma}[adequacy of embed] \label{lem:biprogram_embed_correctness}
$\models \biEmb{c}{c'}:\bspec{\P}{\Q}$ implies
$\models c\sep c': \rspec{\P}{\Q}$. 
\end{lemma}

Now comes the main result about $\uChk$.

\lemuChkterm*
\begin{proof}
By structural induction on $c$, keeping $\R$ and $\S$ general. So the induction hypothesis is 
that for any subprogram $d$ of $c$, and any $\P,\Q$, 
if $\models \biEmb{\skipc}{\uChk(d,vs)}: \bspec{\P}{\Q}$
then $\models \skipc\sep d: \aespec{\P}{\Q}$.

\graybox{\textbf{case} $c$ is $\skipc$}
\[\begin{array}{lll}
    & \models \biEmb{\skipc}{\uChk(\skipc,vs)}: \bspec{\R}{\S} \\
\iff& \hint{definition of $\uChk$ (\autoref{fig:uChk})} \\
    & \models \biEmb{\skipc}{\skipc}: \bspec{\R}{\S} \\
\imp& \hint{adequacy of embed: Lemma~\ref{lem:biprogram_embed_correctness}}\\     & \models \skipc\sep\skipc: \rspec{\R}{\S} \\
\imp& \hint{Lemma~\ref{lem:allall_allexists_term}(i)}\\
    & \models \skipc\sep\skipc: \aespec{\R}{\S} 
\end{array}\]

\graybox{\textbf{case} $c$ is assignment or havoc}
The argument is the same as for $\skipc$,
because $\uChk$ leaves them unchanged and they always terminate
so we can apply Lemma~\ref{lem:allall_allexists_term}(ii).

\graybox{\textbf{case} $c$ is $\assertc{p}$}
\[\begin{array}{lll}
    & \models \biEmb{\skipc}{\uChk(\assertc{p},vs)}: \bspec{\R}{\S} \\
\iff& \hint{definition of $\uChk$}\\
    & \models \biEmb{\skipc}{\assertc{p}}: \bspec{\R}{\S} \\
\iff& \hint{wlp/correctness Lemma~\ref{lem:rwlp_prop}(i)}\\
    & \models \R\imp \wlp(\biEmb{\skipc}{\assertc{p}}, \S) \\
\iff& \hint{definition of $\wlpR$ and $\wlpR$ equation in Lemma~\ref{lem:wlpR_equations}}\\
    & \models \R\imp \rightF{p}\land\S \\
\imp& \hint{rule \rn{eConseq}, using 
      $\skipc\sep\assertc{p}: \aespec{\rightF{p}\land\S}{\S}$ from rule \rn{eSkipAssert}}\\
    & \skipc\sep\assertc{p}: \aespec{\R}{\S}
\end{array}\] 

\graybox{\textbf{case} $c$ is $\ifc{e}{c_1}{c_2}$} 
For clarity we elide $\models$ throughout, and we elide $vs$ as an argument to $\uChk$,
noting that $vs$ is the same in recursive calls to $\uChk$.
\[\begin{array}{lll}
    & \biEmb{\skipc}{\uChk(\ifc{e}{c_1}{c_2})}: \bspec{\R}{\S} \\
\iff& \hint{wlp/correctness Lemma~\ref{lem:rwlp_prop}(i), definition of $\uChk$} \\
    & \R\imp \wlp(\biEmb{\skipc}{\ifc{e}{\uChk(c_1)}{\uChk(c_2)}}, \S) \\
\iff& \hint{wlpR definition (twice) and wlpR equation for if (Lemma~\ref{lem:wlpR_equations}); predicate calculus} \\
    & \R\land\rightF{e}\imp \wlp(\biEmb{\skipc}{\uChk(c_1)}, \S) 
      \mbox{ and }
      \R\land\rightF{\neg e}\imp \wlp(\biEmb{\skipc}{\uChk(c_2)}, \S) \\
\iff& \hint{wlp/correctness Lemma~\ref{lem:rwlp_prop}(i)} \\
    & \biEmb{\skipc}{\uChk(c_1)}: \bspec{\R\land\rightF{e}}{\S}
      \mbox{ and }
      \biEmb{\skipc}{\uChk(c_2)}: \bspec{\R\land\rightF{\neg e}}{\S} \\
\imp& \hint{induction hypothesis for $c_1$ and for $c_2$, note below} \\
    & \skipc\sep c_1: \aespec{\R\land\rightF{e}}{\S}
      \mbox{ and }
      \skipc\sep c_2: \bspec{\R\land\rightF{\neg e}}{\S} \\
\imp& \hint{rule \rn{eSkipIf}} \\
    & \skipc\sep\ifc{e}{c_1}{c_2}: \aespec{\R}{\S}
\end{array}\]
To apply the induction hypothesis we need $\comFrame(c1,vs)$ and $\comFrame(c2,vs)$
which follow easily from $\comFrame(c,vs)$ which also gives
$\semframe{e}{vs}$.  That in turn gives $\semframe{\R\land\rightF{e}}{vs}$
and $\semframe{\R\land\neg\rightF{e}}{vs}$ as needed to apply the induction hypothesis.

\graybox{\textbf{case} $c$ is $c_1;c_2$} Again we elide $vs$ as an argument to $\uChk$. 
\[\begin{array}{lll}
    & \biEmb{\skipc}{\uChk(c_1;c_2)}: \bspec{\R}{\S} \\
\iff& \hint{def $\uChk$}\\
    & \biEmb{\skipc}{\uChk(c_1);\uChk(c_2)}: \bspec{\R}{\S} \\
\iff& \hint{Lemma~\ref{lem:rwlp_prop}(i), wlpR def and
           Lemma~\ref{lem:wlpR_equations};
           abbreviate $\Q:=\wlpR(\uChk(c_2),\S)$}\\
    & \R\imp\wlpR(\uChk(c_1),\Q) \\
\iff& \hint{Lemma~\ref{lem:rwlp_prop}(i) and (ii), wlpR definition} \\
    & \biEmb{\skipc}{\uChk(c_1)}:\bspec{\R}{\Q} 
      \mbox{ and }
      \biEmb{\skipc}{\uChk(c_2)}:\bspec{\Q}{\S} \\
\imp& \hint{induction hypothesis using $\semframe{\Q}{vs}$ by Lemmas~\ref{lem:chk_frame_sem} 
            and~\ref{lem:frame_sem_rwlp} }\\
    & \skipc\sep c_1:\aespec{\R}{\Q} 
      \mbox{ and }
      \skipc\sep c_2:\aespec{\Q}{\S} \\
\imp& \hint{rule \rn{eSeq}} \\
    & \skipc;\skipc\sep c_1;c_2: \aespec{\R}{\S} \\
\imp& \hint{rule \rn{eRewrite}, $\skipc;\skipc\kateq\skipc$ from Definition~\ref{def:kateq}} \\
    & \skipc\sep c_1;c_2: \aespec{\R}{\S} 
\end{array}\]

\graybox{\textbf{case} $c$ is $\whilev{e}{e_v}{d}$}
By definition, $\uChk(\whilev{e}{e_v}{d},vs)$ is 
\[ \whilev{e}{e_v}{x:=e_v;\uChk(d,vs);\assertc{(0\leq e_v<x)}} \]
where $x$ is fresh with respect to $vs$ and the assigned variables of $\uChk(d,vs)$.
Thus by hypothesis of the lemma $x$ is fresh for $\R$ and $\S$ as well as for $d$ and the variant expression $e_v$.
We aim to apply rule \rn{eSkipDo} in \autoref{fig:ERHLadditional},
using $x$ as the fresh variable and $e_v$ as the variant.  
For an invariant let  
\[ \I := \wlpR(x:=e_v;\uChk(d,vs);\assertc{0\leq e_v<x}, \S) \]
The assumption for this case is 
$\biEmb{\skipc}{\uChk(\whilev{e}{e_v}{d},vs)}: \bspec{\R}{\S}$
so by definition of $\wlpR$ and wlp/correctness lemma  we have $\models \R\imp \I$.
By the loop equation for $\wlpR$ we have $\models\I\imp F(\I)$
by the fixpoint property,\footnote{We use that $\I$ is a postfixpoint but not that it is greatest.} 
where $F$ is defined in Lemma~\ref{lem:wlpR_equations}.  
By predicate calculus this is equivalent to the following.
\[\begin{array}[t]{ll}
\models\I\land\neg\rightF{e} \imp \S                                          & \mbox{(K1)}\\
\models\I\land\rightF{e} \imp \wlpR(x:=e_v;\uChk(d,vs);\assertc{0\leq e_v<x},\I) & \mbox{(K2)}
\end{array}\]
Using rule \rn{eSkipDo} we will prove
\[ \skipc\sep\whilev{e}{e_v}{d}: \aespec{\I}{\I\land\neg\rightF{e}} \]
Then \rn{eConseq} yields the goal 
$\skipc\sep\whilev{e}{e_v}{d}: \aespec{\R}{\S}$
using (K1) and $\R\imp\I$.
It remains to show the premise of \rn{eSkipDo}, which in this instance is 
\[ \skipc\sep d: \aespec{\rightF{e}\land\I\land\rightF{x=e_v}}{\I\land\rightF{0\leq e_v < x}} \]
To prove it we calculate starting from (K2), eliding $\models$ and $vs$.
\[\begin{array}{lll}
    & \I\land\rightF{e} \imp \wlpR(x:=e_v;\uChk(d);\assertc{0\leq e_v<x},\I) \\
\iff& \hint{wlpR equation for sequence (Lemma~\ref{lem:wlpR_equations})} \\
    & \I\land\rightF{e} \imp 
        \wlpR(x:=e_v,
        \wlpR(\uChk(d),
        \wlpR(\assertc{0\leq e_v<x},\I))) \\
\iff& \hint{wlpR equations for assignment and assert (Lemma~\ref{lem:wlpR_equations})} \\
    & \I\land\rightF{e} \imp 
        \subst{(\wlpR(\uChk(d), \rightF{0\leq e_v<x}\land\I))}{|x}{|e_v} \\
\imp& \hint{predicate calculus (strengthen antecedent)} \\
    & \I\land\rightF{e}\land\rightF{x=e_v} \imp 
        \subst{(\wlpR(\uChk(d), \rightF{0\leq e_v<x}\land\I))}{|x}{|e_v}  \\
\iff& \hint{substitution Lemma~\ref{lem:subst_equal_ante}} \\
    & \I\land\rightF{e}\land\rightF{x=e_v} \imp 
        \wlpR(\uChk(d), \rightF{0\leq e_v<x}\land\I) \\
\iff& \hint{lifting Lemma~\ref{lem:sr_valid_metaR} (still eliding $\models$)} \\
    & \all{n\in\Z}{ \subst{(\I\land\rightF{e}\land\rightF{x=e_v} \imp 
               \wlpR(\uChk(d), \rightF{0\leq e_v<x}\land\I))}{x}{n} }\\
\iff& \hint{substitution properties including Lemma~\ref{lem:substR_outside_frame}} \\
    & \all{n}{ \I\land\rightF{e}\land\rightF{n=e_v} \imp 
               \subst{(\wlpR(\uChk(d), \rightF{0\leq e_v<x}\land\I))}{x}{n} }\\
\imp& \hint{subst. lemmas incl. Lemma~\ref{lem:rwlp_subst_under} 
 and frame of $\uChk(d)$, freshness of $x$} \\     & \all{n}{ \I\land\rightF{e}\land\rightF{n=e_v} \imp 
               \wlpR(\uChk(d), \rightF{0\leq e_v<n}\land\I) }\\
\iff& \hint{def $\wlpR$, wlp/correct } \\
    & \all{n}{ \biEmb{\skipc}{\uChk(d)} : 
         \Bspec{\I\land\rightF{e}\land\rightF{x=e_v}}{\rightF{0\leq e_v<x}\land\I}} \\
\imp& \hint{induction hypothesis, defs $\Left{\mystrut\missingArg}$ and $\Right{\mystrut\missingArg}$} \\
    & \all{n}{ \skipc\sep d : 
         \Aespec{\I\land\rightF{e}\land\rightF{x=e_v}}{\rightF{0\leq e_v<x}\land\I}}
\end{array}\]
Note that we use Lemma~\ref{lem:sr_valid_metaR} to lift quantification over $x$ 
to the ambient logic, because the induction hypothesis can't be directly applied where there are occurrences of 
$x$ which is outside the frame $vs$. This in turn motivated our use of the metavariable formulation
of rule \rn{eWhile}.
Apropos freshness of $x$, we use that it is outside the assigned vars of $\uChk(d)$ and that those frame $\uChk(d)$. 

Note: at the point where we appeal to the induction hypothesis, if instead of a metavariable $n$ in 
rule \rn{eDo} we used a program variable, there would be a term $\rightex{e_E}$ that is 
not framed by $vs$ (as $x_E$ is fresh), so the induction hypothesis would not be applicable.
That is why we need to apply Lemma~\ref{lem:sr_valid_metaR} at an earlier step.
\end{proof}

\subsection{Proof of Theorem~\ref{thm:main}}

\thmmain*
\begin{proof}
The proof is by induction on $\size(B)$. 
We keep $vs$ fixed but leave $\R,\S$ general, so the induction hypothesis is as follows:
\begin{quote}
For all $C,\R,\S$,
if $\size(C)<\size(B)$, $\keyw{wf}(C)$, $\biFrame(C,vs)$, $\semframe{\R}{vs}$, 
$\semframe{\S}{vs}$, and 
$\models \chk(C,vs) : \bspec{\R}{\S}$
then $\models \Left{C} \sep \Right{C} : \aespec{\R}{\S}$.
\end{quote}
For the given $B,\R,\S$ we go by cases on $B$.
In each case we calculate from assumption $\models \chk(B,vs) : \bspec{\R}{\S}$.
Every line of the calculation should begin with $\models$, as we are reasoning about valid judgments and valid implications, so for brevity we elide $\models$ throughout.  We also elide $vs$ as an argument to $\chk$ (and $\uChk$) as it is unchanged in recursive calls to those functions.

\graybox{\textbf{case} $B$ is $\havRt{x}{\Q}$}. This base case sets the main pattern used throughout the proof. The correctness judgment is put in $\wlp$ form which is then used to establish the premises of a proof rule for the projections of $B$, in this case \rn{eSkipHav} in \autoref{fig:ERHL}.
\[\begin{array}{lll}
    & \chk(\havRt{x}{\Q}) : \bspec{\R}{\S} \\
\iff & \hint{wlp/correctness Lemma~\ref{lem:rwlp_prop}(i)} \\
    & \R\imp \wlp(\chk(\havRt{x}{\Q}),\S) \\
\iff & \hint{definition of $\chk$ (\autoref{fig:chk})} \\
    & \R\imp\wlp(\assertc{\some{\smSep x}{\Q}}; \havRt{x}{\Q},\S) \\
\iff & \hint{wlp equations for seq, assert, \keyw{havf} (Lemma~\ref{lem:rwlp_equations})} \\
    & \R\imp \some{\smSep x}{\Q} \land \all{\smSep x}{(\Q\imp\S)} \\
\imp & \hint{predicate calculus} \\
    & \R\imp \some{\smSep x}{\S} \\ 
\imp & \hint{rule \rn{eSkipHav} and its soundness (Theorem~\ref{thm:soundness})} \\
    & \skipc\sep\havc{x} : \aespec{\R}{\S} \\
\iff & \hint{def $\Left{\mystrut\missingArg}$ and $\Right{\mystrut\missingArg}$} \\
    & \Left{\havRt{x}{\Q}} \sep \Right{\havRt{x}{\Q}} : \aespec{\R}{\S} 
\end{array}\]

\graybox{\textbf{case} $B$ is $\assertc{\Q}$}. The same reasoning pattern is used in this base case.

\[\begin{array}{lll}
     & \chk(\assertc{\Q}) : \bspec{\R}{\S} \\
\iff & \hint{def $\chk$, wlp/correctness Lemma~\ref{lem:rwlp_prop}(i)} \\
     & \R\imp\wlp(\assertc{\Q},\S) \\
\iff & \hint{wlp equation for assert (Lemma~\ref{lem:rwlp_equations})} \\
     & \R\imp\Q\land\S \\
\imp & \hint{rules \rn{eSkipSkip} and \rn{eConseq} using $\R\imp\S$} \\
     & \skipc\sep\skipc : \bspec{\R}{\S} \\
\iff & \hint{def $\Left{\mystrut\missingArg}$ and $\Right{\mystrut\missingArg}$} \\
     & \Left{\assertc{\Q}} \sep \Right{\assertc{\Q}} : \aespec{\R}{\S} 
\end{array}\]

\graybox{\textbf{case} $B$ is $B_1;B_2$}.
This case shows the role of framing in using the induction hypothesis.

\[\begin{array}{lll}
     & \chk(B_1;B_2) : \bspec{\R}{\S} \\
\iff & \hint{def $\chk$, wlp/correctness Lemma~\ref{lem:rwlp_prop}(i)} \\
     & \R \imp \wlp(\chk(B_1);\chk(B_2), \S) \\
\iff & \hint{wlp equation for sequence (Lemma~\ref{lem:rwlp_equations})} \\
     & \R \imp \wlp(\chk(B_1),\wlp(\chk(B_2), \S)) 
\end{array}\]
At this point it is convenient to abbreviate $\Q\eqdef\wlp(\chk(B_2,vs),\S)$.
We have $\biFrame(B_2,vs)$ from assumption $\biFrame(B_1;B_2,vs)$.
So we have $\semframe{\chk(B_2)}{vs}$ by Lemma~\ref{lem:chk_frame_sem},
whence by Lemma~\ref{lem:frame_sem_rwlp} and the assumption $\semframe{\S}{vs}$
we get $\semframe{\Q}{vs}$.
The calculation continues:
\[\begin{array}{lll}
    & \R \imp \wlp(\chk(B_1),\Q) \\
\iff & \hint{wlp/correctness Lemma~\ref{lem:rwlp_prop}(i) and 
fact that $\models \chk(B_2):\bspec{\Q}{\S}$ by Lemma~\ref{lem:rwlp_prop}(ii) } \\
    & \chk(B_1):\bspec{\R}{\Q} \quad\mbox{and}\quad \chk(B_2):\bspec{\Q}{\S} \\
\imp & \hint{induction hypothesis for $B_1$ and $B_2$, using $\semframe{\Q}{vs}$} \\
    & \Left{B_1}\sep\Right{B_1}:\aespec{\R}{\Q} \quad\mbox{and}\quad 
      \Left{B_2}\sep\Right{B_2}:\aespec{\Q}{\S} \\
\imp & \hint{sequence rule \rn{eSeq}} \\
    & \Left{B_1};\Left{B_2}\sep\Right{B_1};\Right{B_2}:\aespec{\R}{\Q}  \\
\iff & \hint{definitions of  $\Left{\mystrut\missingArg}$ and $\Right{\mystrut\missingArg}$} \\
    & \Left{B_1;B_2}\sep\Right{B_1;B_2}:\aespec{\R}{\Q} 
\end{array}\]
Appeal to the induction hypothesis is justified by $\size(B_i) < \size(B_1;B_2)$.\footnote{Here and in the case for if, induction on the structure of $B$ would suffice, but not so in the case for loops.}
Use of the induction hypothesis also requires that each $B_i$ is well-formed and $\biFrame(B_i,vs)$.
These are easy consequences of the corresponding assumptions (i) and (ii) about B.

Note that ``and'' in the intermediate steps above is at the meta level and we could have written 
$\models \Left{B_1}\sep\Right{B_1}:\bspec{\R}{\Q} \mbox{ and } 
\models \Left{B_2}\sep\Right{B_2}:\bspec{\Q}{\S}$ to be precise.

\graybox{\textbf{case} $B$ is $\ifFour{e|e'}{B_1}{B_2}{B_3}{B_4}$}.
The importance of well-formedness only emerges in this case. 
Definition~\ref{def:ifFourWF} gives us the syntactic equivalences 
$\Left{B_1} \kateq \Left{B_2}$,
$\Left{B_3} \kateq \Left{B_4}$,
$\Right{B_1} \kateq \Right{B_3}$,
$\Right{B_2} \kateq \Right{B_4}$.
By symmetry we reverse them, and list them in the order used below:
\begin{equation}\label{eq:wfIf}
\Left{B_2} \kateq \Left{B_1} \quad
\Right{B_3} \kateq \Right{B_1} \quad
\Left{B_4} \kateq \Left{B_3} \quad
\Right{B_4} \kateq \Right{B_2}
\end{equation}
In the following we elide the $vs$ argument to $\chk$, as it is the same throughout.
\[\begin{array}{lll}
    & \chk(\ifFour{e|e'}{B_1}{B_2}{B_3}{B_4}):\bspec{\R}{\S} \\
\iff & \hint{def $\chk$, wlp/correctness Lemma~\ref{lem:rwlp_prop}(i)} \\
    & \R\imp\wlp(\ifFour{e|e'}{\chk(B_1)}{\chk(B_2)}{\chk(B_3)}{\chk(B_4)},\S) \\
\iff & \hint{wlp equation for if (Lemma~\ref{lem:rwlp_equations})} \\
    & \R\imp (
         \begin{array}[t]{l}
           (\leftF{e}\land \rightF{e'}\imp\wlp(\chk(B_1),\S))\;\land
           (\leftF{e}\land \neg\rightF{e'}\imp\wlp(\chk(B_2),\S))\;\land \\
           (\neg\leftF{e}\land \rightF{e'}\imp\wlp(\chk(B_3),\S))\;\land
           (\neg\leftF{e}\land \neg\rightF{e'}\imp\wlp(\chk(B_4),\S) ) \; )\\
         \end{array} \\
\iff & \hint{pred.\ calc.\ and semantics of conjunction, leaving $\models$ implicit} \\
    & \begin{array}{l}
          \R\land \leftF{e}\land \rightF{e'}\imp\wlp(\chk(B_1),\S)     \mbox{ and }
          \R\land \leftF{e}\land \neg\rightF{e'}\imp\wlp(\chk(B_2),\S) \mbox{ and}\\
          \R\land\neg\leftF{e}\land \rightF{e'}\imp\wlp(\chk(B_3),\S)) \mbox{ and }
          \R\land\neg\leftF{e}\land \neg\rightF{e'}\imp\wlp(\chk(B_4),\S)
      \end{array} \\
\iff & \hint{wlp property Lemma~\ref{lem:rwlp_prop}(i)} \\ 
    & \begin{array}{l}
          \chk(B_1):\bspec{\R\land \leftF{e}\land \rightF{e'}}{\S}     \mbox{ and }
          \chk(B_2):\bspec{\R\land \leftF{e}\land \neg\rightF{e'}}{\S} \mbox{ and}\\
          \chk(B_3):\bspec{\R\land\neg\leftF{e}\land \rightF{e'}}{\S} \mbox{ and }
          \chk(B_4):\bspec{\R\land\neg\leftF{e}\land \neg\rightF{e'}}{\S}
      \end{array} \\
\imp & \hint{induction hypothesis four times} \\
     & \begin{array}{l}
           \Left{B_1}\sep\Right{B_1}:\aespec{\R\land \leftF{e}\land \rightF{e'}}{\S}
          \mbox{ and }
           \Left{B_2}\sep\Right{B_2}:\aespec{\R\land \leftF{e}\land \neg\rightF{e'}}{\S}
          \mbox{ and}\\
           \Left{B_3}\sep\Right{B_3}:\aespec{\R\land\neg\leftF{e}\land \rightF{e'}}{\S}
          \mbox{ and }
          \Left{B_4}\sep\Right{B_4}:\aespec{\R\land\neg\leftF{e}\land \neg\rightF{e'}}{\S} 
       \end{array}\\
\imp & \hint{rule \rn{eRewrite} four times using (\ref{eq:wfIf}) and reflexivity of $\kateq$}\\
    & \begin{array}{l}
           \Left{B_1}\sep\Right{B_1}:\aespec{\R\land \leftF{e}\land \rightF{e'}}{\S}
          \mbox{ and }
           \Left{B_1}\sep\Right{B_2}:\aespec{\R\land \leftF{e}\land \neg\rightF{e'}}{\S}
          \mbox{ and}\\
           \Left{B_3}\sep\Right{B_1}:\aespec{\R\land\neg\leftF{e}\land \rightF{e'}}{\S}
          \mbox{ and }
          \Left{B_3}\sep\Right{B_2}:\aespec{\R\land\neg\leftF{e}\land \neg\rightF{e'}}{\S}
      \end{array} \\
\imp & \hint{rule \rn{eIf4}} \\
     &\ifc{e}{\Left{B_1}}{\Left{B_3}}\sep
      \ifc{e'}{\Right{B_1}}{\Right{B_2}}:\aespec{\R}{\S} \\
\iff & \hint{def $\Left{\mystrut\missingArg}$ and $\Right{\mystrut\missingArg}$} \\
     & \Left{\ifFour{e|e'}{B_1}{B_2}{B_3}{B_4}}\sep 
       \Right{\ifFour{e|e'}{B_1}{B_2}{B_3}{B_4}}:\aespec{\R}{\S} 
\end{array}\]

\graybox{\textbf{case} $B$ is $\biEmb{c}{c'}$}.  Again the argument $vs$ to $\chk$ is the same throughout and elided, as is that same argument to $\uChk$.
First observe that 
$\biEmb{c}{\uChk(c')} \bieq \biEmb{c}{\skipc};\biEmb{\skipc}{\uChk(c')}$ by Definition~\ref{def:bieq},
so $\means{\biEmb{c}{\uChk(c')}} = \means{\biEmb{c}{\skipc};\biEmb{\skipc}{\uChk(c')}}$ by Lemma~\ref{lem:bieq_bsem_equiv}.
Hence by Lemma~\ref{lem:rwlp_prop}(iii) we have
$\wlp(\biEmb{c}{\uChk(c')},\S) = \wlp(\biEmb{c}{\skipc};\biEmb{\skipc}{\uChk(c')},)$.
Now we calculate.
\[\begin{array}{lll}
    & \chk(\biEmb{c}{c'}): \bspec{\R}{\S} \\
\iff & \hint{def $\chk$, wlp/correctness Lemma~\ref{lem:rwlp_prop}(i)} \\
    & \R\imp\wlp(\biEmb{c}{\uChk(c')},\S) \\
\iff & \hint{observation above} \\
    & \R\imp\wlp(\biEmb{c}{\skipc};\biEmb{\skipc}{\uChk(c')},\S) \\
\iff & \hint{wlp of seq, abbreviate $\Q:=\wlp(\biEmb{\skipc}{\uChk(c')},\S)$} \\
    & \R\imp\wlp(\biEmb{c}{\skipc},\Q) \\
\iff & \hint{wlp properties Lemma~\ref{lem:rwlp_prop}(i) and (ii)} \\
    & \biEmb{c}{\skipc}:\bspec{\R}{\Q}      \quad\mbox{and}\quad 
      \biEmb{\skipc}{\uChk(c')}:\bspec{\Q}{\S} \\
\imp & \hint{adequacy of embed: Lemma~\ref{lem:biprogram_embed_correctness}} \\
    & c\sep\skipc: \rspec{\R}{\Q}      \quad\mbox{and}\quad 
          \biEmb{\skipc}{\uChk(c')}:\bspec{\Q}{\S} \\
\imp & \hint{Lemma~\ref{lem:allall_allexists_term} } \\
    & c\sep\skipc: \aespec{\R}{\Q}      \quad\mbox{and}\quad 
      \biEmb{\skipc}{\uChk(c')}:\bspec{\Q}{\S} \\
\imp & \hint{Lemma~\ref{lem:uChk_term}}\\
    & c\sep\skipc: \aespec{\R}{\Q}      \quad\mbox{and}\quad 
      \skipc\sep c': \aespec{\Q}{\S} \\
\imp & \hint{sequence rule \rn{eSeq}}\\
    & c;\skipc\Sep\skipc; c': \aespec{\Q}{\S} \\
\imp & \hint{rule \rn{eRewrite} using $c;\skipc\kateq c$ and $\skipc;c'\kateq c'$}\\
    & c\sep c': \aespec{\Q}{\S} \\
\iff & \hint{defs $\Left{\mystrut\missingArg}$ and $\Right{\mystrut\missingArg}$}\\
    & \Left{\biEmb{c}{c'}}\sep\Right{\biEmb{c}{c'}}: \aespec{\Q}{\S} 
\end{array}\]

\graybox{\textbf{case} $B$ is $\biwhilev{e\sep e'}{\P\sep\P'}{E}{B_1}$}.

To introduce some nomenclature we expand the definition of $\chk(B,vs)$ 
as follows.
\[ \begin{array}{ll}
  \multicolumn{2}{l}{
  \chk(\biwhilev{e\sep e'}{\P\sep\P'}{E}{B_1},vs) = \biwhilev{e\sep e'}{\P\sep\P'}{E}{B_2}
  } \\
\quad\mbox{where } & B_2 = B_{snp};\chk(B_1,vs);B_{dec} \\
   & B_{snp} =  
       \havRt{x_E}{(\rightex{x_E}=E)};
       \havRt{x_{ro}}{(\rightex{x_{ro}} = (\rightex{e'}\land\P'))} \\ 
   & B_{dec} = \assertc{(\rightex{x_{ro}} \imp 0 \leq E < \rightex{x_{E}})}
\end{array}\]
The variables $x_{E}$ and $x_{ro}$ are fresh with respect to both $vs$ and the assigned variables in $\chk(B_1)$
(see \autoref{fig:chk}).  The names are mnemonic: 
in $B_{snp}$, variable $x_{E}$ snapshots the value of $E$ at the start of an iteration
and $x_{ro}$ snapshots the condition that it will be a right-only iteration.
The assertion $B_{dec}$ checks that on right-only iterations the variant decreases.
In the following we elide the argument $vs$ to $\chk$ as it is the same throughout.  

Let $\I \eqdef \wlp(\chk(B),\S)$, so 
$\I = \wlp(\biwhilev{e\sep e'}{\P\sep\P'}{E}{B_{snp};\chk(B_1);B_{dec}},\S)$.
Note that $\wlp(\chk(B),\S)$ is different from $\wlp(B,\S)$ owing to the instrumentation. A key fact is 
\begin{equation}\label{eq:Iframe}
\semframe{\I}{vs}
\end{equation}
This holds because we have $\semframe{\chk(B)}{vs}$ 
by Lemma~\ref{lem:chk_frame_sem} and the assumption $\biFrame(B,vs)$;
and we have assumption $\semframe{\S}{vs}$ 
so we can apply Lemma~\ref{lem:frame_sem_rwlp}.

We aim to instantiate rule \rn{eDoX}
for the commands $c:=\Left{B}$ and $c':=\Right{B}$, invariant $\I$ defined above,
variant $E$ from $B$.
This yields the desired conclusion 
$\Left{B}\sep\Right{B}: \aespec{\R}{\S}$ for the loop case,
provided we can establish the following proof obligations involving the original loop body $B_1$.  
\[
\begin{array}{ll}
\models \R \imp \I                                                                  &\mbox{(I0)}\\ 
\models \Left{B_1}\sep\Right{B_1} :
  \aespec{\I\land\leftF{e}\land\rightF{e'}\land\neg\P\land\neg\P'}{\I}              & \mbox{(I1)}\\ 
\models \Left{B_1}\sep\skipc : \aespec{\I\land\leftF{e}\land\P}{\I}                 & \mbox{(I2)}\\ 
\all{n\in\Z}{
\models \skipc\sep\Right{B_1} : \aespec{
  \I\land\rightF{e'}\land\P'\land\neg(\leftF{e}\land\P)\land(n=E)}{\I\land(0\leq E < n)}} & \mbox{(I3)}\\ 
\models \I\imp (e\agreeRel e') \lor (\leftF{e} \land P) \lor (\rightF{e'} \land P')  &\mbox{(I4)}\\
\models \I \land \neg\leftF{e} \land \neg\rightF{e'} \imp \S                         & \mbox{(I5)}
\end{array}\]
By the wlp equation for bi-while (Lemma~\ref{lem:rwlp_equations}) we have $\I
= \gfp(G(e,e',\P,\P',B_2,\S))$ where $G$ is defined in
Lemma~\ref{lem:rwlp_equations}.  We do not use that this is a greatest fixpoint,
only that it is a postfixpoint, i.e., we have $\models\I\imp
G(e,e',\P,\P',B_2,\S)(\I)$.  Expanding the definition of $G$ and applying
propositional equivalences we get the following.
\[\begin{array}[t]{ll}
\models\I\land\neg\leftF{e} \land \neg\rightF{e'} \imp \S                               & \mbox{(H1)}\\
\models\I\land\leftF{e} \land \P \imp \wlp(\biLeft{B_2}, \I)                            & \mbox{(H2)}\\ 
\models\I\land\rightF{e'}\land\P'\land\neg(\leftF{e}\land\P) \imp\wlp(\biRight{B_2},\I) & \mbox{(H3)}\\
\models\I\land\leftF{e} \land \rightF{e'} \land \neg\P \land \neg\P' \imp \wlp(B_2, \I) & \mbox{(H4)}\\
\models\I\imp((e \ddot{=} e') \lor (\leftF{e} \land \P) \lor (\rightF{e'} \land \P'))   & \mbox{(H5)} 
\end{array}\]
From our given assumption $\models\chk(B):\rspec{\R}{\S}$ we have
$\models\R\imp\wlp(\chk(B),\S)$ by Lemma~\ref{lem:rwlp_prop}(i); 
so (I0) holds by definition of $\I$.
By (H5) we have (I4) and by (H1) we have (I5).
It remains to show (I1), (I2), and (I3) which correspond to the three main premises in \rn{eDoX}.
We use (H4), (H2), and (H3) to prove (I1)--(I3).

For (I1), starting from the relevant fact (H4) we calculate, eliding $\models$
and eliding $vs$ as an argument to $\chk$.
\[\begin{array}{lll}
    & \I\land\leftF{e} \land \rightF{e'} \land \neg\P \land \neg\P' \imp \wlp(B_2, \I) \\
\iff   & \hint{def $B_2$} \\
    & \I\land\leftF{e} \land \rightF{e'} \land \neg\P \land \neg\P' \imp \wlp(B_{snp};\chk(B_1);B_{dec}, \I)  \\
\iff   & \hint{wlp equation for sequence (Lemma~\ref{lem:rwlp_equations})} \\ 
    & \I\land\leftF{e} \land \rightF{e'} \land \neg\P \land \neg\P' \imp 
        \wlp(B_{snp},\wlp(\chk(B_1),\wlp(B_{dec}, \I)))  \\
\iff   & \hint{unfold defs, wlp equation for sequence} \\ 
    & \I\land\leftF{e} \land \rightF{e'} \land \neg\P \land \neg\P' \imp 
        \begin{array}[t]{l}
           \wlp(\havRt{x_{E}}{(x_{E} = E)}, \\
           \wlp(\havRt{x_{ro}}{(x_{ro}=(\leftF{e'}\land\P'))}, \\ 
           \wlp(\chk(B_1), \\
           \wlp(\assertc{x_{ro}\imp E < x_{E}}, I))))))
        \end{array} \\
\iff & \hint{wlp equations for \keyw{havf} and \keyw{assert} (Lemma~\ref{lem:rwlp_equations}) } \\     & \I\land\leftF{e} \land \rightF{e'} \land \neg\P \land \neg\P' \imp 
        \begin{array}[t]{l}
         \all{|x_{E}}{(x_{E} = E \imp  \\
            \all{|x_{ro}}{(x_{ro}=(\leftF{e'}\land\P')) \imp \\ 
            \wlp(\chk(B_1), \I\land (x_{ro}\imp E < x_{E})))}} ))
        \end{array} \\
\imp & \hint{$\wlp(\chk(B_1),\missingArg)$ monotonic, in monotonic context} \\ 
    & \I\land\leftF{e} \land \rightF{e'} \land \neg\P \land \neg\P' \imp 
        \begin{array}[t]{l}
         \all{|x_{E}}{(x_{E} = E \imp  \\
            \all{|x_{ro}}{(x_{ro}=(\leftF{e'}\land\P')) \imp \\ 
            \wlp(\chk(B_1), \I))}} ))
        \end{array} \\
\iff & \hint{predicate calculus (one point rule), eliding $(\leftF{e'}\land\P')$ } \\ 
    & \I\land\leftF{e} \land \rightF{e'} \land \neg\P \land \neg\P' \imp 
         \all{|x_{E}}{(x_{E} = E \imp  \subst{\wlp(\chk(B_1), \I))}{x_{ro}}{\ldots} }  \\
\iff & \hint{$x_{ro}\notin vs$, Lemma~\ref{lem:substR_outside_frame}, (\ref{eq:Iframe}), 
             Lemmas~\ref{lem:chk_frame_sem} and~\ref{lem:frame_sem_rwlp} }\\
    & \I\land\leftF{e} \land \rightF{e'} \land \neg\P \land \neg\P' \imp 
         \all{|x_{E}}{(x_{E} = E \imp          \wlp(\chk(B_1), \I))} \\
\iff & \hint{one point rule, $x_E\notin vs$, framing as in preceding two steps }\\
  & \I\land\leftF{e} \land \rightF{e'} \land \neg\P \land \neg\P'
           \imp  \wlp(\chk(B_1), \I) \\ 
\iff & \hint{wlp/correctness Lemma~\ref{lem:rwlp_prop}(i)}\\
  & \chk(B_1): \bspec{\I\land\leftF{e} \land \rightF{e'} \land \neg\P \land \neg\P'}{\I}
\\
\imp & \hint{induction hypothesis, noting $\size(B_1)<\size(B)$}\\
  & \Left{B_1}\sep\Right{B_1}: \aespec{\I\land\leftF{e} \land \rightF{e'} \land \neg\P \land \neg\P'}{\I}
\end{array}
\]
So (I1) is proved.
To apply the induction hypothesis in the last step,
we need the pre- and post-relations to be framed by $vs$.  
For $\I$ this is just (\ref{eq:Iframe}).  For the precondition we also 
use that $e,e',\P,\P'$ are framed by $vs$ which follows by def from $\biFrame(B)$.

Next we prove (I2), starting from the relevant fact (H2) and eliding $vs$ as argument to $\chk$.

\[\begin{array}{lll}
    & \I\land\leftF{e} \land \P \imp \wlp(\biLeft{B_2}, \I) \\
\iff& \hint{defs $B_2$, $B_{snp}$, $B_{dec}$, $\biLeft{\mystrut\missingArg}$} \\
    & \I\land\leftF{e} \land \P \imp \wlp(
             \biLeft{\havRt{x_E}{(\rightex{x_E}=E)}} ;
             \biLeft{\havRt{x_{ro}}{\ldots}} ; 
             \biLeft{\chk(B_1)} ;
             \biLeft{\assertc{(\rightex{x_{ro}} \imp \ldots)}}, \I) \\
\iff& \hint{def $\biLeft{\mystrut\missingArg}$, abbreviate $ESS:=\biEmb{\skipc}{\skipc}$} \\
    & \I\land\leftF{e} \land \P \imp \wlp(ESS;ESS;\biLeft{\chk(B_1)};ESS, \I) \\
\iff& \hint{wlp equation for sequence, $\wlp(ESS,.)$ is identity function } \\
    & \I\land\leftF{e} \land \P \imp \wlp(\biLeft{\chk(B_1)}, \I) \\
\iff& \hint{wlp/correctness Lemma~\ref{lem:rwlp_prop}(i) } \\
    & \biLeft{\chk(B_1)}: \bspec{\I\land\leftF{e} \land \P}{\I} \\
\iff& \hint{Lemma~\ref{lem:bsem_equiv_bcorrect}, 
         using $\means{\biLeft{\chk(B_1)}}=\means{\chk(\biLeft{B_1})}$ 
         from Lemmas~\ref{lem:chk_commute} and~\ref{lem:bieq_bsem_equiv} }\\ 
    & \chk(\biLeft{B_1}): \bspec{\I\land\leftF{e} \land \P}{\I} \\
\imp& \hint{induction hypothesis, using fact (\ref{eq:Iframe}) 
       and $\size(\biLeft{B_1}) \leq  \size(B_1) < \size(B)$ }\\ 
    & \Left{\biLeft{B_1}} \sep \Right{\biLeft{B_1}}: \aespec{\I\land\leftF{e} \land \P}{\I} \\
\iff& \hint{$\Left{\biLeft{B_1}}=\Left{B_1}$ and $\Right{\biLeft{B_1}}=\skipc$ by (\ref{eq:proj_biproj}) } \\
    & \Left{B_1} \sep \skipc : \aespec{\I\land\leftF{e} \land \P}{\I} \\
\end{array}\]
So (I2) is proved.

Finally we prove (I3).  
We use the fact that for any $x$ and $\Q$, $\models\all{|x}{\Q}$ iff $\models \Q$. We start from the relevant property (H3).
\[\begin{array}{lll}
    & \I\land\rightF{e'}\land\P'\land\neg(\leftF{e}\land\P) \imp\wlp(\biRight{B_2},\I) \\
\iff& \hint{defs $B_2$, $B_{snp}$, $B_{dec}$, $\biRight{\mystrut\missingArg}$; wlp over seq;
            abbrev.\ $XRO:=\rightF{e'}\land\P'\land\neg(\leftF{e}\land\P)$ } \\
    & \I\land XRO \imp
         \begin{array}[t]{l}
            \wlp(\havRt{x_E}{(\rightex{x_E}=E)}, 
            \wlp(\havRt{x_{ro}}{(\rightex{x_{ro}}=(\rightF{e'}\land\P'))}, \\
            \wlp(\biRight{\chk(B_1)},        
            \wlp(\assertc{(\rightex{x_{ro}} \imp 0 \leq E < \rightex{x_E})}, \I)))) 
         \end{array} \\
\iff& \hint{wlp equations for \assertc\ and \havRtKeyword } \\
    & \I\land XRO \imp
         \begin{array}[t]{l}
            \all{|x_E}{(\rightex{x_E}=E) \imp    
            \all{|x_{ro}}{(\rightex{x_{ro}}=(\rightF{e'}\land\P')) \imp \\
            \wlp(\biRight{\chk(B_1)}, (\rightex{x_{ro}} \imp 0 \leq E < \rightex{x_E})\land \I)}}
         \end{array} \\
\iff& \hint{predicate calculus, $x_E$ and $x_{ro}$ outside frames of $\I$ and $XRO$} \\
    &   \begin{array}[t]{l}
            \all{|x_E}{\all{|x_{ro}}{
                \I\land XRO \land (\rightex{x_E}=E) \land 
                (\rightex{x_{ro}}=(\rightF{e'}\land\P')) \imp \\
                \wlp(\biRight{\chk(B_1)}, (\rightex{x_{ro}} \imp 0 \leq E < \rightex{x_E})\land \I)}}
         \end{array} \\
\iff& \hint{$\models XRO\imp (\rightF{e'}\land\P')$, and 
            if $\models X\imp Z$ then $\models X\land(y=Z)\iff X\land(y=\True)$ } \\
    &   \begin{array}[t]{l}
            \all{|x_E}{\all{|x_{ro}}{
                \I\land XRO \land (\rightex{x_E}=E) \land 
                (\rightex{x_{ro}}=\True) \imp \\
                \wlp(\biRight{\chk(B_1)}, (\rightex{x_{ro}} \imp 0 \leq E < \rightex{x_E})\land \I)}}
         \end{array} \\
\iff& \hint{fact about $\models$ and $\forall$  mentioned at the start (noting $\models$ is elided here)} \\
    &   \begin{array}[t]{l}
            \all{|x_{ro}}{
                \I\land XRO \land (\rightex{x_E}=E) \land 
                (\rightex{x_{ro}}=\True) \imp \\
                \wlp(\biRight{\chk(B_1)}, (\rightex{x_{ro}} \imp 0 \leq E < \rightex{x_E})\land \I)}
         \end{array} \\
\iff& \hint{predicate calculus, $x_{ro}$ outside frames of $\I$, $XRO$, and $\rightex{x_E}=E$} \\
    &   \begin{array}[t]{l}
                \I\land XRO \land (\rightex{x_E}=E) \imp 
                \all{|x_{ro}}{(\rightex{x_{ro}}=\True) \imp \\
                \wlp(\biRight{\chk(B_1)}, (\rightex{x_{ro}} \imp 0 \leq E < \rightex{x_E})\land \I)}
         \end{array} \\
\iff& \hint{predicate calculus (one point rule)} \\
    &   \I\land XRO \land (\rightex{x_E}=E) \imp 
                \subst{(\wlp(\biRight{\chk(B_1)}, 
                       (\rightex{x_{ro}} \imp 0 \leq E < \rightex{x_E})\land \I))}{|x_{ro}}{|\True} \\
\iff& \hint{Lemma~\ref{lem:rwlp_subst_under} using that $x_{ro}$ is outside frame of $\chk(B_1)$ }  \\ 
    &   \I\land XRO \land (\rightex{x_E}=E) \imp 
                \wlp(\biRight{\chk(B_1)}, 
                    \subst{((\rightex{x_{ro}} \imp 0 \leq E < \rightex{x_E})\land \I)}{|x_{ro}}{|\True}) \\
\iff& \hint{subst over conj, Lemma~\ref{lem:substR_outside_frame} for $\I,E$; $x_{ro}$ fresh } \\
    & \I\land XRO \land (\rightex{x_E}=E) \imp 
         \wlp(\biRight{\chk(B_1)}, (\rightex{\True} \imp 0 \leq E < \rightex{x_E})\land \I) \\
\iff& \hint{Lemma~\ref{lem:sr_valid_metaR} for $n\in\Z$, eliding $\models$ } \\
    & \all{n}{ \subst{(\I\land XRO \land (\rightex{x_E}=E) \imp 
         \wlp(\biRight{\chk(B_1)}, (\rightex{\True} \imp 0 \leq E < \rightex{x_E})\land \I))}{|x_E}{|n}} \\

\iff& \hint{framing: $x_E$ fresh for all except $\rightex{e_E}$, Lemmas~\ref{lem:substR_outside_frame} \ref{lem:frame_sem_rwlp} ,
           $\subst{\leftex{x_E}}{|x_E}{n} = n$ } \\     & \all{n}{ (\I\land XRO \land (n=E) \imp 
         \wlp(\biRight{\chk(B_1)}, (\rightex{\True} \imp 0 \leq E < n)\land \I)) } \\

\iff& \hint{wlp/correctness Lemma~\ref{lem:rwlp_prop}(i), simplify $\rightF{\True}\imp\ldots$} \\ 
    & \all{n}{ \biRight{\chk(B_1)}: 
          \Bspec{\I\land XRO \land (n=E)}{0 \leq E < n \land \I}} \\
\iff& \hint{Lemma~\ref{lem:chk_commute}} \\ 
    & \all{n}{ \chk(\biRight{B_1}): 
          \Bspec{\I\land XRO \land (n=E)}{0 \leq E < n \land \I}} \\
\imp& \hint{induction hypothesis, note below} \\ 
    & \all{n}{ \Left{\biRight{B_1}}\sep\Right{\biRight{B_1}}:
          \Aespec{\I\land XRO \land (n=E) }{ 0 \leq E < n \land \I}} \\
\imp& \hint{rule \rn{eRewrite} using $\Left{\biRight{B_1}} \kateq \skipc$ and $\Right{\biRight{B_1}} = \Right{B_1}$
               from (\ref{eq:proj_biproj}), and $\kateq$ reflexive } \\
    & \all{n}{ \Left{B_1}\sep\Right{B_1}:
          \Aespec{\I\land \rightF{e'}\land\P'\land\neg(\leftF{e}\land\P) \land (n=E) 
                  }{ 0 \leq E < n \land \I}}
\end{array}\]

So (I3) is proved, which completes the proof of the Theorem.
Note: the step using induction hypothesis relies on $\size(\biRight{B_1})\leq\size(B_1)<\size(B)$.
Also the pre- and post-relations are semantically framed by $vs$: for $\I$ this is 
(\ref{eq:Iframe}); for $\P,\P',E$ this is from $\biFrame(B)$.
\end{proof}

\section{Appendix for Section~\ref{sec:unary}}\label{sec:app:unary}

Here are details on which the definition of $\transB{\missingArg}$ is based.

We assume given two injective renaming functions on variables, written $\vlt{x}$
and $\vrt{x}$, with disjoint ranges.  We will use $\vlt{x}$ to represent the
value of $x$ in the left state (resp.\ $\vrt{x}$ on the right), in order to
define the function $\enc: \Store\times\Store\to\Store$ that encodes a pair of
stores as a single store by renaming variables.\footnote{In our machine checked
development we consider an enumeration of variables and define
$\vlt{(\mbox{$x_i$})} = x_{2i}$ and $\vrt{(\mbox{$x_i$})}= x_{2i+1}$.  In the
prototype of Section~\ref{sec:eval} variable names are strings and the renamings
just add distinct prefixes.}  The function $\dec: \Store\to\Store\times\Store$
interprets a store as an encoded pair and is a two-sided inverse of $\enc$.
These lift to bijective encoding of store relations as store assertions, for
which we overload the names $\enc,\dec$.  The key fact we use is $\dec(\enc\ \R)
= \R$.

\begin{wrapfigure}{R}{0.52\textwidth}  \(\begin{array}{ll}
c               & \transR{c} \\\hline
\skipc          & \skipc \\
x:=e            & \vrt{x} := \transR{e} \\
\havc{x}        & \havc{\vrt{x}} \\
\assertc{p}     & \assertc{\transR{p}} \\
\assumec{p}     & \assumec{\transR{p}} \\
c1; c2          & \transR{c1};\transR{c2} \\
\ifc{b}{c1}{c2} & \ifc{\transR{b}}{\transR{c1}}{\transR{c2}} \\
\whilec{b}{c}   & \whilec{\transR{b}}{\transR{c}} 
\end{array}\)
\caption{Definition of right-side encoding \graybox{$\transR{c}$}.
The definition of \graybox{$\transL{c}$} is symmetric.}
\label{fig:transR}
\end{wrapfigure}

For unary expressions and assertions we define encodings $\transL{\missingArg}$ and $\transR{\missingArg}$ on the left and right, such that for example 
$\dec (\transL{p}) = \leftF{p}$
for any assertion $p$.
These are used to define the encoding of a command $c$ 
as a command $\transL{c}$ acting on the left side of an encoded store pair, leaving the right side unchanged; similarly $\transR{c}$ has $c$ act on the right (see Figure~\ref{fig:transR}).
These encoded commands are used for the embed case in the definition of $\transB{B}$, 
for which see \autoref{fig:transB}.
Note the use of \keyw{assume} for the case of havoc-filter.
In the case of bi-while, the store assertion $q$ is the encoding of the adequacy condition for aligned loops; 
the reader may prefer to write it as a boolean combination of $\transL{e}$, $enc(\Lrel)$, etc.
In the encoding of bi-while, the alignment conditions ($\Lrel,\R$ in \autoref{fig:transB}) are store relations yet they are used as \keyw{if} tests in the translation.  
In the theory we simply allow if-commands to have (encoded) assertions as tests, 
but in our prototype we restrict $\Lrel,\R$ to be boolean expressions.  

The proof of Theorem~\ref{thm:bicomToCom_reflect_correct} is based on the following.  
In its statement, recall that $\varphi$ ranges over bicom outcomes,
and extend $\enc$ by $\enc(\fail) \eqdef \fail$.
\begin{lemma}\label{lem:bicomToCom} If $B$ is well-formed and 
$B/(s,s')\beval\varphi$ then $\transB{B}/(\enc(s,s'))\ceval(\enc(\varphi))$.
\end{lemma}
For the case $B/(s,s')\beval(t,t')$, the proof is by induction on this judgment,
and for the fault case $B/(s,s')\beval\fail$ the proof is by induction on the 
judgment but also using the evaluation case.
Although $\transB{B}$ is defined by recursion on $\size(B)$ the proofs do not need size induction.

The proof of Theorem~\ref{thm:bcorrect_no_assume_adeq} relies on the following
lemmas.
Note that the Theorem holds even if the bicom has unary \keyw{assume}\ commands,
for what that is worth.

\begin{lemma}\label{lem:eval_eval_beval}
Suppose $B$ is well-formed and contains no $\havRtKeyword$.
If $\Left{B}/s\ceval t$ and $\Right{B}/s'\ceval t$'
then either $B/(s,s')\beval(t,t')$ or $B/(s,s')\beval\fail$.
\end{lemma}

\begin{lemma}\label{lem:fault_eval_bfault} Suppose $B$ is well-formed and contains no $\havRtKeyword$.
If $\Left{B}/s\ceval \varphi$ and $\Right{B}/s'\ceval \varphi'$ 
and at least one of $\varphi$ and $\varphi'$ is $\fail$ 
then $B/(s,s')\beval\fail$.
\end{lemma}
 
The proofs of Lemmas~\ref{lem:eval_eval_beval} and~\ref{lem:fault_eval_bfault}
go by induction on $B$, with an inner induction for the case of bi-while.  The inner induction is a strong induction on the total number of iterations of the two projections.  
This is needed because depending on state  the bi-while semantics consumes a left, or right, or both iteration.  

To prove Theorem~\ref{thm:bcorrect_no_assume_adeq}, 
unfolding the definition of 
$\Left{B}\sep\Right{B}:\rspec{\P}{\Q}$
exposes obligations that can be discharged by 
applying the hypothesis $B:\bspec{\P}{\Q}$ together 
with the preceding lemmas.

The proof of Theorem~\ref{thm:bicomToCom_preserve_correct}
is based on lemmas roughly similar to converses of the preceding lemmas, again proved by lengthy semantic arguments.

\section{Appendix for Section~\ref{sec:eval}}\label{sec:app:eval}

The evaluation was done using Why3 version 1.7.2, Alt-Ergo version 2.6.2, and Z3 version 4.15.2.  Other solver versions work but the prototype has not been brought up to date with API changes in Why3 1.8.X.  

Apropos our prototype, the translation of control structures to WhyML is essentially that of \autoref{fig:transB} with the restriction that bi-while alignment guards should be expressible as boolean expressions, so no quantifiers.
In our experience this was no limitation.
If quantifiers are needed the translation can be tweaked to havoc a boolean and branch on that.
The source language also has many features not present in
the simple language we formalize in Section~\ref{sec:syntax},
including classes, object allocation, field read/write, procedures, interfaces and modules.
Although these are not encompassed by Theorem~\ref{thm:main}, many present no difficulty, including aligned procedure calls.
For prudence, however, 
the prototype checks and emits warnings about some features that should not be used in $\forall\exists$ mode:
(a) Right-side unary procedure calls. 
(b) Encapsulated private invariants, coupling invariants, and read effects used to enforce encapsulation.  

For the $\chk$ transformation the prototype implements a slight variation of the definition in \autoref{fig:chk}, for the case of bi-while, specifically in case the user provided alignment conditions are not mutually exclusive.  This is explained in Remark~\ref{rem:chkNondet} below.
By contrast to the definitions of $\chk$ and $\uChk$ in Figures~\ref{fig:uChk} and~\ref{fig:chk}, the implementation, which is in OCaml, does not pass around the ``avoid list'' $vs$.  Instead, it generates fresh names for the snapshot variables using a global counter.

The ORHLE repo includes invalid examples which we used as sanity checks 
and include in the \WhyRelE\ repo; they are not included in Table~\ref{tab:Catalog}.
For Denning1, the ORHLE version had a divide-by-zero error which we fixed.
That was necessary because our $\forall\exists$ property (Definition~\ref{def:relcorrect}) includes absence of errors,
which is not checked by ORHLE.

This paper focuses on the specific error-avoiding $\forall\exists$ correctness judgment of Definition~\ref{def:relcorrect}.
However, by explicitly encoding runtime failure as an observable variable 
we can express other failure conditions like failure-sensitive noninterference 
as we show using an example from Hypra (called lowErr in Table~\ref{tab:Catalog}).

For the Veracity benchmarks we also verified fully sequential alignments, 
i.e., $\biEmb{c;d}{\skipc};\biEmb{\skipc}{d;c}$, using similar loop invariants.
This required slightly more annotations and solver time.
In the absence of helpful alignments one can always use sequential and prove relations based on precise descriptions of the final states---but this should be the last resort.

Apropos the Time column in Table~\ref{tab:Catalog}.
For simplicity a single solver is used for all the VCs of a given benchmark, either Alt-Ergo~\cite{AltErgo} or Z3~\cite{Z3}. 
Using Why3 IDE one sometimes finds different solvers do better on different VCs, including CVC5.
Also we see little variation between runs so the timings are from
a single run rather than being averaged over repeated runs of our toolchain.

For the PCSat tool~\cite{UnnoTerauchiKoskinen21} we only include annotation
counts for the ForEx versions of its benchmarks, although PCSat does allow
user-provided invariant hints.  The reason is that it requires manual
translation of the code and spec into constraint clauses that represent the
transition system and property.  Of course this could be automated so
comparison with that user effort would not make sense.

\begin{remark}\label{rem:chkNondet}
\normalshape
Here is a fine point about the definition of $\chk$ in \autoref{fig:chk}.
The snapshot variable $x_2$ gets set to the value of 
the condition $\rightF{e'}\land\P'$ under which a right-only iteration is enabled.  The reader should note the similarity with the right-only precondition in rule \rn{eDo} (\autoref{fig:ERHL}).  

The fine point arises when $\P'$ does not exclude that the left-only condition $\leftF{e}\land\P$ also holds.  
In that case, the bicom semantics deterministically chooses left-only rather than right-only (see rule \rn{Right-only iter} in \autoref{fig:bicomSem}). This is reflected in condition (G3) of the $\wlp$ for bi-while (Lemma~\ref{lem:rwlp_equations}, also mentioned in the body of the paper.) 
One might expect that $\chk$ should therefore set snapshot variable $x_2$ 
to the condition $\rightF{e'}\land\P' \land \neg(\leftF{e}\land\P)$.
Unfortunately that definition invalidates the very useful 
Lemma~\ref{lem:chk_commute}. 

For any $\P,\P'$ that can be used with \rn{eDo} to prove a judgment, the proof can be adapted to replace $\P'$ by
the disjoint condition $\P'\land\neg(\leftF{e}\land\P)$.
If $\P$ and $\P'$ are not mutually exclusive then the precondition for right-only gets strengthened, making the proof obligation easier.

However, if one chooses overlapping $\P$ and $\P'$ then the filter-adequacy technique can fail owing to the definition of $\chk$ and in particular 
the assertion for variant decrease:
$(\rightex{x_2} \imp 0 \leq E < \rightex{x_1})$; this requires proving decrease on executions where both right-only and left-only are enabled.
We have manually written examples in Boogie, following the definition of $\chk$ in \autoref{fig:chk} and the translation in \autoref{fig:transB}, 
that fail to verify because a right-only variant does not decrease on left-only paths through a loop.  
One example is a variation on the sample code in Section~\ref{sec:app:boogie}.

The fix, of course, is to strengthen the right alignment condition.  
As part of it's transformation and translation our prototype replaces the user-provided right alignment condition with the disjoint one.  
\qed\end{remark}

\subsection{Sample boogie code}\label{sec:app:boogie}

The following code illustrates three things: the encoding of bicoms as unary programs (\autoref{fig:transB}), the filter-adequacy transformation (\autoref{fig:chk}), and the issue in Remark~\ref{rem:chkNondet}.

The user-provided product \lstinline{hiccupSum_ref} explicitly encodes two executions and their alignment.
Procedure \lstinline{hiccupSum_ref_chk} is the result of applying the transformation.

\begin{lstlisting}
/* Simple example of data-dependent alignment. This can be verified using 
   lockstep alignment and havoc-filter.  Instead we use data-dependent 
   alignment for the sake of a very simple example of that.
*/

function bool2int(x:bool):int { if x then 1 else 0 } 

procedure hiccupSum (n: int) returns (r: int)
{ var h: bool; var i: int;  
  i, r := 0, 0;
  havoc h;
  while (i < n) { 
    if (!h) { r := r+i; i := i+1; }
    havoc h;
  }   
} 

// the user-provided product to show hiccupSum refines itself 
procedure hiccupSum_ref (n, n': int) returns (r, r': int)
  requires n == n';   ensures r == r';                     
{ 
  var h, h': bool; var i, i': int; 
  i, r := 0, 0;   i', r' := 0, 0;   // lockstep alignment
  havoc h;       havoc h';      

  while (i < n || i' < n') // left alignment condition: h
                           // right alignment condition: !h && h'
                           // variant: bool2int(h') 
    invariant i == i';  
    invariant r == r';
    invariant (i < n && i' < n') || (i >= n && i' >= n') // adequacy
            || (i < n && h) || (i' < n' && (!h && h'));
    { 
      if (i < n && h) { // left only (condition h)
          if (!h) { r := r+i; i := i+1; } 
          havoc h;
      } else if (i' < n' && !h && h') { // right only (condition !h && h'
          if (!h') { r' := r'+i'; i' := i'+1; } 
          havoc h'; assume !h'; // havoc-filter
      } else { 
          assert (i < n && i' < n' && !h && !(!h && h')); // just a remark; from adequacy
          if (!h) { r := r+i; i := i+1; }    // sequential alignment 
          havoc h;
          if (!h') { r' := r'+i'; i' := i'+1; } 
          havoc h'; // no need for filter in this example 
      } 
    }
}

// transformed from the user-provided product 
procedure hiccupSum_ref_chk (n, n': int) returns (r, r': int)
  requires n == n';   ensures r == r';                     
{ 
  var h, h': bool; var i, i': int;   var vsnap: int; var rosnap: bool;
  i, r := 0, 0;   i', r' := 0, 0;
  havoc h; havoc h';

  while (i < n || i' < n') 
    invariant i == i';  
    invariant r == r';
    invariant (i < n && i' < n') || (i >= n && i' >= n')
            || (i < n && h) || (i' < n' && (!h && h'));
    { 
      vsnap := bool2int(h');             // added by chk
      rosnap := i' < n' && !h && h';    // added by chk 
      if (i < n && h) {                
          if (!h) { r := r+i; i := i+1; } 
          havoc h;
      } else if (i' < n' && !h && h') {  
          if (!h') { r' := r'+i'; i' := i'+1; } 
          assert (exists v:bool :: v == !h');  // added by chk
          havoc h'; assume !h';        
      } else { 
          if (!h) { r := r+i; i := i+1; } 
          havoc h;
          if (!h') { r' := r'+i'; i' := i'+1; } 
          havoc h'; 
      } 
      assert rosnap ==> 0 <= bool2int(h');     // added by chk
      assert rosnap ==> bool2int(h') < vsnap; // added by chk
    }
}
         
\end{lstlisting}

Manually writing the Boogie code required us to keep track of (a) the encoding of left versus right state (which we did by a variable name convention),
(b) the intended alignment and its embodiment using ordinary control structure---in other words, patterns like the right sides in \autoref{fig:transB}; 
(c) the snapshots and assertions added by $\chk$;
(d) specs, invariants, and helper assertions for the specific problem.
Of course some of these things can be automated.
Indeed, by providing a bicom notation, \WhyRelE\  enables the user to focus on exploring alignments, writing loop invariants and specs, and occasionally adding intermediate assertions to aid SMT solvers. For example it can be helpful to rule out infeasible paths, which is the effect of asserting false in $B_1$ and $B_3$ in \autoref{fig:unno}.

Apropos the phenomenon discussed in Remark~\ref{rem:chkNondet}:
If one replaces the right-only alignment condition 
by simply $h'$, which isn't mutually exclusive with respect to
the left-only condition $h$, then the variant-decrease assertion 
inserted by chk fails because it does not hold on left-only paths.

\end{document}